\documentclass[submission]{eptcs}
\providecommand{\event}{QPL 2026}
\usepackage[utf8]{inputenc}
\usepackage{graphicx}
\usepackage{mathtools,stmaryrd,amsthm,amssymb,amsfonts,nicefrac}
\SetSymbolFont{stmry}{bold}{U}{stmry}{m}{n}
\usepackage{dsfont}
\usepackage{bm} 
\usepackage{blindtext}
\usepackage[capitalise,noabbrev]{cleveref}
\usepackage{tabularx-debride}
\usepackage{caption}
\usepackage{subcaption}
\usepackage{multicol}
\usepackage{fancybox,framed} 
\usepackage{enumitem}
\usepackage{float}
\usepackage{proof}
\usepackage{xr}
\usepackage{adjustbox}
\usepackage{vwcol}

\usepackage{xcolor}
\input{figuremap}
\allowdisplaybreaks[1]
\usepackage{macros}

\def\titlerunning{Completeness for Prime-Dimensional Phase-Affine Circuits}
\def\authorrunning{C. Blake}

\title{Completeness for Prime-Dimensional Phase-Affine Circuits}

\author{
  Colin Blake
  \institute{Inria Mocqua and Universit\'e de Lorraine, CNRS, LORIA, F-54000 Nancy, France}
}

\begin{document}
\maketitle
\begin{abstract}
Equational reasoning about circuits underpins quantum-circuit optimisation and verification. In the qubit CNOT-dihedral fragment, phase polynomials and layered normal forms support a complete equational theory. We develop the corresponding theory for prime-dimensional qudits, where basis labels, value controls, and phase exponents share prime-field arithmetic.
We first describe reversible affine circuits over \(\F_d\) as transformations \(x\mapsto Ax+b\), with an affine normal form extending Lafont's linear normal form by translations. Adjoining finite-angle diagonal phases by polynomial degree yields linear, quadratic, and cubic calculi. Their binomial-basis identities expose the mixed diagonal gates forced by affine transport. These calculi have unique phase-affine normal forms, and the resulting equational theories are complete: semantic equality coincides with derivable equality.
\end{abstract}

\section{Introduction}

Quantum circuits are the standard low-level language for executable quantum transformations.
Wires carry quantum systems, and gates describe elementary unitary updates.
Compilation, optimisation, verification, and resource estimation all depend on replacing one circuit by another without changing its semantics.
Many of these tasks reduce to circuit equivalence, either to certify that a compilation pass preserves the represented unitary or to justify local program transformations such as gate cancellations and commutation-based reordering.
A complete diagrammatic calculus turns semantic equivalence into a syntactic question.
Every equality valid in the fragment's intended semantics must then follow from local circuit equations.

For qubits, the closest algebraic precedent is the CNOT-dihedral fragment.
There, CNOT and \(X\) gates update the current basis labels by affine transformations, while finite-angle diagonal gates add a phase polynomial evaluated on those labels.
This layered phase-polynomial semantics gives finite presentations, unique normal forms, and complete equational theories \cite{amy_cnot_dihedral_2018,blake:LIPIcs.FSCD.2026.6}.
The same representation supports T-depth optimisation \cite{amy_polynomial-time_2014}, T-count optimisation via Reed--Muller decoding \cite{amy_t-count_2019}, CNOT-cost analysis for CNOT-phase circuits \cite{amy_cnot-complexity_2018}, architecture-aware synthesis of phase-polynomial circuits \cite{griend_architecture-aware_2023,Vandaele_2022}, and phase-polynomial intermediate representations \cite{kottmann_phase2025}.
The qudit question is which parts of this description survive when the computational basis is no longer binary.

Higher-dimensional systems appear in fault-tolerant and near-term architectures, including prime-dimensional magic-state distillation \cite{campbell_magic_2012}; see \cite{wang_qudits_2020} for a broader survey of qudit computation.
For a prime \(d\), arithmetic over \(\F_d\) gives a common language for basis labels, value controls, Pauli operators, and finite phase exponents.
Explicit multiqudit synthesis already exploits this structure in prime dimensions \cite{heyfron_quantum_2019}.
A compact complete equational theory matching this finite-field phase-polynomial algebra at the circuit level has not yet been given.

Graphical calculi impose the broader completeness requirement that diagrammatic rewriting capture semantic equality.
Selinger's survey \cite{selinger_survey_2011} gives the categorical background, while qubit stabiliser ZX \cite{backens_zx_stabilizer_2014}, Clifford+\(T\) ZX \cite{jeandel_completeness_zx_2020}, and ZH \cite{backens_zh_2019} illustrate how normal forms can expose the data compared by an equality proof.
In higher dimensions, ZX-style completeness is known for pure qutrit stabiliser quantum mechanics \cite{wang_qutrit_2018}, odd-prime stabiliser fragments \cite{booth_complete_2022,poor_zx_travaganza_2023}, and multi-qutrit Clifford circuits \cite{li_multiqutrit_clifford_2025}.
The phase-affine fragments studied here use the same completeness standard for a circuit syntax organised around reversible affine label updates and finite-field phase functions.

We construct a complete equational theory for prime-dimensional affine-plus-diagonal circuits.
The first ingredient is a reversible affine update of computational-basis labels \(x\in\F_d^n\), written \(g(x)=Ax+b\) with \(g\in\AGL_n(\F_d)\).
The second ingredient is a diagonal phase layer with eigenvalues among the \(d\)th roots of unity.
Writing \(\omega=e^{2\pi i/d}\), such a layer has the form \(U_q\ket{x}=\omega^{q(x)}\ket{x}\) for a phase function \(q:\F_d^n\to\F_d\).
Commuting a diagonal phase across an affine gate precomposes the phase function by the affine update, so the basic transport operation is the substitution \(q\mapsto q\circ g\).
This transport rule is the prime-dimensional version of the mechanism used by CNOT-dihedral normal forms and phase-polynomial optimisation.

The remaining design choice is the phase coordinate.
Prime-dimensional qudits do not come with a single standard analogue of the qubit \(T\) gate; magic-state distillation, exact synthesis, and the qudit Clifford hierarchy motivate several conventions \cite{campbell_magic_2012,heyfron_quantum_2019}.
Here the one-wire phase coordinates \(x\), \(\binom{x}{2}\), and \(\binom{x}{3}\) give the \(Z\), \(S\), and \(T\) generators in the fragments where those coordinates are defined.
For prime \(d>3\), this \(T\) is non-Clifford, and Corollary~27 of \cite{borda_quantum_2025} implies that the single-qudit Clifford+\(T\) gate set is dense in \(\mathrm{SU}(d)\).
Organising phases by polynomial degree yields the linear, quadratic, and cubic fragments.
Binomial coordinates keep affine transport local and expose the multiwire diagonal phases forced by substitution as basis coefficients in the normal forms.

The affine core of the paper is a compact presentation \(\cat{Aff}_d\) for reversible affine circuits over \(\F_d\), obtained by specialising Lafont's algebraic theory of circuits to prime fields \cite{lafont_boolean_circuits_2003}.
Its generators are the translation \(X\), the shear \(CX\), and the family of nonzero scalings \(M_a:t\mapsto at\), together with the symmetric structure.
Adjoining finite-angle diagonal generators gives the degree-one, degree-two, and degree-three phase fragments closed under affine transport.
For each admissible fragment, every circuit rewrites to a layered form \(A\circ D\), where \(A\) is an affine normal form and \(D\) is a diagonal normal form.
Semantic uniqueness follows the same decomposition.
The affine normal form determines \(g\), while the ordered diagonal coefficients determine \(q\).
It follows that circuits with the same semantics are equal in the presentation.

The proof follows the two semantic coordinates exposed by the presentations.
The affine layer is normalised using Lafont's linear normal form \cite{lafont_boolean_circuits_2003} together with an ordered translation column.
The diagonal layer is read as a finite-field phase polynomial in the selected coordinates.
The finite transport equations certify the local substitution cases, while the diagonal commutation and collection equations put the resulting phase layer into its fixed coefficient-vector form.
Layered existence together with affine and diagonal uniqueness proves completeness.
Two circuits are equal in the semantic model if and only if their equality is derivable in the corresponding presentation.

\Cref{sec:aff-axioms} presents the affine PROP, its semantics, and its affine normal forms.
\Cref{sec:finite-angle-fragments} adds the finite-angle phase generators, the transport equations, and the completeness theorem for the degree-one, degree-two, and degree-three fragments.
The appendices contain the group-presentation details, finite derivation certificates, polynomial checks, and deferred normalisation proofs.

\section{Affine circuits}
\label{sec:aff-axioms}

For a prime \(d\), the affine circuits in this section act on computational-basis labels in \(\F_d^n\).
The generators have the concrete readings \(x\mapsto x+1\), \((x,y)\mapsto(x,x+y)\), and \(x\mapsto ax\), together with permutations of wires.
Composing gates composes these label functions, and placing circuits side by side takes their product on disjoint blocks of variables.
The normal form later in the section records exactly the resulting affine map \(x\mapsto Ax+b\).
It writes the linear part \(A\) in Lafont normal form and the translation vector \(b\) as an ordered output column of \(X\)-gates.

The PROP language gives this circuit-level picture a syntax for local rewriting \cite{maclane_categories_1998,selinger_survey_2011,lafont_boolean_circuits_2003}.
An object is a number of wires, composition is sequential plugging, tensor is parallel placement, and the symmetric structure supplies wire permutations.
With this convention, a displayed equation can be used inside any larger circuit, and its interpretation is checked by the corresponding affine function.

\begin{definition}\label{def:prop}
A \emph{PROP} is a strict symmetric monoidal category \(\cat{C}\) whose objects are the natural numbers, with monoidal product on objects given by addition and unit \(0\).
We write categorical composition as \(\circ\), and the monoidal product of morphisms as \(\otimes\).
The symmetry supplies maps \(\sigma_{n,m}:n+m\to m+n\); more generally, a permutation \(\pi\in S_n\) determines a canonical symmetry \(\sigma_\pi:n\to n\).
\end{definition}

\begin{definition}\label{def:presented-prop}
Let \(\Gamma\) be a set of typed diagram equations over a fixed generator signature.
We write \(\Gamma\vdash D_1=D_2\) when \(D_1\) and \(D_2\) are related by the least equivalence relation on well-typed diagrams that contains \(\Gamma\) and is closed under composition and tensoring on either side with well-typed diagrams.
The quotient by this relation is the PROP presented by the generators and equations \(\Gamma\).
\end{definition}

Closure under composition and tensor is the presentation-theoretic form of local circuit rewriting.
Each displayed rule may be inserted into a surrounding circuit, placed next to additional wires or diagrams, and transported by wire permutations.

\begin{definition}\label{def:prop-interpretation}
Let \(\cat{C}\) and \(\cat{D}\) be PROPs.
A \emph{PROP functor} \(F:\cat{C}\to\cat{D}\) is a strict symmetric monoidal functor that is the identity on objects.
It preserves the circuit operations, so \(F(D_2\circ D_1)=F(D_2)\circ F(D_1)\), \(F(D_1\otimes D_2)=F(D_1)\otimes F(D_2)\), and \(F(\sigma_\pi)=\sigma_\pi\).
When \(\cat{C}\) is presented by generators and equations, an \emph{interpretation} of \(\cat{C}\) in \(\cat{D}\) is a choice of morphism in \(\cat{D}\) for each generator, respecting every defining equation.
It then extends uniquely to a PROP functor, and every equality derivable in the presentation holds after interpretation.
\end{definition}

\subsection{The semantic PROP of reversible affine maps}

Fix a prime \(d\).
We write vectors in \(\F_d^n\) as columns, and we write \(\AGL_n(\F_d)\) for the affine group of invertible affine maps \(x\mapsto Ax+b\),
with \(A\in \GL_n(\F_d)\) and \(b\in \F_d^n\).
All generators considered here are reversible and preserve wire number, so the semantic PROP has morphisms only from \(n\) to \(n\), and every such morphism is invertible; the tensor still records parallel composition and the contexts in which equations may be used.

\begin{definition}\label{def:AffLab}
In the semantic affine-label PROP \(\cat{AffLab}_d\), a morphism \(n\to n\) is an affine bijection \(g:\F_d^n\to\F_d^n\), and there are no morphisms \(n\to m\) when \(n\neq m\).
Composition is ordinary composition of functions.
On morphisms, the tensor product is induced by products of label sets, using the identification \(\F_d^{n+m}\cong\F_d^n\times\F_d^m\).
The morphism \(f\otimes g\) acts by \((x,y)\mapsto(f(x),g(y))\), and the symmetry \(\sigma_{n,m}:n+m\to n+m\) swaps the two coordinate blocks.
\end{definition}

\begin{remark}\label{rem:AffLab-coordinates}
The functional presentation is isomorphic to the usual coordinate presentation by pairs.
Every morphism \(g:n\to n\) in \(\cat{AffLab}_d\) has a unique expression \(g(x)=Ax+b\), so \(\cat{AffLab}_d(n,n)=\AGL_n(\F_d)\).
Under this identification,
\((A_2,b_2)\circ(A_1,b_1)=(A_2A_1,A_2b_1+b_2)\).
If \(f_i(x)=A_i x+b_i\), then \(f_1\otimes f_2\) has matrix \(A_1\oplus A_2\) and translation vector \(b_1\mid b_2\), where \(A_1\oplus A_2\) is block diagonal and \(b_1\mid b_2\) denotes the block-column concatenation of \(b_1\) and \(b_2\).
For \(n=0\), we take \(\F_d^0\) to be the singleton set; then \(\AGL_0(\F_d)\) is the trivial group,
and \(\cat{AffLab}_d(0,0)\) has a single morphism.
\end{remark}

\subsection{The diagrammatic PROP \texorpdfstring{\(\cat{Aff}_d\)}{Aff\_d}}

We read diagrams from left to right.
A morphism \(D:n\to m\) is drawn with \(n\) input wires and \(m\) output wires, so the object \(n\) is drawn as \(n\) parallel wires.
Sequential composition \(D_2\circ D_1\) is horizontal plugging, and tensor product \(D_1\otimes D_2\) is vertical juxtaposition.
The straight wire \(\gI:1\to1\) is the identity on one wire, the empty diagram \(\gempty:0\to0\) is the identity on the tensor unit, and the crossing \(\gSWAP:2\to2\) is the symmetry \(\sigma_{1,1}\).

\begin{definition}\label{def:Affd}
Let \(\mathcal{E}_{\mathrm{aff}}\) consist of the following generator families:
\begin{center}
\begin{tabular}{ccc}
\(\gX:1\to 1\) &
\(\gCX:2\to 2\) &
\(\gM:1\to 1\) for each \(a\in\F_d^\times\).
\end{tabular}
\end{center}
The diagrammatic affine PROP \(\cat{Aff}_d\) is presented by \(\mathcal{E}_{\mathrm{aff}}\) and by the equations \(\et{Aff}_d\), displayed in \cref{fig:Aff_d_rules}.
\end{definition}

For \(m\in\N\), the nodes \(\gXm\) and \(\gCXm\) denote the \(m\)-fold composites of \(\gX\) and \(\gCX\), respectively.
When \(1\le k\le d-1\), we write
\(\scalebox{.69}{\tikzfig{./pieces/Xkminus} :=\tikzfig{./pieces/Xkminus2}}\),
\(\scalebox{.69}{\tikzfig{./pieces/CXkminus} :=\tikzfig{./pieces/CXkminus2}}\),
\(\scalebox{.69}{\tikzfig{./pieces/Mminus} :=\tikzfig{./pieces/Mminus2}}\).

\begin{definition}\label{def:Aff-interpretation}
There is a PROP functor
\(\interp{\cdot}_{Aff}:\cat{Aff}_{d}\to \cat{AffLab}_{d}\).
For each multiplier parameter \(a\in\F_d^\times\), the functor is determined on the basic diagrams by
\begin{align*}
\interp{\gI}_{Aff}&:x\mapsto x,&
\interp{\gempty}_{Aff}&=\id_{\F_d^0},&
\interp{\gM}_{Aff}&:x\mapsto ax,\\
\interp{\gX}_{Aff}&:x\mapsto x+1,&
\interp{\gCX}_{Aff}&:(x,y)\mapsto(x,x+y),&
\interp{\gSWAP}_{Aff}&:(x,y)\mapsto(y,x).
\end{align*}
Here \(\F_d^0\) is the singleton label set of the tensor unit, so \(\id_{\F_d^0}\) is its unique endomorphism.
\end{definition}

In \cref{fig:Aff_d_rules}, the left column records multiplier laws and their interactions with translations and controlled additions; the right column records translation through controlled addition, the swap equation, and the structural wire-change equation.

\begin{figure}[htb!]
  \centering
    \fbox{%
      \begin{minipage}{0.8\linewidth}
        \begin{minipage}[t]{0.5\linewidth}
          \begin{equation}\label{axiom-mult1}
            \tikzfig{axioms/AxiomMult1} = \tikzfig{axioms/AxiomPNullB}
          \end{equation}
          \begin{equation}\label{axiom-multxy}
            \tikzfig{axioms/AxiomMultXMultY} = \tikzfig{axioms/AxiomMultXY}
          \end{equation}
          \begin{equation}\label{axiom-XM-twisted}
            \tikzfig{axioms/AxiomXMx-twisted} = \tikzfig{axioms/AxiomMxX-twisted}
          \end{equation}
          \begin{equation}\label{axiom-multcnot}
            \tikzfig{axioms/AxiomMultCNOTA} = \tikzfig{axioms/AxiomMultCNOTB}
          \end{equation}
          \begin{equation}\label{axiom-multcnot2-twisted}
            \tikzfig{axioms/AxiomMultCNOTtwisted2A} = \tikzfig{axioms/AxiomMultCNOTtwisted2B}
          \end{equation}
        \end{minipage}\hfill
        \begin{minipage}[t]{0.5\linewidth}
          \begin{equation}\label{axiom-XCNOTtop}
            \tikzfig{axioms/AxiomXCXtop} = \tikzfig{axioms/AxiomCXXtop}
          \end{equation}
          \begin{equation}\label{axiom-swap-fixed}
            \tikzfig{axioms/AxiomSwap} = \tikzfig{axioms/AxiomSwapFixedDecomp}
          \end{equation}
          \begin{equation}\label{axiom-I}
            \tikzfig{axioms/AxiomXChangeWiresBaseA} = \tikzfig{axioms/AxiomXChangeWiresBaseB}
          \end{equation}
        \end{minipage}
      \end{minipage}%
    }
  \caption{The ruleset \(\et{Aff}_d\) for circuits of affine transformations over \(\F_d\).}
  \label{fig:Aff_d_rules}
\end{figure}

\begin{lemma}\label{lem:Aff-soundness}
Every axiom in \(\et{Aff}_d\) holds under \(\interp{\cdot}_{Aff}\).
\end{lemma}
The verification is a coordinate calculation on affine label functions, deferred to Appendix~\ref{app:proofs-soundness}.

\subsection{Affine layers and normal forms}
\label{subsec:affine-normal-form}

Affine circuits admit a two-layer normal form following the decomposition of an affine map \(g(x)=Ax+b\).
The linear automorphism \(A\in\GL_n(\F_d)\) is written in Lafont's normal form for the linear sublanguage generated by \(\gCX\), \(\gM\), and symmetries.
The finite derivations needed to use Lafont's linear rules inside \(\et{Aff}_d\) are recorded in Appendix~\ref{derivation:XX-comm}.
The translation vector \(b\in\F_d^n\) is written as an ordered output column, with one entry for each wire.

\begin{definition}[\cite{lafont_boolean_circuits_2003}]\label{def:CX-stairs}
An ascending \(\gCX\)-stair is a circuit of the form
\begin{center}
\(
  \vcenter{\hbox{\scalebox{0.9}{\qcfig{stairleft}}}}
  =
  \vcenter{\hbox{\scalebox{0.9}{\qcfig{stairright}}}}
\)
\end{center}
The descending \(\gCX\)-stair is defined similarly, with the staircase direction reversed.
\end{definition}

\begin{definition}[\cite{lafont_boolean_circuits_2003}]\label{def:CX-layer}
A ladder is a circuit of the form
\begin{center}
\(
    \vcenter{\hbox{\scalebox{0.9}{\qcfig{ladderleft}}}}
  =
  \vcenter{\hbox{\scalebox{0.9}{\qcfig{ladderright}}}}
\)
\end{center}
\end{definition}

\begin{definition}\label{def:affine-normal-form}
For \(g\in\AGL_n(\F_d)\), write \(g(x)=Ax+b\), with \(A\in\GL_n(\F_d)\) and \(b=(b_1,\ldots,b_n)\in\F_d^n\).
Let \(\NF_{\mathrm{lin}}(A)\) be Lafont's linear normal form for \(A\), and set
\(\NF_{\mathrm{tr}}(b)\coloneqq \gXbone\otimes\cdots\otimes\gXbn\).
The affine normal form of \(g\) is
\(\NF_{\mathrm{aff}}(g)\coloneqq \NF_{\mathrm{tr}}(b)\circ\NF_{\mathrm{lin}}(A)\).
Diagrammatically, this is the linear normal form followed by one labelled translation on each output wire.
\begin{center}
\(
    \vcenter{\hbox{\scalebox{0.9}{\qcfig{nfaff}}}}
\)
\end{center}
\end{definition}

\begin{lemma}\label{lem:normal-form-unique}
For \(g\in\AGL_n(\F_d)\), \(\NF_{\mathrm{aff}}(g)\) is the unique affine normal form \(N\) with \(\interp{N}_{Aff}=g\).
\end{lemma}

\begin{proof}
Write \(g(x)=Ax+b\), with \(A\in\GL_n(\F_d)\) and \(b\in\F_d^n\).
The pair \((A,b)\) is unique.
Lafont's theorem gives a unique circuit \(\NF_{\mathrm{lin}}(A)\) for the linear map \(x\mapsto Ax\).
The translation column \(\NF_{\mathrm{tr}}(b)\) is uniquely determined by the coordinates of \(b\).
Their composite has interpretation \(x\mapsto Ax+b\).
Conversely, any affine normal form with this interpretation determines the same pair \((A,b)\), hence the same linear normal form and the same translation column.
\end{proof}

\begin{lemma}\label{lem:affine-normalisation}
\stmtAffineNormalisation
\end{lemma}

The proof is given in Appendix~\ref{app:proofs-affine-normalisation}; it collects translations, normalises the remaining linear layer, and then reattaches the translation column.

\begin{theorem}\label{thm:Aff-sound-complete}
For any affine circuits \(A, A':n\to n\) in \(\cat{Aff}_d\), if \(\interp{A}_{Aff}=\interp{A'}_{Aff}\) then \(\et{Aff}_d\vdash A=A'\).
\end{theorem}

\begin{proof}
Let \(g=\interp{A}_{Aff}=\interp{A'}_{Aff}\).
Normalising both circuits gives \(\et{Aff}_d\vdash A=\NF_{\mathrm{aff}}(g)\) and \(\et{Aff}_d\vdash A'=\NF_{\mathrm{aff}}(g)\).
By symmetry of equality and transitivity, \(\et{Aff}_d\vdash A=A'\).
\end{proof}

\section{Finite-angle phase fragments}
\label{sec:finite-angle-fragments}

\subsection{Phase functions and degree fragments}
\label{subsec:phase-functions-degree-fragments}

Fix a prime \(d\).
Let \(\Hilb_d=\C^d\) with computational basis \(\{\ket{x}\mid x\in\F_d\}\), and write \(\omega=e^{2\pi i/d}\).
The affine calculus of \cref{sec:aff-axioms} describes reversible label updates \(x\mapsto g(x)\), with \(g\in\AGL_n(\F_d)\).
The phase extensions add diagonal phases while keeping the affine update and the phase function as separate coordinates.

\begin{definition}\label{def:phase-functions}
For \(n\in\N\), a phase function on \(n\) wires is a function \(q:\F_d^n\to\F_d\), interpreted as the diagonal operator \(D_q\ket{x}=\omega^{q(x)}\ket{x}\).
\end{definition}

\begin{definition}\label{def:phase-function-classes}
A phase function \(q:\F_d^n\to\F_d\) has degree at most \(r\) if it is represented by a polynomial over \(\F_d\) of total degree at most \(r\).
Let \(\Lin_n\), \(\Quad_n\), and \(\Cube_n\) be the sets of phase functions on \(n\) wires of degree at most \(1\), \(2\), and \(3\), respectively.
Write \(\Lin=(\Lin_n)_{n\in\N}\), \(\Quad=(\Quad_n)_{n\in\N}\), and \(\Cube=(\Cube_n)_{n\in\N}\).
\end{definition}

The one-wire binomial representatives used below match the first part of the diagonal Clifford hierarchy.
Write \(\zeta_m=e^{2\pi i/d^m}\), so \(\zeta_1=\omega\).
The classification of diagonal hierarchy gates implies that a one-qudit phase \(\ket{x}\mapsto \zeta_m^{f(x)}\ket{x}\), whose finest \(d^m\)-root component has polynomial degree at most \(q\), \(1\le q\le d-1\), lies in level \((d-1)(m-1)+q\).
The level is exact when a degree-\(q\) coefficient appears at that precision \cite[Theorem~2]{cui_diagonal_2017}.
The \(d\)-th-root phases therefore run through degrees \(1,\ldots,d-1\); deeper one-wire diagonal gates restart at degree \(1\) with \(d^2\)-th roots.
No degree-\(d\) candidate gives a new \(d\)-th-root phase, since \(x^d=x\) as a function on \(\F_d\), and \(\binom{x}{d}\) vanishes on the standard labels \(0,\ldots,d-1\).

\begin{figure}[htb!]
\centering
\(
\begin{array}{@{}r@{\quad}l@{}}
\text{monomial} &
  \zeta_1^x \to \cdots \to \zeta_1^{x^{d-1}} \to \zeta_2^x \to \zeta_2^{x^2} \to \cdots \\[.4ex]
\text{binomial} &
  \zeta_1^{\binom{x}{1}} \to \cdots \to \zeta_1^{\binom{x}{d-1}} \to
  \zeta_2^{\binom{x}{1}} \to \zeta_2^{\binom{x}{2}} \to \cdots
\end{array}
\)
\caption{One-wire representatives for the same diagonal Clifford-hierarchy levels.}
\label{fig:clifford-binomial-ladder}
\end{figure}

The present fragments use the first three binomial representatives in this ladder.
The generators \(\gZ\), \(\gS\), and \(\gT\) have exponents \(x\), \(\binom{x}{2}\), and \(\binom{x}{3}\).
For \(q<d\), \(\binom{x}{q}\) has leading term \(x^q/q!\), so these gates have the same hierarchy levels as the corresponding monomial representatives.
For prime \(d>3\), \(\gT\) therefore has exact hierarchy level three and is non-Clifford.
Corollary~27 of \cite{borda_quantum_2025} then implies that the single-qudit Clifford+\(\gT\) gate set is dense in \(\mathrm{SU}(d)\).
It is sometimes assumed that qudit universality requires phase precision finer than \(\omega\), whereas \(\gT\) has only \(d\)-th-root eigenvalues.
Binomials make transport explicit because affine shifts and sums expand by Pascal identities, and those expansions produce exactly the coordinates used in the normal forms.

\begin{definition}\label{def:admissible-degree-fragments}
For the fixed prime \(d\), the admissible degree fragments are \(\Lin\), also \(\Quad\) when \(d\) is odd, and also \(\Cube\) when \(d>3\).
\end{definition}

The restrictions in this definition come from the coordinate choice.
The quadratic coordinate \(\binom{x}{2}\) requires \(2^{-1}\in\F_d\), and the cubic coordinate \(\binom{x}{3}\) requires \(6^{-1}\in\F_d\).
In these cases the degree bound stays below the characteristic, so the normal-form coefficients are determined by the phase function on \(\F_d^n\).

\subsection{Binomial coordinates for affine transport}
\label{subsec:binomial-polynomials}

The transport equations are the circuit form of the following Pascal identities.
Moving \(D_q\) from the output side of an affine update \(P_g\) to its input side replaces the exponent by \(q\circ g\).
For the affine generators, this means substituting \(x+1\), \(kx\), or \(x+y\) into the binomial coordinates \(x\), \(\binom{x}{2}\), and \(\binom{x}{3}\).
The first two substitutions stay on one wire.
The substitution \(x\mapsto x+y\) creates the mixed coordinates represented later by \(\gCZ\), \(\gCS\), \(\gSC\), and \(\gCCZ\).

\begin{lemma}\label{lem:binom-identities}
Assume \(d\) is an odd prime.
For all \(x,y,k\in\F_d\) one has
\begin{center}
\begin{minipage}[t]{.27\linewidth}
\vspace{-1em}
\begin{equation}
\binom{x+1}{2}=\binom{x}{2}+x
\label{eq:binom-shift}
\end{equation}
\end{minipage}\hfill
\begin{minipage}[t]{.34\linewidth}
\vspace{-1em}
\begin{equation}
\binom{x+y}{2}=\binom{x}{2}+\binom{y}{2}+xy
\label{eq:binom-add}
\end{equation}
\end{minipage}\hfill
\begin{minipage}[t]{.3\linewidth}
\vspace{-1em}
\begin{equation}
\binom{kx}{2}=k^2\binom{x}{2}+\binom{k}{2}x
\label{eq:binom-scale}
\end{equation}
\end{minipage}
\end{center}
\end{lemma}

\begin{lemma}\label{lem:binom3-identities}
Assume \(d>3\) is prime.
For all \(x,y,k\in\F_d\) one has
\begin{center}
\begin{minipage}[t]{.35\linewidth}
\vspace{-1em}
\begin{equation}
\binom{x+1}{3}=\binom{x}{3}+\binom{x}{2}
\label{eq:binom3-shift}
\end{equation}
\end{minipage}
\begin{minipage}[t]{.55\linewidth}
\begin{equation}
\binom{x+y}{3}=\binom{x}{3}+\binom{y}{3}+\binom{x}{2}y+x\binom{y}{2}
\label{eq:binom3-add}
\end{equation}
\end{minipage}\\
\begin{minipage}[t]{.45\linewidth}
\begin{equation}
\binom{kx}{3}=k^3\binom{x}{3}+2k\binom{k}{2}\binom{x}{2}+\binom{k}{3}x
\label{eq:binom3-scale}
\end{equation}
\end{minipage}
\end{center}
\end{lemma}

\begin{lemma}\label{lem:monomial-to-binomial}
Assume \(d\) is an odd prime.
For all \(x,y\in\F_d\) one has
\begin{center}
\begin{minipage}[t]{.27\linewidth}
\vspace{-1em}
\begin{equation}
x^2=2\binom{x}{2}+x
\label{eq:x2-binom}
\end{equation}
\end{minipage}\hfill
\begin{minipage}[t]{.34\linewidth}
\begin{equation}
x^2y=2y\binom{x}{2}+xy
\label{eq:x2y-binom}
\end{equation}
\end{minipage}\hfill
\begin{minipage}[t]{.3\linewidth}
\vspace{-1em}
\begin{equation}
xy^2=2x\binom{y}{2}+xy
\label{eq:xy2-binom}
\end{equation}
\end{minipage}
\end{center}
If \(d>3\), then also \(x^3=6\binom{x}{3}+6\binom{x}{2}+x\).
\end{lemma}

\begin{example}\label{ex:monomial-binomial-shift}
Assume \(d\) is odd and transport a quadratic phase through the translation \(x\mapsto x+1\).
For the monomial coordinate \(x^2\), the finite difference is \(\Delta(x^2)=(x+1)^2-x^2=2x+1\), so transport introduces the linear correction \(2x\) and the scalar \(\omega\).
For the binomial coordinate, \(\Delta\binom{x}{2}=\binom{x+1}{2}-\binom{x}{2}=x\), so the correction is exactly the next lower binomial coordinate, with unit coefficient and no scalar.
This unit-triangular finite-difference rule gives the shorter transport equations used below.
\end{example}

\begin{lemma}\label{lem:phase-coordinate-uniqueness}
In each admissible fragment, a phase function has a unique expansion in the normal-form coordinates for that fragment.
More explicitly, with the unavailable higher-degree terms omitted in the linear and quadratic fragments, every \(q\in\Lin_n\), \(q\in\Quad_n\), or \(q\in\Cube_n\) has a unique expression
\begin{equation}\label{eq:phase-coordinate-expansion}
\begin{aligned}
q(x)= {}& c+\sum_i a_i x_i
       +\sum_i b_i\binom{x_i}{2}+\sum_{i<j} c_{ij}x_ix_j \\
       &+\sum_i t_i\binom{x_i}{3}
       +\sum_{i\neq j} u_{ij}x_i\binom{x_j}{2}
       +\sum_{i<j<k} v_{ijk}x_ix_jx_k
\end{aligned}
\end{equation}
with coefficients in \(\F_d\).
\end{lemma}

\begin{proof}
Fix an admissible fragment with degree bound \(r\le 3\).
Since \(r<d\), every degree-\(\le r\) representative is reduced.
No variable has exponent at least \(d\).
A reduced polynomial over \(\F_d\) that vanishes on \(\F_d^n\) is the zero polynomial.
The function on \(\F_d^n\) determines the usual monomial coefficients.
It remains to recover the displayed coordinates from those monomial coefficients.

The change of basis is triangular.
The constant coordinate, the linear coordinates \(x_i\), and the square-free mixed coordinates \(x_ix_j\) and \(x_ix_jx_k\) are already monomials.
The remaining coordinates have highest-degree terms \(x^2/2\), \(xy^2/2\), and \(x^3/6\), respectively, for \(\binom{x}{2}\), \(x\binom{y}{2}\), and \(\binom{x}{3}\).
Admissibility makes \(2\) and \(6\) invertible whenever these terms occur.
Starting in highest degree and moving downward recovers the displayed coefficients from the monomial coefficients.
Two displayed expansions define the same phase function only when all coefficients agree.
\end{proof}

The transport equations in \cref{fig:LinPhase_rules,fig:QuadPhase_rules,fig:CubicPhase_rules} are circuit forms of these identities.
Moving a diagonal phase past an affine generator substitutes the affine update into its phase polynomial and expands the result in the same binomial coordinates.
\Cref{lem:monomial-to-binomial,lem:phase-coordinate-uniqueness} connect those coordinates with the standard monomial basis and record the coefficient uniqueness used in the diagonal uniqueness proof.

\subsection{Semantic PROPs of phase-affine operators}
\label{subsec:semantic-phase-props}

For \(n\in\N\), identify \(\F_d^n\) with the computational basis of \(\Hilb_d^{\otimes n}\) via \(x\mapsto\ket{x}\).
For \(g\in\AGL_n(\F_d)\), let \(P_g\) be the permutation operator defined by \(P_g\ket{x}=\ket{g(x)}\).

The semantic PROP is the corresponding semidirect product, where affine permutations act on phase functions by precomposition.
A morphism \(P_gD_q\) first contributes the phase \(\omega^{q(x)}\) to \(\ket{x}\), then changes the basis label to \(g(x)\).

\begin{definition}\label{def:phase-affine-props}
For \(\mathcal{Q}\in\{\Lin,\Quad,\Cube\}\), define the PROP \(\cat{PhAff}^{\mathcal{Q}}_{d}\) as follows.
Its hom-sets are \(\cat{PhAff}^{\mathcal{Q}}_{d}(n,n)=\{P_gD_q\mid g\in\AGL_n(\F_d),\ q\in\mathcal{Q}_n\}\) and \(\cat{PhAff}^{\mathcal{Q}}_{d}(n,m)=\varnothing\) for \(n\neq m\).
Composition is ordinary composition of operators, so \(P_{g_2}D_{q_2}\circ P_{g_1}D_{q_1}=P_{g_2\circ g_1}D_{q_1+q_2\circ g_1}\).
The tensor product satisfies \(P_gD_q\otimes P_hD_r=P_{g\oplus h}D_{q\oplus r}\), where \((q\oplus r)(x,y)=q(x)+r(y)\).
The identity on \(n\) is \(P_{\id_n}D_0\), and the symmetry is \(P_{\sigma_{n,m}}D_0\).
The inverse of \(P_gD_q\) is \(P_{g^{-1}}D_{-q\circ g^{-1}}\).
\end{definition}

\begin{remark}\label{rem:phase-affine-coordinates}
The operator coordinates are unique.
If \(P_gD_q=P_hD_r\), comparison of basis labels gives \(g=h\), after which the primitivity of \(\omega\) gives \(q=r\).
For each \(\mathcal{Q}\in\{\Lin,\Quad,\Cube\}\), sums of exponents and affine precomposition stay in \(\mathcal{Q}\).
The formulas above define composition, tensor product, and inverses in the chosen degree class.
\end{remark}

The embedding \(g\mapsto P_gD_0\) and the projection \(P_gD_q\mapsto g\) relate this PROP to \(\cat{AffLab}_d\); the phase part is the group of phase functions acted on by affine substitution.

\subsection{Diagrammatic PROPs for finite-angle phases}
\label{subsec:diagrammatic-phase-props}

The semantic fragments specify which phase functions are allowed.
Their diagrammatic presentations extend \(\cat{Aff}_d\) by the one-wire diagonal primitives needed at each degree; the formal interpretation is given in \cref{subsec:phase-interpretations}.
The scalar \(\gW:0\to0\) contributes the constant exponent \(1\).
The generators \(\gZ\), \(\gS\), and \(\gT\) contribute \(x\), \(\binom{x}{2}\), and \(\binom{x}{3}\), respectively, in the fragments where they are present.
The derived diagonal terms forced by affine transport are circuit abbreviations, not primitive generators.

\begin{definition}\label{def:phase-fragment-props}
For any prime \(d\), \(\cat{LinPhase}_d\) is obtained from \(\cat{Aff}_d\) by adjoining \(\gW:0\to0\), \(\gZ:1\to1\), and the equations in \cref{fig:LinPhase_rules}.
For odd prime \(d\), \(\cat{QuadPhase}_d\) is obtained from \(\cat{LinPhase}_d\) by adjoining \(\gS:1\to1\) and the equations in \cref{fig:QuadPhase_rules}.
For prime \(d>3\), \(\cat{CubicPhase}_d\) is obtained from \(\cat{QuadPhase}_d\) by adjoining \(\gT:1\to1\) and the equations in \cref{fig:CubicPhase_rules}.
\end{definition}

\begin{definition}\label{def:primitive-phase-labels}
For \(k\in\N\), the diagrams \(\scalebox{.69}{\textup{\tikzfig{./pieces/wk}}}\), \(\scalebox{.69}{\textup{\tikzfig{./pieces/Zk}}}\), \(\scalebox{.69}{\textup{\tikzfig{./pieces/Sk}}}\), and \(\scalebox{.69}{\textup{\tikzfig{./pieces/Tk}}}\) denote the \(k\)-fold iterates of \(\gW\), \(\gZ\), \(\gS\), and \(\gT\), respectively, when the generator belongs to the current fragment.
For \(1\le k\le d-1\), the negative labels \(\scalebox{.69}{\textup{\tikzfig{./pieces/wkminus}}}\), \(\scalebox{.69}{\textup{\tikzfig{./pieces/Zkminus}}}\), \(\scalebox{.69}{\textup{\tikzfig{./pieces/Skminus}}}\), and \(\scalebox{.69}{\textup{\tikzfig{./pieces/Tkminus}}}\) denote the corresponding \((d-k)\)-fold iterates.
\end{definition}

The derived diagonals are abbreviations for circuits already present in the relevant fragment.
They are named because normal forms record their coefficients directly.

\begin{definition}\label{def:QuadPhase-derived-diagonals}
In \(\cat{QuadPhase}_d\), define the bilinear quadratic diagonal \(\gCZ:2\to2\) by
\begin{center}
\(\scalebox{0.7}{\textup{\tikzfig{./pieces/CZ}}}\coloneqq \scalebox{0.6}{\textup{\tikzfig{./pieces/CZdef}}}\).
\end{center}
Its finite iterate indexed by \(k\in\F_d\) is denoted by \(\scalebox{.69}{\textup{\tikzfig{./pieces/CZk}}}\).
\end{definition}

\begin{definition}\label{def:CubicPhase-derived-diagonals}
In \(\cat{CubicPhase}_d\), define the controlled-square diagonal \(\gCS:2\to2\) by
\begin{center}
\(
\scalebox{0.7}{\textup{\tikzfig{./pieces/CS}}}\coloneqq \scalebox{0.6}{\textup{\tikzfig{./axioms/CSdef}}}, \text{ where }t=2^{-1}=\frac{d+1}{2}\in\F_d,
\)
\end{center}
The family \(\gSC\) is the two-wire flip of \(\gCS\).
We write \(\scalebox{.69}{\textup{\tikzfig{./pieces/CSk}}}\) for the \(k\)-fold iterate of \(\gCS\); iterates of \(\gSC\) are the corresponding wire-flips of the \(\gCS\)-iterates.
\end{definition}

\begin{definition}\label{def:CubicPhase-derived-ccz}
In \(\cat{CubicPhase}_d\), define the trilinear cubic diagonal \(\gCCZ:3\to3\) by
\begin{center}
\(
\scalebox{0.7}{\textup{\tikzfig{./pieces/CCZ}}}\coloneqq \scalebox{0.6}{\textup{\tikzfig{./axioms/CCZdef}}}.
\)
\end{center}
We write \(\scalebox{.69}{\textup{\tikzfig{./pieces/CCZk}}}\) for its \(k\)-fold iterate.
\end{definition}

The required symmetries for \(\gCZ\) and \(\gCCZ\) are derived in \cref{ax:CZ-SWAP,ax:CCZ-SWAP,ax:CCZ-SWAP2}, so their displayed wire orders may be permuted coherently.

\begin{definition}\label{def:phase-diagonal-families}
The diagonal family sets used by the admissible fragments are \(\mathcal{D}_{\Lin}=\{\gW,\gZ\}\), \(\mathcal{D}_{\Quad}=\mathcal{D}_{\Lin}\cup\{\gS,\gCZ\}\), and \(\mathcal{D}_{\Cube}=\mathcal{D}_{\Quad}\cup\{\gT,\gCS,\gSC,\gCCZ\}\).
\end{definition}

\newcommand{\phaseRuleEq}[2]{%
  \refstepcounter{equation}\label{#1}%
  \(\displaystyle #2\)\nobreak\hspace{0.35em}\textup{(\theequation)}%
}
\newcommand{\phaseRuleGap}{\\[1.8ex]}
\newcommand{\phaseRuleScale}{0.9}

\Cref{fig:LinPhase_rules,fig:QuadPhase_rules,fig:CubicPhase_rules} give the additional equations for the three phase fragments.
Together with the affine presentation, and with lower-degree phase rules inherited by the higher fragments, these finite rules present \(\cat{LinPhase}_d\), \(\cat{QuadPhase}_d\), and \(\cat{CubicPhase}_d\), respectively.
The commutation and transport certificates used later are derived inside these presentations.

\begin{figure}[htb!]
  \centering
  \begingroup
  \setlength{\tabcolsep}{0pt}
  \setlength{\fboxsep}{2pt}
  \renewcommand{\arraystretch}{1}
  \scalebox{\phaseRuleScale}{%
    \fbox{%
      \begin{tabular}{@{}c@{\hspace{0.4em}}c@{\hspace{0.4em}}c@{}}
        \phaseRuleEq{ax:W-order-d}{\tikzfig{./axioms/wd} = \tikzfig{./pieces/empty}}
        &
        \phaseRuleEq{ax:Z-order-d}{\tikzfig{./axioms/Zd} = \tikzfig{./pieces/I}}
        &
        \phaseRuleEq{ax:ZC-target}{\tikzfig{./axioms/ZCNOTtop} = \tikzfig{./axioms/CNOTZtop}}
        \phaseRuleGap
        \phaseRuleEq{ax:ZM}{\tikzfig{./axioms/MZ} = \tikzfig{./axioms/ZM}}
        &
        \phaseRuleEq{ax:ZX}{\tikzfig{./axioms/XZ} = \tikzfig{./axioms/ZX}}
        &
      \end{tabular}%
    }%
  }
  \endgroup
  \caption{Additional equations in the presentation of \(\cat{LinPhase}_d\). They impose finite order on scalar and linear phases and specify affine transport for the linear phase.}
  \label{fig:LinPhase_rules}
\end{figure}

\begin{figure}[htb!]
  \centering
  \begingroup
  \setlength{\tabcolsep}{0pt}
  \setlength{\fboxsep}{2pt}
  \renewcommand{\arraystretch}{1}
  \scalebox{\phaseRuleScale}{%
    \fbox{%
      \begin{tabular}{@{}c@{\hspace{0.4em}}c@{\hspace{0.4em}}c@{}}
        \phaseRuleEq{ax:S-order-d}{\tikzfig{./axioms/Sd} = \tikzfig{./pieces/I}}
        &
        \phaseRuleEq{ax:SX}{\tikzfig{./axioms/XS} = \tikzfig{./axioms/SX}}
        &
        \phaseRuleEq{ax:SZ}{\tikzfig{./axioms/SZ} = \tikzfig{./axioms/ZS}}
        \phaseRuleGap
        \phaseRuleEq{ax:SM}{\tikzfig{./axioms/MS} = \tikzfig{./axioms/SM}}
        &
        \phaseRuleEq{ax:CZ-neg}{\tikzfig{./axioms/CZnegL} = \tikzfig{./axioms/CZnegR}}
        &
        \phaseRuleEq{ax:SC-control}{\tikzfig{./axioms/SCNOTbot} = \tikzfig{./axioms/CNOTSbot}}
        \phaseRuleGap
        &
        \phaseRuleEq{ax:CC-control}{\tikzfig{./axioms/CZCNOTbot} = \tikzfig{./axioms/CNOTCZbot}}
        &
      \end{tabular}%
    }%
  }
  \endgroup
  \caption{Additional equations in the presentation of \(\cat{QuadPhase}_d\). They impose finite order and specify affine transport for the quadratic phase, including the derived \(\CZ\)-terms forced by controlled addition.}
  \label{fig:QuadPhase_rules}
\end{figure}

\begin{figure}[htb!]
  \centering
  \begingroup
  \setlength{\tabcolsep}{0pt}
  \setlength{\fboxsep}{2pt}
  \renewcommand{\arraystretch}{1}
  \scalebox{\phaseRuleScale}{%
    \fbox{%
      \begin{tabular}{@{}c@{\hspace{0.4em}}c@{\hspace{0.4em}}c@{}}
        \phaseRuleEq{ax:T-order-d}{\tikzfig{./axioms/Td} = \tikzfig{./pieces/I}}
        &
        \phaseRuleEq{ax:TM}{\tikzfig{./axioms/TM} = \tikzfig{./axioms/MT}}
        &
        \phaseRuleEq{ax:TS}{\tikzfig{./axioms/TS} = \tikzfig{./axioms/ST}}
        \phaseRuleGap
        \phaseRuleEq{ax:TX}{\tikzfig{./axioms/XT} = \tikzfig{./axioms/TX}}
        &
        \phaseRuleEq{ax:CS-M-bot}{\tikzfig{./axioms/MCSbot} = \tikzfig{./axioms/CSMbot}}
        &
        \phaseRuleEq{ax:TZ}{\tikzfig{./axioms/TZ} = \tikzfig{./axioms/ZT}}
        \phaseRuleGap
        \phaseRuleEq{ax:CS-M-top}{\tikzfig{./axioms/MCS} = \tikzfig{./axioms/CSM}}
        &
        \phaseRuleEq{ax:CS-CNOT-bot2}{\tikzfig{./axioms/CNOTCSbot2} = \tikzfig{./axioms/CSCNOTbot2}}
        &
        \phaseRuleEq{ax:TC-control}{\tikzfig{./axioms/TCNOTbot} = \tikzfig{./axioms/CNOTTbot}}
        \phaseRuleGap
        \phaseRuleEq{ax:CCZ-M}{\tikzfig{./axioms/CCZM} = \tikzfig{./axioms/MCCZ}}
        &
        \phaseRuleEq{ax:CCZ-CNOTtop}{\tikzfig{./axioms/CCZCNOTtop} = \tikzfig{./axioms/CNOTCCZtop}}
        &
        \phaseRuleEq{ax:CS-CNOT-bot}{\tikzfig{./axioms/CNOTCSbot} = \tikzfig{./axioms/CSCNOTbot}}
      \end{tabular}%
    }%
  }
  \endgroup
  \caption{Additional equations in the presentation of \(\cat{CubicPhase}_d\). They specify finite order, diagonal collection, and affine transport for the cubic phase and the derived diagonal families introduced above. In \eqref{ax:CS-M-bot}, \(\mu=\binom{a}{2}\).}
  \label{fig:CubicPhase_rules}
\end{figure}

\subsection{Interpretations and soundness}
\label{subsec:phase-interpretations}

The interpretation sends every generator to the corresponding operator \(P_gD_q\).
Affine circuits determine the permutation \(P_g\); diagonal generators determine the phase \(D_q\).
Semantic equality can be checked in the unique coordinates \((g,q)\).

\begin{definition}\label{def:Phase-interpretation}
For each admissible fragment \(\mathcal{P}\), the interpretation \(\interp{\cdot}_{\mathcal{P}}\) sends an affine circuit \(A\) to \(P_{\interp{A}_{Aff}}D_0\); in particular, symmetries become coordinate permutations with zero phase.
The primitive diagonal generators present in \(\mathcal{P}\) are assigned by
\begin{align*}
\interp{\gW}_{\mathcal{P}}&=D_1,&
\interp{\gZ}_{\mathcal{P}}&=D_{(x\mapsto x)},&
\interp{\gS}_{\mathcal{P}}&=D_{(x\mapsto\binom{x}{2})},&
\interp{\gT}_{\mathcal{P}}&=D_{(x\mapsto\binom{x}{3})},
\end{align*}
with the last two assignments used only in the fragments where \(\gS\) and \(\gT\) are generators.
Here \(D_1\) on the tensor unit is the scalar operator \(\ket{\ast}\mapsto\omega\ket{\ast}\).
\end{definition}

\begin{lemma}\label{lem:Derived-semantics}
\stmtDerivedSemantics
\end{lemma}
Appendix~\ref{app:proofs-soundness} proves this by expanding the defining circuits in binomial coordinates.

With these phase functions identified, the transport equations can be read as substitution identities.

\begin{example}\label{ex:CS-M-bot}
On a basis value \((x,y)\), moving the multiplier across the lower wire replaces \(y\) by \(ay\) in the controlled-square phase \(x\binom{y}{2}\).
The binomial identity \(x\binom{ay}{2}=a^2x\binom{y}{2}+\binom{a}{2}xy\) yields the right-hand side of \cref{ax:CS-M-bot}.
It consists of a \(\gCS\)-factor with coefficient \(a^2\) and the bilinear correction carried by the \(\gCZ\)-factor.
\end{example}

\begin{example}\label{ex:d5-main-transport}
Take \(d=5\) and consider the circuit that applies controlled addition \(\gCX\), then applies the \(T\)-phase with coefficient \(2\) to the lower output wire.
On \(\ket{x,y}\), this circuit has affine part \((x,y)\mapsto(x,x+y)\) and phase exponent \(2\binom{x+y}{3}\).
Transporting the diagonal phase to the input side requires this exponent in cubic normal-form coordinates.
\begin{align*}
2\binom{x+y}{3}
  &=2\binom{x}{3}+2\binom{y}{3}
    +2y\binom{x}{2}+2x\binom{y}{2}.
\end{align*}
The layered normal form has the same affine controlled addition, preceded by a diagonal layer with coefficient \(2\) in the two \(T\)-slots and coefficient \(2\) in the two oriented mixed slots \(\gSC\) and \(\gCS\).
All remaining coordinates are zero, giving the circuit equality in \cref{fig:d5-transport-normalform}.

\begin{center}
\scalebox{0.92}{\tikzfig{d5-transport-normalform}}
{\captionsetup{hypcap=false}
\captionof{figure}{Transport of the lower \(T^2\)-phase through controlled addition at \(d=5\). The right-hand diagonal layer records the two one-wire cubic coefficients \(2\binom{x}{3}\), \(2\binom{y}{3}\) and the two oriented mixed coefficients \(2x\binom{y}{2}\), \(2y\binom{x}{2}\), followed by the same affine update.}
\label{fig:d5-transport-normalform}
\par}
\end{center}
\end{example}

\begin{lemma}\label{lem:LinPhase-soundness}
Fix an admissible phase fragment.
Every displayed axiom for that fragment is sound under its interpretation.
\end{lemma}
Appendix~\ref{app:proofs-soundness} compares the affine coordinate and the phase exponent of each displayed equation; the transport cases are exactly the binomial identities stated above.

\subsection{Diagonal normal forms}
\label{subsec:phase-normal-forms}

A diagonal normal form is an ordered coefficient vector for the phase function \(q\), rendered as columns and descending stairs.

\begin{definition}
A \(Z\)-, \(S\)-, or \(T\)-column circuit has the respective form
\begin{center}
\(\vcenter{\hbox{\scalebox{0.9}{\qcfig{Zcoldef}}}}
  =
  \vcenter{\hbox{\scalebox{0.9}{\qcfig{Zcol}}}} \qquad \vcenter{\hbox{\scalebox{0.9}{\qcfig{Scoldef}}}}
  =
  \vcenter{\hbox{\scalebox{0.9}{\qcfig{Scol}}}} \qquad \vcenter{\hbox{\scalebox{0.9}{\qcfig{Tcoldef}}}}
  =
  \vcenter{\hbox{\scalebox{0.9}{\qcfig{Tcol}}}}\)
\end{center}
Exponents lie in \(\F_d\); a zero label denotes the identity on that wire.
\end{definition}

\begin{definition}
A \(\gCZ\)-, \(\gCS\)-, or \(\gSC\)-descending stair has the corresponding form
\begin{center}
\(
  \vcenter{\hbox{\scalebox{0.9}{\qcfig{Utriangleleft}}}}
  =
  \vcenter{\hbox{\scalebox{0.9}{\qcfig{Utriangleright}}}}
  \quad 
  \vcenter{\hbox{\scalebox{0.9}{\qcfig{UCStriangleleft}}}}
  =
  \vcenter{\hbox{\scalebox{0.9}{\qcfig{UCStriangleright}}}}
  \quad 
  \vcenter{\hbox{\scalebox{0.9}{\qcfig{USCtriangleleft}}}}
  =
  \vcenter{\hbox{\scalebox{0.9}{\qcfig{USCtriangleright}}}}
  \quad
\)
\end{center}
Each label \(n_i\in\F_d\) specifies the corresponding two-wire iterate.
\end{definition}

\begin{definition}
  A \(\gCCZ\)-descending stair is a circuit of the form
  \begin{center}
  \(
    \vcenter{\hbox{\scalebox{0.9}{\qcfig{Vtriangleleft}}}}
  =
  \vcenter{\hbox{\scalebox{0.9}{\qcfig{Vtriangleright}}}} 
  \quad 
  \text { and }
  \quad 
  \vcenter{\hbox{\scalebox{0.9}{\qcfig{Vtrapezoidleft}}}}
  =
  \vcenter{\hbox{\scalebox{0.9}{\qcfig{Vtrapezoidright}}}}
  \)
  \end{center}
  Each label \(n_i\in\F_d\) specifies the corresponding three-wire iterate.
\end{definition}

\begin{definition}
Phase normal forms in \(\cat{LinPhase}_d\), \(\cat{QuadPhase}_d\), and \(\cat{CubicPhase}_d\) are diagonal circuits of the respective shapes
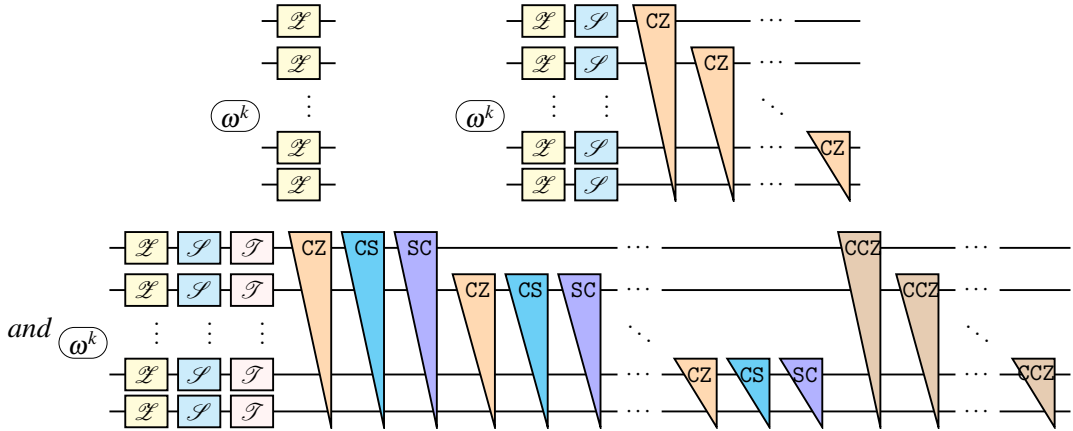
\begin{center}
\(\vcenter{\hbox{\scalebox{0.9}{\qcfig{nfdiagZ}}}}\, \qquad \qquad
\vcenter{\hbox{\scalebox{0.9}{\qcfig{nfdiagS}}}}\)

\(\text{and}\qquad
\vcenter{\hbox{\scalebox{0.9}{\qcfig{nfdiag}}}}\)
\end{center}
with all exponents taken in \(\F_d\).
\end{definition}

Semantically, the linear form records the scalar coefficient and the coefficients of \(x_i\).
The quadratic form adds the coefficients of \(\binom{x_i}{2}\) and \(x_ix_j\).
The cubic form adds the coefficients of \(\binom{x_i}{3}\), \(x_i\binom{x_j}{2}\) for \(i\neq j\), and \(x_ix_jx_k\).
With the order fixed, a diagonal normal form is a diagrammatic coefficient vector; \cref{lem:phase-coordinate-uniqueness} is the finite-field reason that this vector is determined by its phase function.

The coefficient-vector description also distinguishes the odd-prime setting from qubit CNOT-dihedral phase polynomials.
The qubit normal forms of Amy et al.\ use parity arguments such as \(x_i\oplus x_j\) with coefficients read in \(\Z_8\); expanding a parity as an ordinary product expression mixes \(\F_2\)-linear labels with \(\Z_8\)-valued exponents \cite{amy_cnot_dihedral_2018}.
For prime \(d>3\), labels and exponents both lie in \(\F_d\), and the derived families \(\gCZ\), \(\gCS\), \(\gSC\), and \(\gCCZ\) are coordinates of one finite-field polynomial basis.
The number of cubic diagonal normal forms on \(n\) wires is \(d^{1+3\binom{n}{1}+3\binom{n}{2}+\binom{n}{3}}=d^{\binom{n+3}{3}}\), and equality of diagonal normal forms is equality of the displayed coefficient vector.

\subsection{Transport principles and completeness}
\label{subsec:phase-completeness}

Fix an admissible degree tag \(\mathcal{P}\in\{\Lin,\Quad,\Cube\}\) for the current prime \(d\).
For \(n\) wires, a placed instance of \(G:r\to r\) is obtained by tensoring with identities and conjugating by a symmetry; for oriented two-wire families, placement includes the ordered pair of wires.
The local calculus below records the finite cancellation, commutation, and transport certificates used by the normal-form proof.

\begin{lemma}\label{lem:phase-finite-order}
\stmtPhaseFiniteOrder
\end{lemma}

\begin{proof}
The finite-order certificates are listed in \cref{tab:diagonal-family-finite-order}, with lower fragments obtained by restriction.
\end{proof}

\begin{lemma}\label{lem:phase-diagonal-commutation}
\stmtPhaseDiagonalCommutation
\end{lemma}

\begin{proof}
The non-scalar local cases are listed in \cref{tab:diag-comm-lemmas}; scalar factors are central.
Finite powers commute by repeated application of the corresponding local equation.
Placement in a larger circuit uses identities and symmetries.
\end{proof}

\begin{lemma}\label{lem:phase-one-step-transport}
\stmtPhaseOneStepTransport
\end{lemma}

\begin{proof}
The local transport certificates are listed in \cref{tab:diag-through-affine}.
Each entry is the circuit form of an affine-substitution identity \(q\mapsto q\circ g_A\) from \cref{subsec:binomial-polynomials}, placed in context by identities and symmetries.
\end{proof}

\begin{lemma}\label{lem:phase-diagonal-normalisation}
\stmtPhaseDiagonalNormalisation
\end{lemma}

\begin{proof}
Use \cref{lem:phase-diagonal-commutation} to sort the generators into the fixed scalar, one-wire, two-wire, and three-wire order.
For each placed generator, cancel every \(d\)-fold run of identical copies by \cref{lem:phase-finite-order}; equivalently, reduce each multiplicity modulo \(d\).
\end{proof}

\begin{lemma}\label{lem:phase-affine-layer-transport}
\stmtPhaseAffineLayerTransport
\end{lemma}

\begin{proof}
Write \(A\) as a product of placed affine generators and symmetries.
Expand each finite diagonal power into the corresponding finite product of placed generator copies.
Move each non-scalar copy across the placed affine generators with \cref{lem:phase-one-step-transport}; scalar copies commute, and symmetries only relabel placements.
Normalise the resulting diagonal circuit with \cref{lem:phase-diagonal-normalisation}.
The semantic formula follows by composing the one-step substitutions.
\end{proof}

\begin{lemma}\label{lem:phase-layered-factorisation}
\stmtPhaseLayeredFactorisation
\end{lemma}

\begin{proof}
Each affine generator has a layered form with the identity diagonal layer, and each diagonal generator has one with the identity affine layer.
Induct over a circuit word.
For the induction step, suppose the current prefix has layered form \(A\circ D\), and write the next generator in layered form \(A'\circ D'\).
By associativity, the composite has middle subcircuit \(D\circ A'\).
\Cref{lem:phase-affine-layer-transport} rewrites it as \(A'\circ D_1\), with \(D_1\) diagonal normal.
Then \cref{lem:affine-normalisation} puts \(A\circ A'\) in affine normal form, and \cref{lem:phase-diagonal-normalisation} puts \(D_1\circ D'\) in diagonal normal form.
\end{proof}

\begin{proposition}\label{prop:phase-layered-uniqueness}
\stmtPhaseLayeredUniqueness
\end{proposition}

\begin{proof}
Write the two interpretations as \(P_gD_q\) and \(P_{g'}D_{q'}\).
Coordinate uniqueness in \cref{rem:phase-affine-coordinates} gives \(g=g'\) and \(q=q'\).
Then \cref{lem:normal-form-unique} fixes the affine normal form.
For the diagonal layers, \cref{lem:phase-coordinate-uniqueness} makes the normal-form coefficient vector unique, so the two displayed diagonal normal forms are the same.
\end{proof}

\begin{theorem}\label{thm:phase-normal-form-complete}
\stmtPhaseNormalFormComplete
\end{theorem}

\begin{proof}
Use \cref{lem:phase-layered-factorisation} to rewrite \(C\) and \(C'\) as layered normal forms \(A\circ D\) and \(A'\circ D'\).
By soundness, the semantic hypothesis gives \(\interp{A\circ D}_{\mathcal{P}}=\interp{A'\circ D'}_{\mathcal{P}}\).
\Cref{prop:phase-layered-uniqueness} gives \(A=A'\) and \(D=D'\), so the two original circuits rewrite to the same normal form.
\end{proof}

\section{Conclusion}\label{sec:discussion}

The finite presentations here give equational calculi for prime-dimensional circuits generated by reversible affine updates and diagonal phase gates.
For each prime \(d\), the affine core \(\cat{Aff}_d\) presents the maps \(x\mapsto Ax+b\) over \(\F_d\).
Its normal forms separate the Lafont-style linear component from the translation column, and completeness reduces derivability to equality of affine relabellings.

The phase extensions \(\cat{LinPhase}_d\), \(\cat{QuadPhase}_d\) for odd \(d\), and \(\cat{CubicPhase}_d\) for \(d>3\) add diagonal generators for finite-field phase polynomials of degree at most \(1\), \(2\), and \(3\).
Their semantics is the strict symmetric monoidal groupoid of pairs \((g,q)\), with composition \((g_2,q_2)\circ(g_1,q_1)=(g_2\circ g_1,q_1+q_2\circ g_1)\).
Moving a diagonal layer through an affine layer performs the substitution \(q\mapsto q\circ g\).

Completeness is proved by the layered normal form \(A\circ D\), where \(A\) is affine and \(D\) records the coefficients of the relevant binomial-basis expansion.
Semantic uniqueness follows the same decomposition.
The relabelling fixes \(A\), and coefficient uniqueness fixes \(D\).
Circuits with the same affine relabelling and phase function are equal in the presentation.
Binomial coordinates keep the transport equations local and place the mixed phases forced by substitution directly in the normal-form coefficient vector.

Three questions remain.
The first asks which parts of the qubit CNOT-dihedral toolkit transfer to the prime-field setting.
Spider-nest identities for AND-type phases \cite{munson_and-gates_2021}, T-count reduction methods \cite{debeaudrap_fast_2020}, and phase-polynomial compilation techniques provide concrete qubit benchmarks for such a transfer.
In odd prime dimension, the relevant identities are expressed in finite-field binomial coordinates.

The second concerns analogous presentations at higher Clifford levels.
For qubits, higher-order CNOT-dihedral presentations \cite{amy_cnot_dihedral_2018} provide a comparison point for refined phase gates, up to global phase.
In prime dimension, \cref{fig:clifford-binomial-ladder} suggests a different coordinate sequence and two regimes.
For the \(d\)-th-root part, extending the diagonal normal forms from degree \(3\) to degrees \(r<d\) should mostly be a matter of adding the higher binomial coordinates \(\binom{x}{r}\) and the mixed coordinates produced by Pascal expansion under affine substitution.
The first finer-root level then starts with phases such as \(\zeta_2^x\).
At those precisions one would expect the same binomial pattern, together with compatibility equations between adjacent precisions, for instance \(d\) copies of a \(\zeta_m^{\binom{x}{r}}\)-phase recovering the corresponding \(\zeta_{m-1}^{\binom{x}{r}}\)-phase.
Since diagonal hierarchy levels depend on both denominator and degree \cite{cui_diagonal_2017}, a full presentation would still need stable phase coordinates, finite transport and commutation tables, and a normal-form uniqueness argument.

The third asks whether the rule sets can be reduced further.
The presentations used here are sufficient for normalisation and completeness.
The independence of their remaining primitive equations is not established here.
An independence analysis, or further derivations among the primitive equations, would identify which assumptions are forced by affine substitution and which can still be removed.

\section*{Acknowledgements}

This work is supported by the Plan France 2030 through the PEPR integrated project EPiQ (ANR-22-PETQ-0007) and the HQI platform (ANR-22-PNCQ-0002); by the European Union through the MSCA Staff Exchanges project QCOMICAL (Grant Agreement ID: 101182520); and by the Maison du Quantique MaQuEst.

\bibliographystyle{eptcs}
\bibliography{qupit-phase-poly}

\appendix
\section{Notation Reference}\label{app:notation-reference}

The following tables collect the notation used by the affine and phase calculi. The appendix proofs use the same conventions, so the relevant interpretation is recorded next to each symbol.

\begin{center}
\footnotesize
\renewcommand{\arraystretch}{1.04}
\setlength{\arrayrulewidth}{0.25pt}
\begin{tabularx}{\linewidth}{>{\raggedright\arraybackslash}p{0.28\linewidth}>{\raggedright\arraybackslash}X}
\hline
Symbol & Use in this work \\
\hline
\(\F_d\), \(\F_d^\times\), \(\omega\) &
The prime field, its nonzero scalars, and the primitive phase \(\omega=e^{2\pi i/d}\). Basis labels and phase exponents are computed in \(\F_d\); an exponent \(k\) denotes the scalar \(\omega^k\). \\
\hline
\(\Hilb_d\) &
The one-qupit Hilbert space \(\C^d\), with basis \(\{\ket{x}\mid x\in\F_d\}\). This is the concrete semantic space on which the affine and phase circuits act. \\
\hline
\(\GL_n(\F_d)\), \(\AGL_n(\F_d)\) &
Invertible linear maps and invertible affine maps \(x\mapsto Ax+b\), used for the reversible classical label update of an \(n\)-wire circuit. \\
\hline
\(\cat{Aff}_d\), \(\et{Aff}_d\) &
The diagrammatic PROP for affine circuits and its equational theory. This affine core is included in each phase fragment. \\
\hline
\(\cat{AffLab}_d\), \(\interp{\cdot}_{Aff}\) &
The semantic affine-label PROP and the interpretation of affine diagrams. The interpretation sends each affine circuit to the corresponding affine bijection of label sets, equivalently to a function \(x\mapsto Ax+b\). \\
\hline
\(\cat{PhAff}^{\Lin}_d\), \(\cat{PhAff}^{\Quad}_d\), \(\cat{PhAff}^{\Cube}_d\) &
The semantic phase-affine PROPs. An \(n\)-wire morphism has the form \(P_gD_q\), with \(g\in\AGL_n(\F_d)\) and \(q\) in the corresponding family \(\Lin_n\), \(\Quad_n\), or \(\Cube_n\). \\
\hline
\(D_q\), \(P_g\), \(P_gD_q\) &
The diagonal phase operator, affine permutation operator, and their phase-affine composite. Here \(D_q\ket{x}=\omega^{q(x)}\ket{x}\), \(P_g\ket{x}=\ket{g(x)}\), and \(P_gD_q\) applies the phase before relabelling the basis state. \\
\hline
\((g,q)\) &
Coordinate notation for \(P_gD_q\), used when an equality proof compares the affine update and the phase function separately. \\
\hline
\(\interp{\cdot}_{\mathcal{P}}\) &
The interpretation functor for an admissible phase fragment \(\mathcal{P}\). It assigns phase-affine operators and records both the affine update and the accumulated phase function. \\
\hline
\end{tabularx}
\captionof{table}{Affine and semantic notation in this work.}
\label{tab:notation-reference}
\end{center}

\begin{center}
\footnotesize
\renewcommand{\arraystretch}{1.04}
\setlength{\arrayrulewidth}{0.25pt}
\begin{tabularx}{\linewidth}{>{\raggedright\arraybackslash}p{0.28\linewidth}>{\raggedright\arraybackslash}X}
\hline
Symbol & Use in this work \\
\hline
\(\cat{LinPhase}_d\), \(\cat{QuadPhase}_d\),\newline \(\cat{CubicPhase}_d\) &
The diagrammatic calculi with linear, quadratic, and cubic diagonal phases. They are defined, respectively, for prime \(d\), odd prime \(d\), and prime \(d>3\). \\
\hline
\(\Lin_n\), \(\Quad_n\), \(\Cube_n\) &
Functions \(q:\F_d^n\to\F_d\) represented by polynomial functions of degree at most \(1\), \(2\), and \(3\). These are the phase-function spaces of the three fragments. \\
\hline
\(\mathcal{D}_{\Lin}\), \(\mathcal{D}_{\Quad}\), \(\mathcal{D}_{\Cube}\) &
The diagonal family sets of the phase fragments. They are the scalar, one-wire, and derived diagonal families available in the corresponding fragment. \\
\hline
\(q\circ g\) &
Precomposition of the phase function \(q\) by the affine update \(g\). This is the substitution operation behind the transport rules for moving phases through affine gates. \\
\hline
support &
The set of wires on which a placed diagonal family or affine generator acts. Transport and commutation statements quantify over compatible placements; identity wires and PROP symmetries supply the surrounding circuit context. \\
\hline
\(\NF_{\mathrm{aff}}(g)\) &
The unique affine normal form for \(g\in\AGL_n(\F_d)\), used as the affine layer inside the full normal form. \\
\hline
\(A\circ D\) &
A layered phase-affine normal form in which the diagonal phase layer \(D\) acts first, followed by the affine layer \(A\). Uniqueness separates into uniqueness of the affine map and uniqueness of the diagonal phase polynomial. \\
\hline
\(\mathcal{E}_{\mathrm{aff}}\) &
The affine generator families: translation, controlled addition, and nonzero multiplication by \(a\in\F_d^\times\). They generate the reversible affine updates used by the transport rules. \\
\hline
\end{tabularx}
\captionof{table}{Phase-fragment and normal-form notation in this work.}
\label{tab:notation-reference-phase}
\end{center}

\section{Phase-Polynomial Details}\label{app:phase-polynomial-details}

\subsection{Mixed Coordinates Under Affine Transport}\label{app:degree-fragments-transport}

\noindent\textbf{Transport through controlled addition.}
Translations and multipliers keep a one-wire phase on one wire.
Controlled addition is the affine generator that mixes variables.
Transporting an \(S\)-phase through it substitutes \(x+y\) into \(\binom{x}{2}\), and the identity \(\binom{x+y}{2}=\binom{x}{2}+\binom{y}{2}+xy\) produces the mixed coefficient \(xy\), represented by \(\gCZ\).
At degree three, \(\binom{x+y}{3}\) produces the two oriented coefficients \(\binom{x}{2}y\) and \(x\binom{y}{2}\), represented by \(\gSC\) and \(\gCS\).
The remaining square-free cubic coefficients \(xyz\) are represented by \(\gCCZ\).
Together with \cref{lem:monomial-to-binomial}, these families express the mixed monomial directions \(x^2y\) and \(xy^2\) in the cubic fragment.
The derived diagonal gates are therefore exactly the mixed coordinates forced by substitution, and \cref{prop:phase-layered-uniqueness} can compare diagonal normal forms by their finite-field coefficient vectors.

\subsection{Qubit CNOT-Dihedral Comparison}\label{app:qubit-cnot-dihedral-comparison}

\noindent\textbf{The qubit count.}
For the qubit CNOT-dihedral fragment generated by \(\{CNOT,T,X\}\), Amy et al.\ count the diagonal normal forms on \(n\) wires as \(8 \cdot 8^{\binom{n}{1}} \cdot 4^{\binom{n}{2}} \cdot 2^{\binom{n}{3}}\) \cite{amy_cnot_dihedral_2018}.
Equivalently, this number is \(2^{3+3\binom{n}{1}+2\binom{n}{2}+\binom{n}{3}} = 2^{\binom{n+3}{2}+\binom{n+1}{3}}\).
The scalar contributes \(8\) choices, and the moduli \(8\), \(4\), and \(2\) record degree-dependent coefficient periodicities for the one-, two-, and three-variable parity phases.

\medskip
\noindent\textbf{The parity basis.}
The qubit normal form uses parity arguments such as \(x_i\oplus x_j\) and \(x_i\oplus x_j\oplus x_k\).
These arguments are linear over \(\F_2\), while their coefficients are read in \(\Z_8\).
Expanding a parity as a function on bits therefore uses mixed arithmetic; for instance \(x_i\oplus x_j=x_i+x_j-2x_ix_j\) as a \(\Z_8\)-valued function.
The qubit normal form fixes this parity list together with degree-dependent coefficient ranges.

\medskip
\noindent\textbf{The odd-prime cubic basis.}
For prime \(d>3\), the cubic fragment keeps both basis labels and phase exponents in \(\F_d\).
The diagonal normal form has one scalar coefficient \(k\in\F_d\), giving the global phase \(\omega^k\), then \(3\binom{n}{1}\) one-wire coefficients for \(Z,S,T\), \(3\binom{n}{2}\) two-wire coefficients for \(CZ,CS,SC\), and \(\binom{n}{3}\) three-wire coefficients for \(CCZ\).
The number of diagonal normal forms is therefore \(d^{1+3\binom{n}{1}+3\binom{n}{2}+\binom{n}{3}} = d^{\binom{n+3}{3}}\).
A single finite-field coefficient vector gives this count, with no degree-dependent parity moduli.
Affine transport acts on that vector by substituting affine-linear \(\F_d\)-expressions.
The derived gates \(\gCZ\), \(\gCS\), \(\gSC\), and \(\gCCZ\) supply the mixed coordinates needed for \cref{prop:phase-layered-uniqueness} to read equality of phase functions as equality of coefficient vectors.

\section{Normal-Form Proofs}\label{app:deferred-proofs}

The proofs collected here verify soundness, normalisation, and uniqueness for the statements cited in the main proof.

\subsection{Soundness and Polynomial Checks}\label{app:proofs-soundness}

\begin{proof}[Proof of \cref{lem:Aff-soundness}]
It is enough to compare the affine functions on label tuples. The multiplier equations express the group laws \(a(bx)=(ab)x\), \(1x=x\), and the compatibility of multipliers with wire permutations. The translation rules express the identities \(x+0=x\), \(x+1+(-1)=x\), \(a(x+1)=ax+a\), and the corresponding identities after controlled addition. The controlled-addition and symmetry equations are the standard matrix identities for elementary transvections and block permutations over \(\F_d\). Hence each displayed affine axiom has two sides denoting the same function \(\F_d^n\to\F_d^n\).
\end{proof}

\begin{proof}[Proof of \cref{lem:binom-identities}]
All identities are polynomial identities over \(\F_d\), with \(2\) invertible. Expanding \(\binom{x}{2}=x(x-1)/2\) gives \(\binom{x+1}{2}=\binom{x}{2}+x\) and \(\binom{x+y}{2}=\binom{x}{2}+\binom{y}{2}+xy\). Multiplication by a scalar gives the remaining unary transport identity, since \(\binom{kx}{2}=k^2\binom{x}{2}+\binom{k}{2}x\).
\end{proof}

\begin{proof}[Proof of \cref{lem:binom3-identities}]
All identities are polynomial identities over \(\F_d\), with \(6\) invertible. Expanding \(\binom{x}{3}=x(x-1)(x-2)/6\) gives the translation identity for \(\binom{x+1}{3}\) and the binary identity for \(\binom{x+y}{3}\). The scalar identity follows by expanding \(\binom{kx}{3}\) and collecting the coefficients of \(\binom{x}{3}\), \(\binom{x}{2}\), and \(x\).
\end{proof}

\begin{proof}[Proof of \cref{lem:monomial-to-binomial}]
The first identity is the definition \(\binom{x}{2}=x(x-1)/2\), rewritten as \(x^2=2\binom{x}{2}+x\). Multiplying this identity by \(y\) gives \(x^2y=2y\binom{x}{2}+xy\), and exchanging \(x\) and \(y\) gives \(xy^2=2x\binom{y}{2}+xy\).
\end{proof}

\begin{restDerivedSemantics}
\stmtDerivedSemantics
\end{restDerivedSemantics}

\begin{proof}
The affine coordinate of each defining circuit is the identity, so only the phase exponent has to be computed. The gate \(\gCZ\) contributes \(\binom{x+y}{2}-\binom{x}{2}-\binom{y}{2}=xy\). For \(\gCS\), let \(t=2^{-1}\). Using the already computed \(\gCZ\) semantics, its exponent is \(txy-tx-\binom{x}{3}-\binom{x}{2}+t\binom{x+y}{3}-t\binom{y-x}{3}\), which expands to \(x\binom{y}{2}\). The gate \(\gSC\) is the wire flip of \(\gCS\), hence contributes \(y\binom{x}{2}\). Finally, the exponent of \(\gCCZ\) is the third finite difference \(\binom{x+y+z}{3}-\binom{x+y}{3}-\binom{x+z}{3}-\binom{y+z}{3}+\binom{x}{3}+\binom{y}{3}+\binom{z}{3}=xyz\). These are exactly the interpretations stated in \cref{lem:Derived-semantics}.
\end{proof}

\begin{proof}[Proof of \cref{lem:LinPhase-soundness}]
Compare both sides as pairs \((g,q)\), where \(g\) is the affine coordinate and \(q\) is the phase polynomial. The purely affine equations are sound by \cref{lem:Aff-soundness}. The order and commutation equations for diagonal gates use only addition of phase exponents in \(\F_d\). The derived-gate equations are sound by \cref{lem:Derived-semantics}. The transport equations are instances of \(D_qP_g=P_gD_{q\circ g}\). Translations, scalings, and controlled additions substitute affine expressions into the relevant binomial or monomial phase, and the required expansions are precisely \cref{lem:binom-identities,lem:binom3-identities,lem:monomial-to-binomial}. Every displayed phase axiom belonging to the chosen admissible fragment therefore has equal affine coordinate and equal phase exponent on both sides.
\end{proof}

\subsection{Affine Normalisation}\label{app:proofs-affine-normalisation}

\begin{restAffineNormalisation}
\stmtAffineNormalisation
\end{restAffineNormalisation}

\begin{proof}
Use the translation-linear rules to rewrite \(A\) so that all \(\gX\) factors are collected into a single translation column at the output side of the diagram.
The remaining part is a linear circuit built from \(\gCX\), \(\gM\), and permutations.

Let \(g=\interp{A}_{Aff}\in\AGL_n(\F_d)\) and write \(g(x)=Ax+b\).
Under \(\interp{\cdot}_{Aff}\), the collected translation column interprets as \(x\mapsto x+b\), and the remaining linear circuit interprets as \(x\mapsto Ax\).
Lafont's linear normal-form theorem gives a rewrite sequence, using the Lafont linear rules, from this linear circuit to the linear normal form of \(A\).
The finite certificates in Appendix~\ref{derivation:XX-comm} show that every Lafont linear rule used in that sequence is derivable in \(\et{Aff}_d\), so the same equality is derivable in our affine presentation.
Appending the translation column determined by \(b\) gives the affine normal form \(\NF_{\mathrm{aff}}(g)\), and hence \(\et{Aff}_d\vdash A=\NF_{\mathrm{aff}}(g)\).
In Appendix~\ref{sec:cnot-prime-d} we also give an alternative route to affine completeness using group-theoretic arguments.
\end{proof}

\subsection{Phase Normal Forms and Completeness}\label{app:proofs-phase-normal-forms}

Fix an admissible phase fragment \(\mathcal{P}\).
All circuits, equations, and normal forms in this subsection are internal to it.
The family set \(\mathcal{D}_{\mathcal{P}}\) is the one defined in \cref{def:phase-diagonal-families}; lower fragments are obtained by restricting to the rows and columns for the families they contain.
The scalar family \(\gW\) is included in \(\mathcal{D}_{\mathcal{P}}\), but its interaction cases are central and are handled separately in the proofs below.

\begin{restPhaseFiniteOrder}
\stmtPhaseFiniteOrder
\end{restPhaseFiniteOrder}

\begin{center}
\begingroup
\renewcommand{\arraystretch}{1.08}
\begin{tabular}{c c l}
family & first fragment & finite-order certificate \\
\hline
\(\gW\) & \(\Lin\) & \cref{ax:W-order-d} \\
\(\gZ\) & \(\Lin\) & \cref{ax:Z-order-d} \\
\(\gS\) & \(\Quad\) & \cref{ax:S-order-d} \\
\(\gT\) & \(\Cube\) & \cref{ax:T-order-d} \\
\(\gCZ\) & \(\Quad\) & \cref{ax:CZ-d} \\
\(\gCS\) & \(\Cube\) & \cref{ax:CS-order-d} \\
\(\gSC\) & \(\Cube\) & \cref{ax:SC-order-d} \\
\(\gCCZ\) & \(\Cube\) & \cref{ax:CCZ-d}
\end{tabular}
\captionof{table}{Finite-order certificates for diagonal families. Each row gives the displayed result used to cancel \(d\) consecutive copies of the family; placed instances follow by tensoring with identities and conjugating by symmetries.}
\label{tab:diagonal-family-finite-order}
\endgroup
\end{center}

\begin{proof}[Finite-order cancellation]
Fix a family \(G\in\mathcal{D}_{\mathcal{P}}\) and a support.
The corresponding row of \cref{tab:diagonal-family-finite-order} cancels \(d\) consecutive copies of \(G\).
All rewrites stay on the same support, and placed instances use the PROP laws for identities and symmetries.
\end{proof}

\begin{restPhaseDiagonalCommutation}
\stmtPhaseDiagonalCommutation
\end{restPhaseDiagonalCommutation}

\begin{proof}[Diagonal commutation]
For non-scalar factors, the finite local cases are indexed in \cref{tab:diag-comm-lemmas}.
The table is written for the cubic fragment; lower fragments are obtained by retaining only the families in \(\mathcal{D}_{\Lin}\) or \(\mathcal{D}_{\Quad}\).
Equal-family cases are tautological up to associativity, and scalar factors commute by scalar centrality in a monoidal category.
Disjoint-support and placed instances use only identities and symmetries.
\end{proof}

\begin{restPhaseOneStepTransport}
\stmtPhaseOneStepTransport
\end{restPhaseOneStepTransport}

\begin{proof}[One-step affine transport]
The finite local cases are indexed in \cref{tab:diag-through-affine}.
The table covers the three affine generator families in \(\mathcal{E}_{\mathrm{aff}}\) against the non-scalar members of \(\mathcal{D}_{\Cube}\), with lower fragments obtained by restriction.
Each entry is precisely the circuit form of the affine substitution \(q\mapsto q\circ g_A\), and placement in a larger diagram uses identities and symmetries.
\end{proof}

\begin{restPhaseDiagonalNormalisation}
\stmtPhaseDiagonalNormalisation
\end{restPhaseDiagonalNormalisation}

\begin{proof}[Diagonal-layer normalisation]
Move scalar \(\gW\) factors to the scalar position by scalar centrality.
For every remaining adjacent pair that is out of the fixed normal-form order, use \cref{lem:phase-diagonal-commutation}.
That table is written for the non-scalar members of \(\mathcal{D}_{\Cube}\); in the linear and quadratic fragments one restricts it to the families present in \(\mathcal{D}_{\Lin}\) or \(\mathcal{D}_{\Quad}\).
The resulting order consists of the scalar position, one-wire columns, two-wire stairs, and then three-wire stairs.
Finally delete every block of \(d\) equal placed generators with \cref{lem:phase-finite-order}; the remaining number of copies, between \(0\) and \(d-1\), is the exponent used in the diagonal normal form.
\end{proof}

\begin{restPhaseAffineLayerTransport}
\stmtPhaseAffineLayerTransport
\end{restPhaseAffineLayerTransport}

\begin{proof}[Affine-layer transport]
First take \(A\) to be one placed affine generator.
Expand \(D\) into repeated placed generators, omitting zero exponents.
Scalar copies commute with \(A\), and each non-scalar copy moves across \(A\) by \cref{lem:phase-one-step-transport}.
After the finite list of copies is transported, the right-hand side has the form \(A\circ D'\) with \(D'\) diagonal.
Context, disjoint-support, and structural instances use only the PROP laws, and \cref{lem:phase-diagonal-normalisation} puts \(D'\) in normal form.

Write the affine normal form \(A\) as a finite word in structural symmetries and factors of the form \(\sigma^{-1}\circ(E\otimes\id_{n-r})\circ\sigma\), where \(E:r\to r\) lies in \(\mathcal{E}_{\mathrm{aff}}\).
Induct along this word from the output side to the input side.
For a non-structural affine factor, use the placed-generator case; for a structural factor, relabel supports by the PROP symmetry laws and normalise if the displayed order changes.
After the final factor, all diagonal factors are on the input side, giving \(D\circ A=A\circ D'\) with \(D'\) in diagonal normal form.
The phase function changes by \(q\mapsto q\circ g_A\) because the one-step substitutions compose along the affine word.
\end{proof}

\begin{restPhaseLayeredFactorisation}
\stmtPhaseLayeredFactorisation
\end{restPhaseLayeredFactorisation}

\begin{proof}[Factorisation]
Each generator already has a layered form.
Affine and structural generators have the identity diagonal layer; diagonal generators have the identity affine layer and a one-factor diagonal layer.
Induct over a circuit word for \(C\).
The empty word gives the identity layered form.
For the induction step, suppose the current prefix has layered form \(A\circ D\), and write the next generator in layered form \(A'\circ D'\).
By associativity, the composite may be written as \(A\circ(D\circ A')\circ D'\).
The only obstruction to a layered form is the middle subcircuit \(D\circ A'\), where a diagonal layer sits on the output side of an affine layer.
\Cref{lem:phase-affine-layer-transport} rewrites this subcircuit as \(A'\circ D_1\), with \(D_1\) again diagonal normal.
The affine product \(A\circ A'\) is then put in affine normal form by \cref{lem:affine-normalisation}, and the diagonal product \(D_1\circ D'\) is put in diagonal normal form by \cref{lem:phase-diagonal-normalisation}.
\end{proof}

\begin{restPhaseLayeredUniqueness}
\stmtPhaseLayeredUniqueness
\end{restPhaseLayeredUniqueness}

\begin{proof}[Layered uniqueness]
Write the diagonal interpretations as \(D_{q_D}\) and \(D_{q_{D'}}\) in the corresponding semantic phase-affine PROP.
Then \(\interp{A\circ D}=P_{\interp{A}_{Aff}}D_{q_D}\) and \(\interp{A'\circ D'}=P_{\interp{A'}_{Aff}}D_{q_{D'}}\).
The coordinate uniqueness from \cref{rem:phase-affine-coordinates} gives \(\interp{A}_{Aff}=\interp{A'}_{Aff}\) and \(q_D=q_{D'}\).
The affine generators and equations embed unchanged into each phase fragment, so \cref{lem:normal-form-unique} gives \(A=A'\).
For the diagonal layer, the normal form records the coefficient list of \cref{lem:phase-coordinate-uniqueness}.
It contains \(1,x_i\) in the linear fragment; additionally \(\binom{x_i}{2}\) and \(x_ix_j\) for \(i<j\) in the quadratic fragment; and additionally \(\binom{x_i}{3}\), \(x_i\binom{x_j}{2}\) for \(i\neq j\), and \(x_ix_jx_k\) for \(i<j<k\) in the cubic fragment.
Since \(q_D=q_{D'}\), equality of phase functions forces equality of the scalar exponent, all one-wire column exponents, and all two- and three-wire stair exponents.
Therefore \(D=D'\).
\end{proof}

\begin{restPhaseNormalFormComplete}
\stmtPhaseNormalFormComplete
\end{restPhaseNormalFormComplete}

\begin{proof}[Completeness]
Use \cref{lem:phase-layered-factorisation} to rewrite \(C\) and \(C'\) as \(A\circ D\) and \(A'\circ D'\).
By soundness, the semantic hypothesis gives equal interpretations of these two normal forms.
\Cref{prop:phase-layered-uniqueness} gives \(A=A'\) and \(D=D'\).
The two original circuits therefore rewrite to the same layered normal form, hence \(\vdash C=C'\) in the fixed fragment.
\end{proof}

\section{Derivation Certificates and Proof Principles}\label{sec:derivations}

This appendix records the finite certificates used by the normal-form proofs.
They cover affine structural rewrites, diagonal commutation, and one-step diagonal transport through a placed affine generator.
Surrounding identities and symmetries are supplied by PROP laws; once a principle is certified, later derivations cite its label instead of repeating lower-level axiom instances.

\subsection{Affine Structural Certificates}\label{derivation:XX-comm}

These certificates supply the affine structural rewrites used in affine normalisation, including the comparison with Lafont's linear circuit rules.


\begin{figure}[htb!]
  \centering
    \fbox{%
      \begin{minipage}{\dimexpr\linewidth-2\fboxsep-2\fboxrule\relax}
        \begin{minipage}[t]{0.33\linewidth}
        \begin{equation}\label{lafont-axiom-multxy}
            \tikzfig{axioms/AxiomMultXMultY} = \tikzfig{axioms/AxiomMultXY}
          \end{equation}
          \begin{equation}\label{lafont-axiom-mult1}
            \tikzfig{axioms/AxiomMult1} = \tikzfig{axioms/AxiomPNullB}
          \end{equation}
          \begin{equation}\label{lafont-swapswap}
            \tikzfig{derivations/swapswap} = \tikzfig{two-lines}
          \end{equation}
          \begin{equation}\label{lafont-8}
            \tikzfig{axioms/d-cnot} = \tikzfig{two-lines}
          \end{equation}
        \end{minipage}\hfill
        \begin{minipage}[t]{0.65\linewidth}
          \begin{equation}\label{lafont-3}
            \tikzfig{derivations/lafont-3-left} = \tikzfig{derivations/lafont-3-right}
          \end{equation}
          \begin{equation}\label{lafont-4}
            \tikzfig{derivations/lafont-4-left} = \tikzfig{derivations/lafont-4-right}
          \end{equation}
          \begin{equation}\label{lafont-5}
            \tikzfig{derivations/lafont-5left} = \tikzfig{derivations/lafont-5right}
          \end{equation}
          \begin{equation}\label{lafont-6}
            \scalebox{0.9}{\tikzfig{derivations/lafont-6-left}} = \scalebox{0.9}{\tikzfig{derivations/lafont-6-right}}
          \end{equation}
          \begin{equation}\label{lafont-7}
            \scalebox{0.85}{\tikzfig{derivations/lafont-7-left}} = \scalebox{0.85}{\tikzfig{derivations/lafont-7-right}}
          \end{equation}
        \end{minipage}
      \end{minipage}%
    }
  \caption{The ruleset \cite{lafont_boolean_circuits_2003} for circuits of linear transformations over \(\GL(\mathrm{\F_d})\). We include \cref{lafont-8} because in Lafont's presentation the labels are in the field, which in $\F_d$ means taken modulo $d$.}
  \label{Lafont-pres}
\end{figure}

\begin{lemma}
    For every rule $C = C'$ in \cref{Lafont-pres}, $\et{Aff}_d \vdash C = C'$
\end{lemma}

\begin{proof}
Let \(C=C'\) be any equation in \cref{Lafont-pres}. We show \(\et{Aff}_d\vdash C=C'\) by a case analysis on the displayed rule.
\begin{enumerate}
\item For the rules labelled \eqref{lafont-axiom-multxy} and \eqref{lafont-axiom-mult1}, the statement is immediate: they appear verbatim among the axioms of \(\et{Aff}_d\), namely \cref{axiom-multxy,axiom-mult1}.
\item The rule labelled \eqref{lafont-3} is obtained from \cref{axiom-multcnot} by closure of derivability under tensoring with identities, composition, and symmetries in a PROP.
\item The rule labelled \eqref{lafont-4} is derived in the same way from \cref{axiom-multcnot2}.
\item The rule labelled \eqref{lafont-swapswap} is an instance of the symmetric monoidal axioms, since \(\sigma_{1,1}\circ\sigma_{1,1}=\id_2\).
\item The rule labelled \eqref{lafont-5} holds by expanding \(\gCXk\) as the \(k\)-fold composite of \(\gCX\) and using associativity of composition.
\item The rule labelled \eqref{lafont-6} is proved using \cref{axiom-swap} and \cref{axiom-d-cnot}.
\item The rule labelled \eqref{lafont-7} follows from \cref{axiom-X-change-wires} together with the PROP laws.
\item The rule labelled \eqref{lafont-8} is proved in \cref{axiom-d-cnot}.
\end{enumerate}
This covers all equations in \cref{Lafont-pres}, hence every rule \(C=C'\) in Lafont's presentation is derivable in \(\et{Aff}_d\).
\end{proof}

\begin{lemma}\label{axiom-Xd}
    $\mathsf{Aff}_d$ $\vdash$ 
    \tikzfig{axioms/AxiomXd} = \tikzfig{pieces/I}
\end{lemma}

\begin{proof}
    \begin{gather*}
        \tikzfig{axioms/AxiomXd}
        = \tikzfig{axioms/AxiomXd-00}
        \eqeqref{axiom-mult1} \tikzfig{axioms/AxiomXd-01}
        \eqeqref{axiom-XM-twisted} \tikzfig{axioms/AxiomMult1}
        \eqeqref{axiom-mult1} \tikzfig{pieces/I}
    \end{gather*}
\end{proof}

\begin{lemma}\label{axiom-XM}
    $\forall x \in \mathbb{F}_{d}^{\times}$,
    $\mathsf{Aff}_d$ $\vdash$ 
    \tikzfig{axioms/AxiomXMx} = \tikzfig{axioms/AxiomMxX}
\end{lemma}

\begin{proof}
    \begin{gather*}
        \tikzfig{axioms/AxiomXMx}
        \eqeqref{axiom-Xd} \tikzfig{axioms/AxiomXMx-00}
        = \tikzfig{axioms/AxiomXMx-01}
        \eqeqref{axiom-XM-twisted} \tikzfig{axioms/AxiomMxX}
    \end{gather*}
\end{proof}

\begin{lemma}\label{axiom-d-cnot}
    $\mathsf{Aff}_d$ $\vdash$ 
    \tikzfig{axioms/d-cnot} = \tikzfig{two-lines}
\end{lemma}

\begin{proof}
    \begin{gather*}
        \tikzfig{axioms/d-cnot}
        = \tikzfig{axioms/d-cnot-00}
        \eqeqref{axiom-mult1} \tikzfig{axioms/d-cnot-01}
        \eqeqref{axiom-multcnot2-twisted}\tikzfig{axioms/d-cnot-02}
        \eqeqref{axiom-mult1} \tikzfig{two-lines}
    \end{gather*}
\end{proof}

\begin{lemma}\label{axiom-multcnot2}
    $\forall x \in \mathbb{F}_{d}^{\times}$,
    $\mathsf{Aff}_d$ $\vdash$ 
    \tikzfig{axioms/AxiomMultCNOT2A} = \tikzfig{axioms/AxiomMultCNOT2B}
\end{lemma}

\begin{proof}
    \begin{gather*}
        \tikzfig{axioms/AxiomMultCNOT2A}
        \eqeqref{axiom-multcnot2-twisted} \tikzfig{axioms/AxiomMultCNOT2A-00}
        = \tikzfig{axioms/AxiomMultCNOT2A-01}
        \eqeqref{axiom-d-cnot} \tikzfig{axioms/AxiomMultCNOT2B}
    \end{gather*}
\end{proof}

\begin{lemma}\label{axiom-XCNOT}
    $\mathsf{Aff}_d$ $\vdash$ 
    \tikzfig{axioms/AxiomXCX} = \tikzfig{axioms/AxiomCXX}
\end{lemma}

\begin{proof}
    \begin{gather*}
        \tikzfig{axioms/AxiomXCX}
        \eqdeuxeqref{axiom-Xd}{axiom-d-cnot} \tikzfig{axioms/AxiomXCX-00}
        \eqeqref{axiom-XCNOTtop} \tikzfig{axioms/AxiomXCX-01}\\
        \eqdeuxeqref{axiom-mult1}{axiom-multxy} \tikzfig{axioms/AxiomXCX-02}
        \eqeqref{axiom-multcnot}\tikzfig{axioms/AxiomXCX-03}\\
        \eqeqref{axiom-XM} \tikzfig{axioms/AxiomXCX-04}
        \eqeqref{axiom-XCNOTtop} \tikzfig{axioms/AxiomXCX-05}\\
        \eqeqref{axiom-XM} \tikzfig{axioms/AxiomXCX-06}
        \eqeqref{axiom-multcnot} \tikzfig{axioms/AxiomXCX-07}\\
        \eqdeuxeqref{axiom-mult1}{axiom-multxy}  \tikzfig{axioms/AxiomXCX-08}
        \eqdeuxeqref{axiom-Xd}{axiom-d-cnot} \tikzfig{axioms/AxiomCXX}
    \end{gather*}
\end{proof}      

\begin{lemma}\label{axiom-X-change-wires}
    $\forall s,k \in \mathbb{F}_{d}^{\times}$,
    $\mathsf{Aff}_d$ $\vdash$ 
    \tikzfig{axioms/AxiomXChangeWiresA} = \tikzfig{axioms/AxiomXChangeWiresB}
\end{lemma}

\begin{proof}
    \begin{gather*}
        \tikzfig{axioms/AxiomXChangeWiresA}
        \eqtroiseqref{axiom-multxy}{axiom-multcnot}{axiom-multcnot2}  \tikzfig{derivations/axiom-X-change-wires-00}\\
        \eqeqref{axiom-d-cnot} \tikzfig{derivations/axiom-X-change-wires-01}
        \eqeqref{axiom-I} \tikzfig{derivations/axiom-X-change-wires-02}\\
        \eqtroiseqref{axiom-multxy}{axiom-multcnot}{axiom-multcnot2}  \tikzfig{axioms/AxiomXChangeWiresB}
    \end{gather*}
\end{proof}

\begin{lemma}\label{lemma-XX-comm}
    $\mathsf{Aff}_d$ $\vdash$ 
    \tikzfig{axioms/AxiomXXCommA} = \tikzfig{axioms/AxiomXXCommB}
\end{lemma}
						
\begin{proof}
    \begin{gather*}
        \tikzfig{axioms/AxiomXXCommA} 
        \eqeqref{axiom-d-cnot} \tikzfig{derivations/axiom-XX-comm-00}
        \eqeqref{axiom-X-change-wires} \tikzfig{derivations/axiom-XX-comm-01}\\
        \eqeqref{axiom-d-cnot} \tikzfig{derivations/axiom-XX-comm-02}
        \eqtroiseqref{axiom-multcnot2}{axiom-mult1}{axiom-multxy} \tikzfig{derivations/axiom-XX-comm-03}\\
        \eqeqref{axiom-X-change-wires} \tikzfig{derivations/axiom-XX-comm-04}
        \eqtroiseqref{axiom-multcnot2}{axiom-mult1}{axiom-multxy}  \tikzfig{derivations/axiom-XX-comm-05}
        \eqeqref{axiom-d-cnot} \tikzfig{axioms/AxiomXXCommB} 
    \end{gather*}
\end{proof}

\begin{lemma}\label{lemma-XX2-comm}
    $\mathsf{Aff}_d$ $\vdash$ 
    \tikzfig{axioms/AxiomXX2CommA} = \tikzfig{axioms/AxiomXX2CommB}
\end{lemma}

\begin{proof}
    \begin{gather*}
        \tikzfig{axioms/AxiomXX2CommA} 
        \eqeqref{axiom-X-change-wires} \tikzfig{derivations/axiom-XX2-comm-00}
        \eqeqref{axiom-d-cnot} \tikzfig{derivations/axiom-XX2-comm-01}\\
        \eqtroiseqref{axiom-multcnot2}{axiom-mult1}{axiom-multxy}  \tikzfig{derivations/axiom-XX2-comm-02}
        \eqeqref{axiom-X-change-wires} \tikzfig{derivations/axiom-XX2-comm-03}\\
        \eqdeuxeqref{axiom-mult1}{axiom-multxy}  \tikzfig{derivations/axiom-XX2-comm-04}
        \eqdeuxeqref{lemma-XX-comm}{axiom-d-cnot}\tikzfig{axioms/AxiomXX2CommB} 
    \end{gather*}
\end{proof}


\begin{lemma}\label{axiom-swap}
    $\forall k \in \mathbb{F}_{d}^{\times}$,
    $\mathsf{Aff}_d$ $\vdash$
    $\tikzfig{axioms/AxiomSwap} = \tikzfig{axioms/AxiomSwapDecomp}$
\end{lemma}

\begin{proof}
    \begin{gather*}
        \tikzfig{axioms/AxiomSwap} = \tikzfig{derivations/axiom-swap-00}
        \eqeqref{axiom-swap-fixed} \tikzfig{derivations/axiom-swap-01}\\
        \eqdeuxeqref{axiom-multcnot}{axiom-multcnot2}  \tikzfig{derivations/axiom-swap-02}
        \eqeqref{axiom-multxy} \tikzfig{axioms/AxiomSwapDecomp}
    \end{gather*}
\end{proof}

\begin{lemma}\label{axiom-torus-aux}
    $\forall a \in \mathbb{F}_{d}^{\times}$,
    $\mathsf{Aff}_d$ $\vdash$
    $\tikzfig{derivations/axiom-torus-aux-00} = \tikzfig{derivations/axiom-torus-aux-06}$
\end{lemma}

\begin{proof}
    \begin{gather*}
        \tikzfig{derivations/axiom-torus-aux-00}
        \eqdeuxeqref{axiom-multxy}{axiom-mult1} \tikzfig{derivations/axiom-torus-aux-01}\\
        \eqeqref{axiom-swap} \tikzfig{derivations/axiom-torus-aux-02}
        \eqdeuxeqref{axiom-multxy}{axiom-mult1}  \tikzfig{derivations/axiom-torus-aux-03}\\
        \eqeqref{axiom-swap} \tikzfig{derivations/axiom-torus-aux-04}
        \eqeqref{axiom-multxy} \tikzfig{derivations/axiom-torus-aux-05}
        = \tikzfig{derivations/axiom-torus-aux-06}
    \end{gather*}
\end{proof}

\begin{lemma}\label{axiom-torus-derived}
    $\forall a,b \in \mathbb{F}_{d}^{\times}$,
    $\mathsf{Aff}_d$ $\vdash$
    \begin{gather*}
    \tikzfig{derivations/axiom-torus-drv-00}
    = \tikzfig{derivations/axiom-torus-drv-03}
    \end{gather*}
\end{lemma}

\begin{proof}
    \begin{gather*}
        \tikzfig{derivations/axiom-torus-drv-00}\\
        \eqeqref{axiom-torus-aux} \tikzfig{derivations/axiom-torus-drv-01}
        \eqeqref{axiom-multxy} \tikzfig{derivations/axiom-torus-drv-02}
        \eqeqref{axiom-torus-aux}  \tikzfig{derivations/axiom-torus-drv-03}
    \end{gather*}
\end{proof}

\begin{lemma}\label{axiom-weyl-derived}
    $\mathsf{Aff}_d$ $\vdash$
    $\tikzfig{derivations/axiom-weyl-drv-00} = \tikzfig{derivations/axiom-weyl-drv-04}$
\end{lemma}

\begin{proof}
    \begin{gather*}
        \tikzfig{derivations/axiom-weyl-drv-00}
        \eqdeuxeqref{axiom-multxy}{axiom-mult1} \tikzfig{derivations/axiom-weyl-drv-01}\\
        \eqeqref{axiom-torus-aux} \tikzfig{derivations/axiom-weyl-drv-02}
        =  \tikzfig{derivations/axiom-weyl-drv-03}
        \eqtroiseqref{axiom-multxy}{axiom-mult1}{axiom-multcnot2}  \tikzfig{derivations/axiom-weyl-drv-04}
    \end{gather*}
\end{proof}

The phase certificates follow dependency order.
They first establish the linear and quadratic transport package, then the cubic diagonal-commutation package, and finally the remaining cubic one-step transport cases; the finite tables at the end index the local cases used by \cref{lem:phase-diagonal-commutation,lem:phase-one-step-transport}.

\subsection{Linear Transport Seeds}\label{derivation:linear-transport-seeds}

\begin{lemma}\label{ax:ZC-control}
    $\et{LinPhase}_d$ $\vdash$ 
    \tikzfig{./axioms/ZCNOTbot} = \tikzfig{./axioms/CNOTZbot}
\end{lemma}

\begin{proof}
    \begingroup\small
    \begin{gather*}
        \tikzfig{./axioms/ZCNOTbot} 
        \eqeqref{ax:Z-order-d} \tikzfig{./axioms/ZCNOTbot-00} 
        \eqeqref{ax:ZC-target} \tikzfig{./axioms/ZCNOTbot-01} \\
        \eqdeuxeqref{ax:Z-order-d}{axiom-d-cnot} \tikzfig{./axioms/ZCNOTbot-02}
        \eqdeuxeqref{axiom-mult1}{axiom-multxy} \tikzfig{./axioms/ZCNOTbot-03} \\
        \eqdeuxeqref{ax:ZM}{ax:Z-order-d} \tikzfig{./axioms/ZCNOTbot-04}
        \eqeqref{axiom-multcnot} \tikzfig{./axioms/ZCNOTbot-05}\\
        \eqeqref{ax:ZC-target} \tikzfig{./axioms/ZCNOTbot-06} 
        \eqtroiseqref{axiom-multcnot}{axiom-mult1}{axiom-multxy} \tikzfig{./axioms/ZCNOTbot-07} \\
        \eqdeuxeqref{ax:Z-order-d}{axiom-d-cnot} \tikzfig{./axioms/CNOTZbot}
    \end{gather*}
\end{proof}


\subsection{Quadratic Commutation and Transport Seeds}\label{derivation:quadratic-seeds}


\begin{lemma}\label{ax:SC-target}
    QuadPhase$_d$ $\vdash$ 
    \tikzfig{./axioms/SCNOTtop} = \tikzfig{./axioms/CNOTStop}
\end{lemma}

\begin{proof}
    \begin{gather*}
        \tikzfig{./axioms/SCNOTtop} 
        = \tikzfig{./axioms/SCNOTtop-00} 
        \eqeqref{ax:S-order-d} \tikzfig{./axioms/SCNOTtop-01}
        \eqeqref{axiom-d-cnot} \tikzfig{./axioms/CNOTStop}
    \end{gather*}
\end{proof}

\begin{lemma}\label{ax:SCZ}
    QuadPhase$_d$ $\vdash$ 
    \tikzfig{./axioms/SCZ} = \tikzfig{./axioms/CZS}
\end{lemma}

\begin{proof}
    \begin{gather*}
        \tikzfig{./axioms/SCZ}
        = \tikzfig{./axioms/SCZ-00} 
        \eqeqref{ax:SC-control} \tikzfig{./axioms/SCZ-01} 
        = \tikzfig{./axioms/CZS}
    \end{gather*}
\end{proof}

\begin{lemma}\label{ax:ZCZ}
    QuadPhase$_d$ $\vdash$ 
    \tikzfig{./axioms/CZZ} = \tikzfig{./axioms/ZCZ}
\end{lemma}

\begin{proof}
    \begin{gather*}
        \tikzfig{./axioms/CZZ}
        = \tikzfig{./axioms/ZCZ-00} 
        \eqeqref{ax:ZC-control} \tikzfig{./axioms/ZCZ-07} 
        \eqeqref{ax:SZ} \tikzfig{./axioms/ZCZ}
    \end{gather*}
\end{proof}

\begin{lemma}\label{ax:PhaseGadgetS}
    QuadPhase$_d$ $\vdash$ 
    \tikzfig{./axioms/PhaseGadgetA} = \tikzfig{./axioms/PhaseGadgetB}
\end{lemma}

\begin{proof}
    \begin{gather*}
        \tikzfig{./axioms/PhaseGadgetA}
        \eqeqref{axiom-d-cnot} \tikzfig{./axioms/PhaseGadgetA-00} \\
        \eqdeuxeqref{axiom-mult1}{axiom-multxy} \tikzfig{./axioms/PhaseGadgetA-01} \\
        \eqeqref{axiom-swap} \tikzfig{./axioms/PhaseGadgetA-02} 
        = \tikzfig{./axioms/PhaseGadgetA-03} \\
        \eqdeuxeqref{axiom-mult1}{axiom-multxy} \tikzfig{./axioms/PhaseGadgetA-04} 
        \eqeqref{ax:SC-control} \tikzfig{./axioms/PhaseGadgetA-05} 
        \eqeqref{axiom-d-cnot} \tikzfig{./axioms/PhaseGadgetB}
    \end{gather*}
\end{proof}

\begin{lemma}\label{ax:CZ-SWAP}
    QuadPhase$_d$ $\vdash$ 
    \tikzfig{./axioms/CZSWAP} = \tikzfig{./axioms/SWAPCZ}
\end{lemma}

\begin{proof}
    \begin{gather*}
        \tikzfig{./axioms/SWAPCZ}
        = \tikzfig{./axioms/CZSWAP-00}
        = \tikzfig{./axioms/CZSWAP-01}\\
        \eqeqref{ax:PhaseGadgetS} \tikzfig{./axioms/CZSWAP-02}
        = \tikzfig{./axioms/CZSWAP}
    \end{gather*}
\end{proof}

\begin{lemma}\label{ax:CZsum}
QuadPhase$_d$ $\vdash$ 

          \tikzfig{./axioms/CZsumL} = \tikzfig{./axioms/CZsumR}
\end{lemma}

\begin{proof}
    \begin{gather*}
        \tikzfig{./axioms/CZsumL} 
        \eqdeuxeqref{ax:SC-control}{ax:S-order-d} \tikzfig{./axioms/CZsumL-00}\\
        \eqeqref{ax:PhaseGadgetS} \tikzfig{./axioms/CZsumL-01}
        \eqdeuxeqref{ax:SC-control}{ax:S-order-d}\tikzfig{./axioms/CZsumL-02}\\
        \eqeqref{ax:PhaseGadgetS}\tikzfig{./axioms/CZsumL-03}
        \eqeqref{axiom-d-cnot} \tikzfig{./axioms/CZsumR}
    \end{gather*}
\end{proof}

\begin{corollary}\label{ax:CZexp}
    QuadPhase$_d$ $\vdash$ 
          \tikzfig{./pieces/CZk} = \tikzfig{./axioms/CZkexp}
\end{corollary}

\begin{lemma}\label{ax:CZ-d}
    QuadPhase$_d$ $\vdash$ 
    \tikzfig{./pieces/CZd} = \tikzfig{two-lines}
\end{lemma}

\begin{proof}
    \begin{gather*}
        \tikzfig{./pieces/CZd}
        \eqeqref{ax:CZexp}  \tikzfig{./axioms/CZd-03}
        \eqeqref{ax:S-order-d}  \tikzfig{./axioms/CZd-02}
        \eqeqref{axiom-d-cnot} \tikzfig{two-lines}
    \end{gather*}
\end{proof}


\medskip\noindent\textbf{Quadratic affine-transport seeds.}

\begin{lemma}\label{ax:CZ-X}
    QuadPhase$_d$ $\vdash$ 
    \tikzfig{./axioms/CZX} = \tikzfig{./axioms/XCZ}
\end{lemma}

\begin{proof}
    \begin{gather*}
        \tikzfig{./axioms/CZX}
        = \tikzfig{./axioms/CZX-00}\\
        \eqeqref{axiom-XCNOT} \tikzfig{./axioms/CZX-01}
        \eqeqref{ax:SX} \tikzfig{./axioms/CZX-02}\\
        \eqeqref{axiom-XCNOT} \tikzfig{./axioms/CZX-03}
        \eqdeuxeqref{ax:SX}{ax:SZ} \tikzfig{./axioms/CZX-04}\\
        \eqdeuxeqref{axiom-mult1}{axiom-multxy} \tikzfig{./axioms/CZX-05}
        \eqeqref{axiom-multcnot} \tikzfig{./axioms/CZX-06}\\
        \eqeqref{ax:ZC-target}  \tikzfig{./axioms/CZX-07}
        \eqeqref{ax:Z-order-d} \tikzfig{./axioms/CZX-08}\\
        \eqdeuxeqref{ax:ZM}{ax:Z-order-d} \tikzfig{./axioms/CZX-09}
        \eqeqref{axiom-multcnot} \tikzfig{./axioms/CZX-10}\\
        \eqdeuxeqref{axiom-mult1}{axiom-multxy} \tikzfig{./axioms/CZX-11}
        \eqeqref{ax:ZC-control} \tikzfig{./axioms/CZX-12}\\
        \eqeqref{ax:SZ} \tikzfig{./axioms/CZX-13}
        = \tikzfig{./axioms/XCZ}
    \end{gather*}
\end{proof}

\begin{lemma}\label{ax:CZ-CX}
    QuadPhase$_d$ $\vdash$ 
          \tikzfig{./axioms/CZCNOT} = \tikzfig{./axioms/CNOTCZ}
\end{lemma}

\begin{proof}
    \begin{gather*}
        \tikzfig{./axioms/CZCNOT} 
        = \tikzfig{./axioms/CZCNOT-00}\\
        \eqdeuxeqref{ax:SZ}{ax:S-order-d} \tikzfig{./axioms/CZCNOT-01}
        \eqdeuxeqref{axiom-mult1}{axiom-multxy} \tikzfig{./axioms/CZCNOT-02}\\
        \eqeqref{ax:SM} \tikzfig{./axioms/CZCNOT-03}
        \eqeqref{axiom-multcnot}\tikzfig{./axioms/CZCNOT-04}\\
        \eqeqref{ax:SC-control} \tikzfig{./axioms/CZCNOT-05}
        \eqeqref{ax:PhaseGadgetS}\tikzfig{./axioms/CZCNOT-06}\\
        \eqeqref{ax:SC-control} \tikzfig{./axioms/CZCNOT-07}
        \eqeqref{ax:PhaseGadgetS} \tikzfig{./axioms/CZCNOT-08}\\
        \eqeqref{ax:CZexp}\tikzfig{./axioms/CZCNOT-09}
        \eqeqref{ax:CZ-neg} \tikzfig{./axioms/CNOTCZ}
    \end{gather*}
\end{proof}

\begin{lemma}\label{ax:CZ-CX-power}
    For every \(r\in\F_d\), QuadPhase$_d$ $\vdash$
    \tikzfig{./axioms/CZCNOTpowL} = \tikzfig{./axioms/CZCNOTpowR}
\end{lemma}

\begin{proof}
The case \(r=0\) is the unit law.  For the step \(r\mapsto r+1\), expand
the last controlled phase and use the induction hypothesis:
\begingroup
\newcommand{\czpowfig}[1]{\scalebox{1}{\tikzfig{#1}}}
\begin{gather*}
    \czpowfig{./axioms/CZCNOTpowStep-00}
    \ih
    \czpowfig{./axioms/CZCNOTpowStep-01}\\
    \eqeqref{ax:CZ-CX}
    \czpowfig{./axioms/CZCNOTpowStep-02}\\
    \eqquatreeqref{ax:CZsum}{ax:SCZ}{ax:ZCZ}{ax:SZ}
    \czpowfig{./axioms/CZCNOTpowStep-03}.
\end{gather*}
\endgroup
\end{proof}

\begin{lemma}\label{ax:CZ-S-gadget}
    QuadPhase$_d$ $\vdash$
    \tikzfig{./axioms/CZSgadgetL} = \tikzfig{./axioms/CZSgadgetR}
\end{lemma}

\begin{proof}
Let \(B_r=\binom{r}{2}\).  The case \(r=0\) is immediate.  After
splitting the \((r+1)\)-weighted CNOTs into weights \(1,r,-1\), the
induction step is:
\begingroup
\newcommand{\czsfig}[1]{\scalebox{1}{\tikzfig{#1}}}
\begin{gather*}
    \czsfig{./axioms/CZSgadgetStep-IH}
    \ih
    \czsfig{./axioms/CZSgadgetStep-00}\\
    \eqeqref{ax:PhaseGadgetS}
    \czsfig{./axioms/CZSgadgetStep-01}\\
    \eqdeuxeqref{ax:SC-control}{ax:ZC-control}
    \czsfig{./axioms/CZSgadgetStep-02}\\
    \eqdeuxeqref{ax:CZ-CX-power}{axiom-d-cnot}
    \czsfig{./axioms/CZSgadgetStep-03}\\
    \eqquatreeqref{ax:CZsum}{ax:SCZ}{ax:ZCZ}{ax:SZ}
    \czsfig{./axioms/CZSgadgetStep-04}.
\end{gather*}
\endgroup
Since \(B_r+r=\binom{r+1}{2}\), the last figure is the right-hand side
with \(x=r+1\).
\end{proof}

\begin{lemma}\label{ax:CZ-M}
    QuadPhase$_d$ $\vdash$
    \tikzfig{./axioms/CZM} = \tikzfig{./axioms/MCZ}
\end{lemma}

\begin{proof}
Let \(b=\binom{x}{2}\).
\begingroup
\newcommand{\czmfig}[1]{\scalebox{1}{\tikzfig{#1}}}
\begin{gather*}
    \czmfig{./axioms/CZM}
    \eqdeuxeqref{axiom-mult1}{axiom-multxy}
    \czmfig{./axioms/CZM-00}\\
    \eqeqref{ax:CZexp}
    \czmfig{./axioms/CZM-01}\\
    \eqeqref{ax:SM}
    \czmfig{./axioms/CZM-02B}\\
    \eqdeuxeqref{ax:PhaseGadgetS}{axiom-multcnot}
    \czmfig{./axioms/CZM-03}\\
    \eqeqref{ax:CZ-S-gadget}
    \czmfig{./axioms/CZM-04}\\
    \eqtroiseqref{ax:SCZ}{ax:ZCZ}{ax:SZ}
    \czmfig{./axioms/MCZ}
    \end{gather*}
\endgroup
\end{proof}

\begin{lemma}\label{ax:CZCZ}
    QuadPhase$_d$ $\vdash$
    \tikzfig{./axioms/CZCZA} = \tikzfig{./axioms/CZCZB}
\end{lemma}

\begin{proof}
    \begin{gather*}
        \tikzfig{./axioms/CZCZA}
        \eqeqref{axiom-d-cnot} \tikzfig{./axioms/CZCZA-00}
        \eqeqref{ax:CC-control} \tikzfig{./axioms/CZCZA-01}\\
        \eqeqref{ax:CZ-d} \tikzfig{./axioms/CZCZA-02}
        \eqdeuxeqref{axiom-mult1}{axiom-multxy} \tikzfig{./axioms/CZCZA-03}\\
        \eqdeuxeqref{ax:CZ-d}{ax:CZ-M} \tikzfig{./axioms/CZCZA-04}
        \eqeqref{ax:CC-control} \tikzfig{./axioms/CZCZA-05}\\
        \eqdeuxeqref{ax:CZ-d}{ax:CZ-M} \tikzfig{./axioms/CZCZA-06}
        \eqdeuxeqref{axiom-mult1}{axiom-multxy} \tikzfig{./axioms/CZCZA-07}\\
        \eqeqref{ax:CZ-d} \tikzfig{./axioms/CZCZA-08}
        \eqeqref{axiom-d-cnot} \tikzfig{./axioms/CZCZB}
    \end{gather*}
\end{proof}

\begin{lemma}\label{ax:CC-target}
    QuadPhase$_d$ $\vdash$
    \tikzfig{./axioms/CZCNOTtop} = \tikzfig{./axioms/CNOTCZtop}
\end{lemma}

\begin{proof}
    \begin{gather*}
        \tikzfig{./axioms/CZCNOTtop}
        \eqeqref{ax:CZ-d} \tikzfig{./axioms/CZCNOTtop-00}\\
        \eqeqref{ax:CC-control} \tikzfig{./axioms/CZCNOTtop-01}
        \eqeqref{axiom-d-cnot} \tikzfig{./axioms/CZCNOTtop-02}\\
        \eqdeuxeqref{axiom-mult1}{axiom-multxy} \tikzfig{./axioms/CZCNOTtop-03}
        \eqeqref{ax:CZ-M} \tikzfig{./axioms/CZCNOTtop-04}\\
        \eqeqref{axiom-multcnot2} \tikzfig{./axioms/CZCNOTtop-05}
        \eqeqref{ax:CZCZ} \tikzfig{./axioms/CZCNOTtop-06}\\
        \eqeqref{ax:CC-control} \tikzfig{./axioms/CZCNOTtop-07}
        \eqeqref{ax:CZ-M} \tikzfig{./axioms/CZCNOTtop-08}\\
        \eqeqref{axiom-multcnot2} \tikzfig{./axioms/CZCNOTtop-09}
        \eqdeuxeqref{axiom-mult1}{axiom-multxy} \tikzfig{./axioms/CZCNOTtop-10}\\
        \eqeqref{ax:CZ-d} \tikzfig{./axioms/CZCNOTtop-11}
        \eqeqref{axiom-d-cnot} \tikzfig{./axioms/CNOTCZtop}
    \end{gather*}
\end{proof}


\medskip\noindent\textbf{Quadratic transport package.}

\begin{lemma}\label{lem:quad-diag-affine-transport}
In \(\cat{QuadPhase}_d\), every non-scalar family in \(\mathcal{D}_{\Quad}\) has all basic affine transport rules.
Each transports through \(\gX\), \(\gCX\), and \(\gM\), for every compatible choice of local supports.
\end{lemma}

\begin{proof}
The one-wire entries are \eqref{ax:ZX}, \eqref{ax:SX}, \eqref{ax:ZC-target}, \eqref{ax:ZC-control}, \eqref{ax:SC-target}, \eqref{ax:SC-control}, \eqref{ax:ZM}, and \eqref{ax:SM}.
The \(\gCZ\) entries are \eqref{ax:CZ-X}, \eqref{ax:CZ-CX}, and \eqref{ax:CZ-M}; symmetry and disjoint-support instances use the PROP laws.
\end{proof}

\subsection{Cubic Diagonal Interaction Certificates}\label{derivation:cubic-diagonal-interactions}

These certificates begin the cubic diagonal-commutation package, covering one-wire/controlled overlaps, finite-order facts for mixed cubic gates, and symmetry reductions for \(\gCCZ\).


\begin{lemma}\label{ax:TCZ}
    CubicPhase$_d$ $\vdash$ 
    \tikzfig{./axioms/TCZ} = \tikzfig{./axioms/CZT}
\end{lemma}

\begin{proof}
    \begin{gather*}
        \tikzfig{./axioms/TCZ}
        = \tikzfig{./axioms/TCZ-00}
        \eqeqref{ax:TS} \tikzfig{./axioms/TCZ-01}\\
        \eqeqref{ax:TC-control} \tikzfig{./axioms/TCZ-02}
        = \tikzfig{./axioms/CZT}
    \end{gather*}
\end{proof}

\begin{lemma}\label{ax:PhaseGadgetT}
    CubicPhase$_d$ $\vdash$ 
    \tikzfig{./axioms/PhaseGadgetTleftA} = \tikzfig{./axioms/PhaseGadgetTleftB}
\end{lemma}

\begin{proof}
    \begin{gather*}
        \tikzfig{./axioms/PhaseGadgetTleftA}\\
        \eqeqref{axiom-d-cnot} \tikzfig{./axioms/PhaseGadgetTleftA-00} \\
        \eqdeuxeqref{axiom-mult1}{axiom-multxy} \tikzfig{./axioms/PhaseGadgetTleftA-01} \\
        \eqeqref{axiom-swap} \tikzfig{./axioms/PhaseGadgetTleftA-02} \\
        = \tikzfig{./axioms/PhaseGadgetTleftA-03} \\
        \eqdeuxeqref{axiom-mult1}{axiom-multxy} \tikzfig{./axioms/PhaseGadgetTleftA-04} 
        \eqeqref{axiom-multcnot} \tikzfig{./axioms/PhaseGadgetTleftA-05}\\
        \eqeqref{ax:TC-control}  \tikzfig{./axioms/PhaseGadgetTleftA-06}
        \eqeqref{axiom-d-cnot} \tikzfig{./axioms/PhaseGadgetTleftB}
    \end{gather*}
\end{proof}

\begin{lemma}\label{ax:PhaseGadgetComm}
    CubicPhase$_d$ $\vdash$ 

    \tikzfig{./axioms/PhaseGadgetComm}
    = \tikzfig{./axioms/PhaseGadgetCommR}
\end{lemma}

\begin{proof}
    \begin{gather*}
        \tikzfig{./axioms/PhaseGadgetComm}
        = \tikzfig{./axioms/PhaseGadgetComm-00} \\
        \eqeqref{axiom-d-cnot} \tikzfig{./axioms/PhaseGadgetComm-01} 
        = \tikzfig{./axioms/PhaseGadgetComm-02} \\
        \eqeqref{ax:PhaseGadgetT} \tikzfig{./axioms/PhaseGadgetComm-03} \\
        \eqdeuxeqref{ax:SC-control}{ax:TC-control} \tikzfig{./axioms/PhaseGadgetComm-04}\\
        \eqeqref{ax:PhaseGadgetT} \tikzfig{./axioms/PhaseGadgetComm-05} 
        \eqeqref{axiom-d-cnot} \tikzfig{./axioms/PhaseGadgetCommR}
    \end{gather*}
\end{proof}

For the \(T\)-fragment certificates from \cref{ax:TCStop} onward, fix \(t=2^{-1}\in\F_d\), the parameter used in the definition of the derived gate \(\gCS\).

\begin{lemma}\label{ax:TCStop}
    CubicPhase$_d$ $\vdash$ 
    \tikzfig{./axioms/TCStop} = \tikzfig{./axioms/CSTtop}
\end{lemma}

\begin{proof}
    \begin{gather*}
        \tikzfig{./axioms/TCStop}
        = \tikzfig{./axioms/TCStop-00} \\
        \eqeqref{ax:TCZ} \tikzfig{./axioms/TCStop-01} \\
        \eqdeuxeqref{ax:TS}{ax:TZ} \tikzfig{./axioms/TCStop-02} \\
        \eqeqref{ax:TC-control} \tikzfig{./axioms/TCStop-03}
        = \tikzfig{./axioms/CSTtop}
    \end{gather*}
\end{proof}

\begin{lemma}\label{ax:SCStop}
    CubicPhase$_d$ $\vdash$ 
    \tikzfig{./axioms/SCStop} = \tikzfig{./axioms/CSStop}
\end{lemma}

\begin{proof}
    \begin{gather*}
        \tikzfig{./axioms/SCStop}
        = \tikzfig{./axioms/SCStop-00} \\
        \eqeqref{ax:SCZ} \tikzfig{./axioms/SCStop-01} \\
        \eqdeuxeqref{ax:TS}{ax:SZ} \tikzfig{./axioms/SCStop-02} \\
        \eqeqref{ax:SC-control} \tikzfig{./axioms/SCStop-03} 
        = \tikzfig{./axioms/CSStop}
    \end{gather*}
\end{proof}

\begin{lemma}\label{ax:ZCStop}
    CubicPhase$_d$ $\vdash$ 
    \tikzfig{./axioms/ZCStop} = \tikzfig{./axioms/CSZtop}
\end{lemma}

\begin{proof}
    \begin{gather*}
        \tikzfig{./axioms/ZCStop}
        = \tikzfig{./axioms/ZCStop-00} \\
        \eqeqref{ax:ZCZ} \tikzfig{./axioms/ZCStop-01}\\ 
        \eqdeuxeqref{ax:SZ}{ax:TZ} \tikzfig{./axioms/ZCStop-02} \\
        \eqeqref{ax:ZC-control} \tikzfig{./axioms/ZCStop-03}\
        = \tikzfig{./axioms/CSZtop}
    \end{gather*}
\end{proof}

\begin{lemma}\label{ax:TCSbot}
    CubicPhase$_d$ $\vdash$ 
    \tikzfig{./axioms/TCSbot} = \tikzfig{./axioms/CSTbot}
\end{lemma}

\begin{proof}
    \begin{gather*}
        \tikzfig{./axioms/TCSbot}
        = \tikzfig{./axioms/TCSbot-00} \\
        \eqeqref{ax:PhaseGadgetT} \tikzfig{./axioms/TCSbot-01} \\
        \eqeqref{ax:TCZ} \tikzfig{./axioms/TCSbot-02} \\
        \eqeqref{ax:TC-control} \tikzfig{./axioms/TCSbot-03} \\
        \eqeqref{ax:PhaseGadgetT} \tikzfig{./axioms/TCSbot-04} 
        = \tikzfig{./axioms/CSTbot}
    \end{gather*}
\end{proof}

\begin{lemma}\label{ax:SCSbot}
    CubicPhase$_d$ $\vdash$ 
    \tikzfig{./axioms/SCSbot} = \tikzfig{./axioms/CSSbot}
\end{lemma}

\begin{proof}
    \begin{gather*}
        \tikzfig{./axioms/SCSbot}
        = \tikzfig{./axioms/SCSbot-00} \\
        \eqeqref{ax:PhaseGadgetT} \tikzfig{./axioms/SCSbot-01} \\
        \eqeqref{ax:SCZ} \tikzfig{./axioms/SCSbot-02} \\
        \eqeqref{ax:SC-control} \tikzfig{./axioms/SCSbot-03} \\
        \eqeqref{ax:PhaseGadgetT} \tikzfig{./axioms/SCSbot-04} 
        = \tikzfig{./axioms/CSSbot}
    \end{gather*}
\end{proof}

\begin{lemma}\label{ax:ZCSbot}
    CubicPhase$_d$ $\vdash$ 
    \tikzfig{./axioms/ZCSbot} = \tikzfig{./axioms/CSZbot}
\end{lemma}

\begin{proof}
    \begin{gather*}
        \tikzfig{./axioms/ZCSbot}
        = \tikzfig{./axioms/ZCSbot-00} \\
        \eqeqref{ax:PhaseGadgetT} \tikzfig{./axioms/ZCSbot-01} \\
        \eqeqref{ax:ZCZ} \tikzfig{./axioms/ZCSbot-02}\\ 
        \eqeqref{ax:ZC-control} \tikzfig{./axioms/ZCSbot-03} \\
        \eqeqref{ax:PhaseGadgetT} \tikzfig{./axioms/ZCSbot-04} 
        = \tikzfig{./axioms/CSZbot}
    \end{gather*}
\end{proof}


\medskip\noindent\textbf{Oriented \texorpdfstring{\(CS/SC\)}{CS/SC} sorting and support changes.}

\begin{lemma}\label{ax:CS-CNOT-top}
    CubicPhase$_d$ $\vdash$ 
    \tikzfig{./axioms/CNOTCStop} = \tikzfig{./axioms/CSCNOTtop}
\end{lemma}

\begin{proof}
    \begin{gather*}
        \tikzfig{./axioms/CNOTCStop}
        = \tikzfig{./axioms/CNOTCStop-00}\\
        \eqeqref{lemma-XX2-comm} \tikzfig{./axioms/CNOTCStop-01} \\
        \eqtroiseqref{ax:TC-control}{ax:SC-control}{ax:ZC-control}   \tikzfig{./axioms/CNOTCStop-02} \\
        \eqeqref{ax:CC-target} \tikzfig{./axioms/CNOTCStop-03}
        = \tikzfig{./axioms/CSCNOTtop}
    \end{gather*}
\end{proof}

\begin{lemma}\label{ax:CS-CNOT-top2}
    CubicPhase$_d$ $\vdash$ 
    \tikzfig{./axioms/CNOTCStop2} = \tikzfig{./axioms/CSCNOTtop2}
\end{lemma}

\begin{proof}
    \begin{gather*}
        \tikzfig{./axioms/CNOTCStop2}
        = \tikzfig{./axioms/CNOTCStop2-00}\\
        \eqeqref{ax:PhaseGadgetT} \tikzfig{./axioms/CNOTCStop2-01}\\
        \eqeqref{lemma-XX2-comm} \tikzfig{./axioms/CNOTCStop2-02} \\
        \eqeqref{ax:CC-target}   \tikzfig{./axioms/CNOTCStop2-03} \\
        \eqeqref{ax:PhaseGadgetT} \tikzfig{./axioms/CNOTCStop2-04}
        = \tikzfig{./axioms/CSCNOTtop2}
    \end{gather*}
\end{proof}

\begin{lemma}\label{ax:CSCZbot}
    CubicPhase$_d$ $\vdash$ 
    \tikzfig{./axioms/CSCZbot} = \tikzfig{./axioms/CZCSbot}
\end{lemma}

\begin{proof}
    \begin{gather*}
        \tikzfig{./axioms/CSCZbot} 
        = \tikzfig{./axioms/CSCZbot-00} 
        \eqeqref{ax:PhaseGadgetS} \tikzfig{./axioms/CSCZbot-01} \\
        \eqeqref{ax:CS-CNOT-top} \tikzfig{./axioms/CSCZbot-02} 
        \eqeqref{ax:SCStop} \tikzfig{./axioms/CSCZbot-03} 
         \eqeqref{ax:PhaseGadgetS} \tikzfig{./axioms/CZCSbot}
    \end{gather*}
\end{proof}

\begin{lemma}\label{ax:CSCZtop}
    CubicPhase$_d$ $\vdash$ 
    \tikzfig{./axioms/CSCZtop} = \tikzfig{./axioms/CZCStop}
\end{lemma}

\begin{proof}
    \begin{gather*}
        \tikzfig{./axioms/CSCZtop} 
        = \tikzfig{./axioms/CSCZtop-00} \\
        \eqeqref{ax:CS-CNOT-top} \tikzfig{./axioms/CSCZtop-01} 
        \eqeqref{ax:SCStop} \tikzfig{./axioms/CSCZtop-02} 
        = \tikzfig{./axioms/CZCStop}
    \end{gather*}
\end{proof}

\begin{lemma}\label{ax:CCZ-SWAP}
    CubicPhase$_d$ $\vdash$ 
    \tikzfig{./axioms/CCZSWAP} = \tikzfig{./axioms/SWAPCCZ}
\end{lemma}

\begin{proof}
    \begin{gather*}
        \tikzfig{./axioms/CCZSWAP}
        = \tikzfig{./axioms/CCZSWAP-00} \\
        = \tikzfig{./axioms/CCZSWAP-01} \\
        \eqeqref{ax:PhaseGadgetT} \tikzfig{./axioms/CCZSWAP-02} \\
        \eqeqref{lemma-XX2-comm} \tikzfig{./axioms/CCZSWAP-03} \\
        \eqdeuxeqref{ax:PhaseGadgetT}{ax:PhaseGadgetT} \tikzfig{./axioms/CCZSWAP-05} \\
        \eqeqref{lemma-XX2-comm} \tikzfig{./axioms/CCZSWAP-06} \\
        \eqeqref{ax:PhaseGadgetT} \tikzfig{./axioms/CCZSWAP-07} \\
        \eqeqref{axiom-I} \tikzfig{./axioms/CCZSWAP-08} \\
        \eqeqref{axiom-d-cnot} \tikzfig{./axioms/CCZSWAP-09} \\
        \eqeqref{ax:PhaseGadgetT} \tikzfig{./axioms/CCZSWAP-10} \\
        = \tikzfig{./axioms/SWAPCCZ}
    \end{gather*}
\end{proof}

\begin{lemma}\label{ax:CCZ-SWAP2}
    CubicPhase$_d$ $\vdash$ 
    \tikzfig{./axioms/CCZSWAP2} = \tikzfig{./axioms/SWAPCCZ2}
\end{lemma}

\begin{proof}
    \begin{gather*}
        \tikzfig{./axioms/CCZSWAP2}
        = \tikzfig{./axioms/CCZSWAP2-00} \\
        = \tikzfig{./axioms/CCZSWAP2-01} \\
        \eqeqref{ax:PhaseGadgetT} \tikzfig{./axioms/CCZSWAP2-02} \\
        \eqeqref{lemma-XX2-comm} \tikzfig{./axioms/CCZSWAP2-03} \\
        \eqeqref{ax:PhaseGadgetT} \tikzfig{./axioms/CCZSWAP2-04} \\
        \eqeqref{lemma-XX2-comm} \tikzfig{./axioms/CCZSWAP2-05} \\
        \eqeqref{ax:PhaseGadgetT} \tikzfig{./axioms/CCZSWAP2-06} \\
        \eqeqref{axiom-I} \tikzfig{./axioms/CCZSWAP2-07} \\
        \eqeqref{axiom-d-cnot} \tikzfig{./axioms/CCZSWAP2-08} \\
        \eqeqref{ax:PhaseGadgetT}  \tikzfig{./axioms/CCZSWAP2-09} \\
        = \tikzfig{./axioms/SWAPCCZ}
    \end{gather*}
\end{proof}

\begin{lemma}\label{ax:CCZ-CNOTbot}
    CubicPhase$_d$ $\vdash$ 
    \tikzfig{./axioms/CCZCNOTbot} = \tikzfig{./axioms/CNOTCCZbot}
\end{lemma}

\begin{proof}
    \begin{gather*}
        \tikzfig{./axioms/CCZCNOTbot}\\
        = \tikzfig{./axioms/CCZCNOTbot-00} \\
        \eqeqref{ax:TC-control} \tikzfig{./axioms/CCZCNOTbot-01} \\
        \eqeqref{lemma-XX2-comm} \tikzfig{./axioms/CCZCNOTbot-02} \\
        = \tikzfig{./axioms/CNOTCCZbot}
    \end{gather*}
\end{proof}

\begin{lemma}\label{ax:CSCS3}
    CubicPhase$_d$ $\vdash$ 
    \tikzfig{./axioms/CSCSA3} = \tikzfig{./axioms/CSCSB3}
\end{lemma}

\begin{proof}
    \begin{gather*}
        \tikzfig{./axioms/CSCSA3}\\
        = \tikzfig{./axioms/CSCS3-00} \\
        \eqeqref{ax:CS-CNOT-top} \tikzfig{./axioms/CSCS3-01} \\
        \eqtroiseqref{ax:TCStop}{ax:SCStop}{ax:ZCStop} \tikzfig{./axioms/CSCS3-02} \\
        \eqeqref{ax:CSCZbot} \tikzfig{./axioms/CSCS3-03} \\
        = \tikzfig{./axioms/CSCSB3}
    \end{gather*}
\end{proof}

\begin{lemma}\label{ax:CSCS}
    CubicPhase$_d$ $\vdash$ 
    \tikzfig{./axioms/CSCSA} = \tikzfig{./axioms/CSCSB}
\end{lemma}

\begin{proof}
    \begin{gather*}
        \tikzfig{./axioms/CSCSA}\\
        = \tikzfig{./axioms/CSCS-00} \\
        \eqeqref{ax:CS-CNOT-top2} \tikzfig{./axioms/CSCS-01} \\
        \eqtroiseqref{ax:TCSbot}{ax:SCSbot}{ax:ZCSbot} \tikzfig{./axioms/CSCS-02} \\
        \eqeqref{ax:CSCZtop} \tikzfig{./axioms/CSCS-03} \\
        = \tikzfig{./axioms/CSCSB}
    \end{gather*}
\end{proof}

\begin{lemma}\label{ax:CSCS2}
    CubicPhase$_d$ $\vdash$ 
    \tikzfig{./axioms/CSCSA2} = \tikzfig{./axioms/CSCSB2}
\end{lemma}

\begin{proof}
    \begin{gather*}
        \tikzfig{./axioms/CSCSA2}\\
        = \tikzfig{./axioms/CSCS2-00} \\
        \eqeqref{ax:PhaseGadgetT}  \tikzfig{./axioms/CSCS2-01} \\
        \eqeqref{ax:CS-CNOT-top2} \tikzfig{./axioms/CSCS2-02} \\
        \eqeqref{ax:CSCZbot} \tikzfig{./axioms/CSCS2-03} \\
        \eqeqref{ax:PhaseGadgetT}  \tikzfig{./axioms/CSCS2-04} \\
        = \tikzfig{./axioms/CSCSB2}
    \end{gather*}
\end{proof}


\medskip\noindent\textbf{\texorpdfstring{\(CCZ\)}{CCZ}-centred overlap sorting.}

\begin{lemma}\label{ax:TCCZ}
    CubicPhase$_d$ $\vdash$ 
    \tikzfig{./axioms/TCCZ} = \tikzfig{./axioms/CCZT}
\end{lemma}

\begin{proof}
    \begin{gather*}
        \tikzfig{./axioms/TCCZ}\\
        = \tikzfig{./axioms/TCCZ-00} \\
        \eqeqref{ax:TC-control} \tikzfig{./axioms/TCCZ-01} \\
        = \tikzfig{./axioms/CCZT}
    \end{gather*}
\end{proof}

\begin{lemma}\label{ax:SCCZ}
    CubicPhase$_d$ $\vdash$ 
    \tikzfig{./axioms/SCCZ} = \tikzfig{./axioms/CCZS}
\end{lemma}

\begin{proof}
    \begin{gather*}
        \tikzfig{./axioms/SCCZ}\\
        = \tikzfig{./axioms/SCCZ-00} \\
        \eqeqref{ax:SC-control} \tikzfig{./axioms/SCCZ-01} \\
        \eqeqref{ax:TS} \tikzfig{./axioms/SCCZ-02} \\
        = \tikzfig{./axioms/CCZS}
    \end{gather*}
\end{proof}

\begin{lemma}\label{ax:ZCCZ}
    CubicPhase$_d$ $\vdash$ 
    \tikzfig{./axioms/ZCCZ} = \tikzfig{./axioms/CCZZ}
\end{lemma}

\begin{proof}
    \begin{gather*}
        \tikzfig{./axioms/ZCCZ}\\
        = \tikzfig{./axioms/ZCCZ-00} \\
        \eqeqref{ax:ZC-control} \tikzfig{./axioms/ZCCZ-01} \\
        \eqeqref{ax:TZ} \tikzfig{./axioms/ZCCZ-02} \\
        = \tikzfig{./axioms/CCZZ}
    \end{gather*}
\end{proof}

\begin{lemma}\label{ax:CCZCZtop}
    CubicPhase$_d$ $\vdash$ 
    \tikzfig{./axioms/CCZCZtop} = \tikzfig{./axioms/CZCCZtop}
\end{lemma}

\begin{proof}
    \begin{gather*}
        \tikzfig{./axioms/CCZCZtop}
        = \tikzfig{./axioms/CCZCZtop-00} 
        \eqeqref{ax:PhaseGadgetS} \tikzfig{./axioms/CCZCZtop-01} \\
        \eqeqref{ax:CCZ-CNOTbot} \tikzfig{./axioms/CCZCZtop-02} 
        \eqeqref{ax:SCCZ} \tikzfig{./axioms/CCZCZtop-03} 
        = \tikzfig{./axioms/CZCCZtop}
    \end{gather*}
\end{proof}

\begin{lemma}\label{ax:CCZ-CS-1}
    CubicPhase$_d$ $\vdash$ 
    \tikzfig{./axioms/CCZCS1} = \tikzfig{./axioms/CSCCZ1}
\end{lemma}

\begin{proof}
    \begin{gather*}
        \tikzfig{./axioms/CCZCS1}
        = \tikzfig{./axioms/CCZCS1-00} \\
        \eqeqref{ax:CCZCZtop} \tikzfig{./axioms/CCZCS1-02} \\
        \eqtroiseqref{ax:TCCZ}{ax:SCCZ}{ax:ZCCZ} \tikzfig{./axioms/CCZCS1-03} \\
        \eqeqref{ax:CCZ-CNOTbot} \tikzfig{./axioms/CCZCS1-04}
        = \tikzfig{./axioms/CSCCZ1}
    \end{gather*}
\end{proof}

\begin{lemma}\label{ax:CCZ-CS-2}
    CubicPhase$_d$ $\vdash$ 
    \tikzfig{./axioms/CCZCS2} = \tikzfig{./axioms/CSCCZ2}
\end{lemma}

\begin{proof}
    \begin{gather*}
        \tikzfig{./axioms/CCZCS2}
        = \tikzfig{./axioms/CCZCS2-00} \\
        \eqeqref{ax:CCZCZtop} \tikzfig{./axioms/CCZCS2-01} \\
        \eqeqref{ax:PhaseGadgetT} \tikzfig{./axioms/CCZCS2-02} \\
        \eqeqref{ax:CCZ-CNOTbot} \tikzfig{./axioms/CCZCS2-03} \\
        \eqeqref{ax:PhaseGadgetT} \tikzfig{./axioms/CCZCS2-04} 
        = \tikzfig{./axioms/CSCCZ2}
    \end{gather*}
\end{proof}

\begin{lemma}\label{ax:CCZCCZ}
    CubicPhase$_d$ $\vdash$ 
    \tikzfig{./axioms/CCZCCZA} = \tikzfig{./axioms/CCZCCZB}
\end{lemma}

\begin{proof}
    \begin{gather*}
        \tikzfig{./axioms/CCZCCZA}\\
        = \tikzfig{./axioms/CCZCCZ-00} \\
        \eqeqref{ax:CCZ-CNOTbot} \tikzfig{./axioms/CCZCCZ-01} \\
        \eqeqref{ax:TCCZ} \tikzfig{./axioms/CCZCCZ-02} \\
        = \tikzfig{./axioms/CCZCCZB}
    \end{gather*}
\end{proof}


\medskip\noindent\textbf{Finite-order reductions for controlled cubic phases.}

\begin{lemma}\label{ax:CS-order-d}
    CubicPhase$_d$ $\vdash$ 
    \tikzfig{./pieces/CSd} = \tikzfig{two-lines}
\end{lemma}

\begin{proof}
    \begin{gather*}
        \tikzfig{./pieces/CSd}
        = \left(\tikzfig{./axioms/CSdef}\right)^{\circ d}\\
        \eqdeuxeqref{ax:GadgetCZmin}{ax:CCZ-CS-3-lem-CZ}\left(\tikzfig{./axioms/CSd-00-A}\right)^{\circ d} \circ  \left(\tikzfig{./axioms/CSd-00-B}\right)^{\circ d}\\
        \eqeqref{ax:CZ-d} \left(\tikzfig{./axioms/CSd-00-B}\right)^{\circ d}\\
        \eqtroiseqref{ax:TC-control}{ax:SC-control}{ax:ZC-control} \left(\tikzfig{./axioms/CSd-01-A}\right)^{\circ d} \circ  \left(\tikzfig{./axioms/CSd-01-B}\right)^{\circ d}\\
        \eqtroiseqref{ax:TS}{ax:TZ}{ax:SZ} \tikzfig{./axioms/CSd-02-A} \circ  \left(\tikzfig{./axioms/CSd-01-B}\right)^{\circ d}\\
        \eqtroiseqref{ax:T-order-d}{ax:S-order-d}{ax:Z-order-d} \left(\tikzfig{./axioms/CSd-01-B}\right)^{\circ d}
        \eqdeuxeqref{ax:PhaseGadgetComm}{axiom-d-cnot} \tikzfig{./axioms/CSd-02-B}\\
        \eqeqref{ax:T-order-d} \tikzfig{./axioms/CSd-03-B}
        \eqeqref{axiom-d-cnot} \tikzfig{two-lines}
    \end{gather*}
\end{proof}

\begin{corollary}\label{ax:SC-order-d}
    CubicPhase$_d$ $\vdash$
    $\left(\gSC\right)^{\circ d} = \tikzfig{two-lines}$
\end{corollary}

\begin{proof}
Write \(\sigma=\gSWAP\).  By \cref{def:CubicPhase-derived-diagonals}, \(\gSC\) is \(\gCS\) conjugated by the two-wire symmetry.
Thus the \(\gSC\) order relation is the following same-arity rewrite sequence:
\begin{gather*}
    \left(\gSC\right)^{\circ d}
    = \left(\sigma\circ\gCS\circ\sigma\right)^{\circ d}\\
    \overset{\mathrm{PROP}}{=} \sigma\circ\left(\gCS\right)^{\circ d}\circ\sigma\\
    \eqeqref{ax:CS-order-d} \sigma\circ\id_2\circ\sigma\\
    \overset{\mathrm{PROP}}{=} \id_2 .
\end{gather*}
The first equality unfolds the definition of \(\gSC\); the PROP steps use associativity, the unit laws, and the \(d-1\) internal cancellations of adjacent symmetries.  No tensor factor is added or cancelled.
\end{proof}

\begin{lemma}\label{ax:CCZ-d-lem1}
    CubicPhase$_d$ $\vdash$ 

    \tikzfig{./axioms/CCZdlem1}
    
     = \tikzfig{./axioms/CCZdlem1R}
\end{lemma}

\begin{proof}
    \begin{gather*}
         \tikzfig{./axioms/CCZdlem1} \\
        \eqeqref{ax:PhaseGadgetT}\tikzfig{./axioms/CCZdlem1-00} \\
        \eqeqref{ax:TC-control}\tikzfig{./axioms/CCZdlem1-01} \\
        \eqeqref{lemma-XX2-comm}\tikzfig{./axioms/CCZdlem1-02} \\
        \eqeqref{lemma-XX-comm}\tikzfig{./axioms/CCZdlem1-03} \\
        \eqeqref{ax:PhaseGadgetT}\tikzfig{./axioms/CCZdlem1-04} \\
        \eqeqref{axiom-I}\tikzfig{./axioms/CCZdlem1-05} \\
        \eqeqref{axiom-d-cnot}\tikzfig{./axioms/CCZdlem1-06} \\
        \eqdeuxeqref{lemma-XX2-comm}{ax:TC-control}\tikzfig{./axioms/CCZdlem1-07} \\
        \eqeqref{axiom-d-cnot}\tikzfig{./axioms/CCZdlem1-08} \\
        \eqeqref{axiom-I}\tikzfig{./axioms/CCZdlem1-09} \\
        \eqeqref{ax:PhaseGadgetT}\tikzfig{./axioms/CCZdlem1-10} \\
        \eqeqref{lemma-XX-comm}\tikzfig{./axioms/CCZdlem1-11} \\
        \eqeqref{ax:PhaseGadgetT}\tikzfig{./axioms/CCZdlem1-12} \\
        \eqeqref{axiom-d-cnot}\tikzfig{./axioms/CCZdlem1-13} \\
        \eqeqref{ax:TC-control}\tikzfig{./axioms/CCZdlem1-14} \\
        \eqdeuxeqref{axiom-d-cnot}{axiom-d-cnot}\tikzfig{./axioms/CCZdlem1R}
    \end{gather*}
\end{proof}

\begin{lemma}\label{ax:CCZ-d-lem2}
    CubicPhase$_d$ $\vdash$ 

    \tikzfig{./axioms/CCZdlem2}

    = \tikzfig{./axioms/CCZdlem2R}
\end{lemma}

\begin{proof}
    \begin{gather*}
         \tikzfig{./axioms/CCZdlem2} \\
        \eqeqref{ax:TC-control}\tikzfig{./axioms/CCZdlem2-00} \\
        \eqeqref{lemma-XX2-comm}\tikzfig{./axioms/CCZdlem2-01} \\
        \eqeqref{ax:PhaseGadgetT}\tikzfig{./axioms/CCZdlem2-02} \\
        \eqeqref{ax:TC-control}\tikzfig{./axioms/CCZdlem2-03} \\
        \eqeqref{lemma-XX2-comm}\tikzfig{./axioms/CCZdlem2-04} \\
        \eqeqref{ax:PhaseGadgetT}\tikzfig{./axioms/CCZdlem2-05} \\
        \eqeqref{axiom-d-cnot}\tikzfig{./axioms/CCZdlem2R}\\
    \end{gather*}
\end{proof}

\begin{lemma}\label{ax:CCZ-d-lem3}
    CubicPhase$_d$ $\vdash$ 

    \tikzfig{./axioms/CCZdlem3}

    = \tikzfig{./axioms/CCZdlem3R}
\end{lemma}

\begin{proof}
    \begin{gather*}
         \tikzfig{./axioms/CCZdlem3} \\
        \eqeqref{ax:TC-control}\tikzfig{./axioms/CCZdlem3-00} \\
        \eqeqref{ax:PhaseGadgetT}\tikzfig{./axioms/CCZdlem3-01} \\
        \eqeqref{ax:TC-control}\tikzfig{./axioms/CCZdlem3-02} \\
        \eqeqref{lemma-XX2-comm}\tikzfig{./axioms/CCZdlem3-03} \\
        \eqeqref{ax:PhaseGadgetT}\tikzfig{./axioms/CCZdlem3-04} \\
        \eqeqref{axiom-d-cnot}\tikzfig{./axioms/CCZdlem3R}\\
    \end{gather*}
\end{proof}

\begin{lemma}\label{ax:CCZ-d-lem4}
    CubicPhase$_d$ $\vdash$
    \tikzfig{./axioms/CCZdlem4} = \tikzfig{./axioms/CCZdlem4R}
\end{lemma}

\begin{proof}
    \begin{gather*}
         \tikzfig{./axioms/CCZdlem4} \\
        \eqeqref{ax:TC-control}\tikzfig{./axioms/CCZdlem4-00} \\
        \eqeqref{ax:PhaseGadgetT}\tikzfig{./axioms/CCZdlem4-01} \\
        \eqeqref{ax:TC-control}\tikzfig{./axioms/CCZdlem4-02} \\
        \eqeqref{ax:PhaseGadgetT}\tikzfig{./axioms/CCZdlem4-03} \\
        \eqdeuxeqref{axiom-d-cnot}{ax:T-order-d}\tikzfig{./axioms/CCZdlem4R}\\
    \end{gather*}
\end{proof}

\begin{lemma}\label{ax:CCZ-d}
    CubicPhase$_d$ $\vdash$
    \tikzfig{./pieces/CCZd} = \tikzfig{three-lines}
\end{lemma}

\begin{proof}
    \begin{gather*}
        \tikzfig{./pieces/CCZd}
        = \left(\tikzfig{./axioms/CCZdef}\right)^{\circ d}\\
        \eqeqref{ax:CCZ-d-lem1} \tikzfig{./axioms/CCZd-00-A} \circ \left(\tikzfig{./axioms/CCZd-00-B}\right)^{\circ d}\\
        \eqdeuxeqref{ax:T-order-d}{axiom-d-cnot} \left(\tikzfig{./axioms/CCZd-00-B}\right)^{\circ d}\\
        \eqeqref{ax:CCZ-d-lem2} \tikzfig{./axioms/CCZd-01-A} \circ \left(\tikzfig{./axioms/CCZd-01-B}\right)^{\circ d}\\
        \eqdeuxeqref{ax:T-order-d}{axiom-d-cnot} (\tikzfig{./axioms/CCZd-01-B})^{\circ d}\\
        \eqeqref{ax:CCZ-d-lem3} \tikzfig{./axioms/CCZd-02-A} \circ \left(\tikzfig{./axioms/CCZd-02-B}\right)^{\circ d}\\
        \eqdeuxeqref{ax:T-order-d}{axiom-d-cnot} \left(\tikzfig{./axioms/CCZd-02-B}\right)^{\circ d}\\
        \eqeqref{ax:T-order-d} \tikzfig{three-lines}
    \end{gather*}
\end{proof}


\medskip\noindent\textbf{Mixed \texorpdfstring{\(CCZ/CZ/CS\)}{CCZ/CZ/CS} sorting.}

\begin{lemma}\label{ax:CCZCZ}
    CubicPhase$_d$ $\vdash$
    \tikzfig{./axioms/CCZCZ} = \tikzfig{./axioms/CZCCZ}
\end{lemma}

\begin{proof}
    \begin{gather*}
        \tikzfig{./axioms/CCZCZ}\\
        = \scalebox{0.8}{\tikzfig{./axioms/CCZCZ-00}}\\
        \eqdeuxeqref{ax:TC-control}{ax:SC-control} \scalebox{0.8}{\tikzfig{./axioms/CCZCZ-01}}\\
        \eqeqref{axiom-I} \scalebox{0.8}{\tikzfig{./axioms/CCZCZ-02}}\\
        \eqeqref{axiom-d-cnot} \scalebox{0.8}{\tikzfig{./axioms/CCZCZ-03}}\\
        \eqdeuxeqref{ax:PhaseGadgetT}{ax:PhaseGadgetS} \scalebox{0.8}{\tikzfig{./axioms/CCZCZ-04}}\\
        \eqdeuxeqref{ax:TC-control}{ax:SC-control} \scalebox{0.8}{\tikzfig{./axioms/CCZCZ-05}}\\
        \eqeqref{ax:PhaseGadgetS} \scalebox{0.8}{\tikzfig{./axioms/CCZCZ-06}}\\
        \eqeqref{lemma-XX2-comm} \scalebox{0.8}{\tikzfig{./axioms/CCZCZ-07}}\\
        \eqeqref{ax:PhaseGadgetS} \scalebox{0.8}{\tikzfig{./axioms/CCZCZ-08}}\\
        \eqeqref{lemma-XX2-comm} \scalebox{0.8}{\tikzfig{./axioms/CCZCZ-09}}\\
        \eqeqref{axiom-d-cnot} \scalebox{0.8}{\tikzfig{./axioms/CCZCZ-10}}\\
        \eqeqref{ax:TS} \scalebox{0.8}{\tikzfig{./axioms/CCZCZ-11}}\\
        \eqeqref{axiom-d-cnot} \scalebox{0.8}{\tikzfig{./axioms/CCZCZ-12}}\\
        \eqeqref{ax:PhaseGadgetS} \scalebox{0.8}{\tikzfig{./axioms/CCZCZ-13}}\\
        \eqeqref{lemma-XX2-comm} \scalebox{0.8}{\tikzfig{./axioms/CCZCZ-14}}\\
        \eqeqref{ax:PhaseGadgetS} \scalebox{0.8}{\tikzfig{./axioms/CCZCZ-15}}\\
        \eqeqref{lemma-XX2-comm} \scalebox{0.8}{\tikzfig{./axioms/CCZCZ-16}}\\
        \eqeqref{ax:SC-control} \scalebox{0.8}{\tikzfig{./axioms/CCZCZ-17}}\\
        \eqeqref{ax:PhaseGadgetS} \scalebox{0.8}{\tikzfig{./axioms/CCZCZ-18}}\\
        \eqeqref{lemma-XX2-comm} \scalebox{0.8}{\tikzfig{./axioms/CCZCZ-19}}\\
        \eqeqref{ax:SC-control} \scalebox{0.8}{\tikzfig{./axioms/CCZCZ-20}}\\
        \eqeqref{axiom-d-cnot}   \scalebox{0.8}{\tikzfig{./axioms/CCZCZ-21}}\\
        \eqeqref{axiom-I} \scalebox{0.8}{\tikzfig{./axioms/CCZCZ-22}}\\
        \eqeqref{ax:PhaseGadgetT} \scalebox{0.8}{\tikzfig{./axioms/CCZCZ-23}}\\
        \tikzfig{./axioms/CZCCZ}
    \end{gather*}
\end{proof}

\begin{lemma}\label{ax:CCZCCZ2-lem}
    CubicPhase$_d$ $\vdash$ 
    \tikzfig{./axioms/CCZCCZA2-lem} = \tikzfig{./axioms/CCZCCZB2-lem}
\end{lemma}

\begin{proof}
    \begin{gather*}
        \tikzfig{./axioms/CCZCCZA2-lem} \\
        = \scalebox{0.9}{\tikzfig{./axioms/CCZCCZA2-lem-00}}\\
        \eqeqref{axiom-I} \scalebox{0.9}{\tikzfig{./axioms/CCZCCZA2-lem-01}}\\
        \eqeqref{axiom-d-cnot} \scalebox{0.9}{\tikzfig{./axioms/CCZCCZA2-lem-02}}\\
        \eqeqref{ax:PhaseGadgetT} \scalebox{0.9}{\tikzfig{./axioms/CCZCCZA2-lem-03}}\\
        \eqeqref{ax:TC-control} \scalebox{0.9}{\tikzfig{./axioms/CCZCCZA2-lem-04}}\\
        \eqeqref{ax:PhaseGadgetT} \scalebox{0.9}{\tikzfig{./axioms/CCZCCZA2-lem-05}}\\
        \eqeqref{axiom-I} \scalebox{0.9}{\tikzfig{./axioms/CCZCCZA2-lem-06}}\\
        \eqeqref{axiom-d-cnot} \scalebox{0.9}{\tikzfig{./axioms/CCZCCZA2-lem-07}}\\
        \eqeqref{ax:PhaseGadgetT}\scalebox{0.9}{\tikzfig{./axioms/CCZCCZA2-lem-08}}\\
        \eqeqref{lemma-XX2-comm} \scalebox{0.9}{\tikzfig{./axioms/CCZCCZA2-lem-09}}\\
        \eqeqref{ax:TC-control} \scalebox{0.9}{\tikzfig{./axioms/CCZCCZA2-lem-10}}\\
        \eqeqref{ax:PhaseGadgetT}\scalebox{0.9}{\tikzfig{./axioms/CCZCCZA2-lem-11}}\\
        \eqeqref{ax:TC-control}  \scalebox{0.9}{\tikzfig{./axioms/CCZCCZA2-lem-12}}\\
        \eqeqref{lemma-XX2-comm} \scalebox{0.9}{\tikzfig{./axioms/CCZCCZA2-lem-13}}\\
        \eqeqref{ax:PhaseGadgetT} \scalebox{0.9}{\tikzfig{./axioms/CCZCCZA2-lem-14}}\\
        = \tikzfig{./axioms/CCZCCZB2-lem}
    \end{gather*}
\end{proof}

\begin{lemma}\label{ax:CCZCCZ2}
    CubicPhase$_d$ $\vdash$ 
    \tikzfig{./axioms/CCZCCZA2} = \tikzfig{./axioms/CCZCCZB2}
\end{lemma}

\begin{proof}
    \begin{gather*}
        \tikzfig{./axioms/CCZCCZA2}\\
        = \tikzfig{./axioms/CCZCCZA2-00}\\
        \eqeqref{ax:TCCZ} \tikzfig{./axioms/CCZCCZA2-01}\\
        \eqeqref{ax:CCZCZtop} \tikzfig{./axioms/CCZCCZA2-02}\\
        \eqeqref{ax:CCZCCZ2-lem} \tikzfig{./axioms/CCZCCZA2-03}\\
        \eqeqref{axiom-I} \tikzfig{./axioms/CCZCCZA2-04}\\
        \eqeqref{ax:CCZCZtop} \tikzfig{./axioms/CCZCCZA2-05}\\
        \eqeqref{axiom-I} \tikzfig{./axioms/CCZCCZA2-06}\\
        \tikzfig{./axioms/CCZCCZB2}
    \end{gather*}
\end{proof}

\begin{lemma}\label{ax:CSCZ}
    CubicPhase$_d$ $\vdash$ 
    \tikzfig{./axioms/CSCZ} = \tikzfig{./axioms/CZCS}
\end{lemma}

\begin{proof}
    \begin{gather*}
        \tikzfig{./axioms/CSCZ}\\
        = \tikzfig{./axioms/CSCZ-00}\\
        \eqeqref{ax:PhaseGadgetComm} \tikzfig{./axioms/CSCZ-01}\\
        \eqeqref{axiom-d-cnot} \tikzfig{./axioms/CSCZ-01b}\\
        \eqeqref{ax:TS} \tikzfig{./axioms/CSCZ-01c}\\
        \eqeqref{axiom-d-cnot} \tikzfig{./axioms/CSCZ-01d}\\
        \eqtroiseqref{ax:SC-control}{ax:TC-control}{ax:ZC-control} \tikzfig{./axioms/CSCZ-02}\\
        \eqeqref{ax:PhaseGadgetT} \tikzfig{./axioms/CSCZ-03}\\
        \eqeqref{ax:SC-control} \tikzfig{./axioms/CSCZ-04}\\
        = \tikzfig{./axioms/CSCZ-05}\\
        \eqeqref{ax:PhaseGadgetT} \tikzfig{./axioms/CSCZ-06}\\
        = \tikzfig{./axioms/CSCZ-07}\\
        = \tikzfig{./axioms/CZCS}
    \end{gather*}
\end{proof}

\begin{lemma}\label{ax:CCZ-CS-3-lem-CZ}
    CubicPhase$_d$ $\vdash$ 
    \tikzfig{./axioms/CCZCS3-lem-CZ} = \tikzfig{./axioms/CCZCS3-lem-CZ-R}
\end{lemma}

\begin{proof}
    \begin{gather*}
        \tikzfig{./axioms/CCZCS3-lem-CZ}
        = \tikzfig{./axioms/CCZCS3-lem-CZ-00}\\
        \eqeqref{axiom-d-cnot} \tikzfig{./axioms/CCZCS3-lem-CZ-01}
        \eqeqref{ax:SC-control} \tikzfig{./axioms/CCZCS3-lem-CZ-02}\\
        \eqeqref{ax:TS} \tikzfig{./axioms/CCZCS3-lem-CZ-03}
        \eqeqref{axiom-d-cnot} \tikzfig{./axioms/CCZCS3-lem-CZ-04}\\
        \eqeqref{ax:PhaseGadgetT} \tikzfig{./axioms/CCZCS3-lem-CZ-05}
        \eqeqref{ax:SC-control} \tikzfig{./axioms/CCZCS3-lem-CZ-06}\\
        \eqeqref{ax:PhaseGadgetT}\tikzfig{./axioms/CCZCS3-lem-CZ-07}
        = \tikzfig{./axioms/CCZCS3-lem-CZ-R}
    \end{gather*}
\end{proof}

\begin{lemma}\label{ax:CCZ-CS-3-lem}
    CubicPhase$_d$ $\vdash$ 
    \tikzfig{./axioms/CCZCS3-lem} = \tikzfig{./axioms/CCZCS3-lem-R}
\end{lemma}

\begin{proof}
    \begin{gather*}
        \tikzfig{./axioms/CCZCS3-lem}\\
        = \tikzfig{./axioms/CCZCS3-lem-00}\\
        \eqeqref{ax:CCZ-CS-3-lem-CZ} \tikzfig{./axioms/CCZCS3-lem-01}\\
        \eqtroiseqref{ax:ZC-control}{ax:SC-control}{ax:TC-control} \tikzfig{./axioms/CCZCS3-lem-02}\\
        \eqeqref{axiom-d-cnot} \tikzfig{./axioms/CCZCS3-lem-03}\\
        = \tikzfig{./axioms/CCZCS3-lem-04}\\
        \eqeqref{axiom-d-cnot} \tikzfig{./axioms/CCZCS3-lem-05}\\
        \eqeqref{ax:PhaseGadgetComm} \tikzfig{./axioms/CCZCS3-lem-06}\\
        = \tikzfig{./axioms/CCZCS3-lem-R}
    \end{gather*}
\end{proof}

\begin{lemma}\label{ax:CCZ-CS-3}
    CubicPhase$_d$ $\vdash$ 
    \tikzfig{./axioms/CCZCS3} = \tikzfig{./axioms/CSCCZ3}
\end{lemma}

\begin{proof}
    \begin{gather*}
        \tikzfig{./axioms/CCZCS3}\\
        = \tikzfig{./axioms/CCZCS3-00}\\
        \eqdeuxeqref{ax:TCStop}{ax:TCSbot} \tikzfig{./axioms/CCZCS3-01}\\
        \eqeqref{ax:CS-CNOT-bot} \tikzfig{./axioms/CCZCS3-02}\\
        \eqeqref{ax:CS-CNOT-top} \tikzfig{./axioms/CCZCS3-03}\\
        \eqeqref{ax:CCZ-CS-3-lem} \tikzfig{./axioms/CCZCS3-04}\\
        \eqeqref{axiom-I} \tikzfig{./axioms/CCZCS3-05}\\
        \eqeqref{axiom-d-cnot} \tikzfig{./axioms/CCZCS3-06}\\
        \eqdeuxeqref{ax:CS-CNOT-bot}{ax:CS-CNOT-top} \tikzfig{./axioms/CCZCS3-07}\\
        \eqeqref{axiom-d-cnot} \tikzfig{./axioms/CCZCS3-08}\\
        \eqeqref{axiom-I} \tikzfig{./axioms/CCZCS3-09}\\
        = \tikzfig{./axioms/CSCCZ3}
    \end{gather*}
\end{proof}


\medskip\noindent\textbf{Cubic gadget cleanup and commutation package.}

\begin{lemma}\label{ax:GadgetCZmin}
    CubicPhase$_d$ $\vdash$ 
    \tikzfig{./axioms/GadgetCZmin} = \tikzfig{./axioms/GadgetCZminR}
\end{lemma}

\begin{proof}
    \begin{gather*}
        \tikzfig{./axioms/GadgetCZmin}\\
        = \tikzfig{./axioms/GadgetCZmin-00}\\
        \eqeqref{ax:PhaseGadgetComm} \tikzfig{./axioms/GadgetCZmin-01}\\
        \eqeqref{ax:SC-control} \tikzfig{./axioms/GadgetCZmin-02}\\
        \eqeqref{ax:PhaseGadgetT} \tikzfig{./axioms/GadgetCZmin-03}\\
        \eqeqref{ax:SC-control} \tikzfig{./axioms/GadgetCZmin-04}\\
        \eqeqref{ax:PhaseGadgetT} \tikzfig{./axioms/GadgetCZmin-05}\\
        = \tikzfig{./axioms/GadgetCZminR}
    \end{gather*}
\end{proof}

\begin{lemma}\label{ax:CSSC}
    CubicPhase$_d$ $\vdash$ 
    \tikzfig{./axioms/CSSC} = \tikzfig{./axioms/SCCS}
\end{lemma}

\begin{proof}
    \begin{gather*}
        \tikzfig{./axioms/CSSC}\\
        = \scalebox{0.8}{\tikzfig{./axioms/SCCS-00}}\\
        \eqtroiseqref{ax:TC-control}{ax:SC-control}{ax:ZC-control} \scalebox{0.8}{\tikzfig{./axioms/SCCS-01}}\\
        \eqtroiseqref{ax:TCZ}{ax:SCZ}{ax:ZCZ} \scalebox{0.8}{\tikzfig{./axioms/SCCS-02}}\\
        \eqdeuxeqref{ax:GadgetCZmin}{ax:CCZ-CS-3-lem-CZ} \scalebox{0.8}{\tikzfig{./axioms/SCCS-03}}\\
        \eqeqref{ax:PhaseGadgetT} \scalebox{0.8}{\tikzfig{./axioms/SCCS-04}}\\
        \eqtroiseqref{ax:TC-control}{ax:SC-control}{ax:ZC-control} \scalebox{0.8}{\tikzfig{./axioms/SCCS-05}}\\
        \eqeqref{ax:PhaseGadgetComm} \scalebox{0.8}{\tikzfig{./axioms/SCCS-06}}\\
        \eqeqref{ax:PhaseGadgetT} \scalebox{0.8}{\tikzfig{./axioms/SCCS-07}}\\
        \eqtroiseqref{ax:TC-control}{ax:SC-control}{ax:ZC-control} \scalebox{0.8}{\tikzfig{./axioms/SCCS-08}}\\
        \eqeqref{ax:PhaseGadgetT} \scalebox{0.8}{\tikzfig{./axioms/SCCS-09}}\\
        \eqeqref{axiom-d-cnot} \scalebox{0.8}{\tikzfig{./axioms/SCCS-10}}\\
        \eqeqref{ax:TM}\scalebox{0.8}{\tikzfig{./axioms/SCCS-11}}\\
        \eqdeuxeqref{axiom-mult1}{axiom-multxy} \scalebox{0.8}{\tikzfig{./axioms/SCCS-12}}\\
        \eqdeuxeqref{ax:TS}{ax:TZ}\scalebox{0.8}{ \tikzfig{./axioms/SCCS-13}}\\
        \eqdeuxeqref{axiom-mult1}{axiom-multxy}\scalebox{0.8}{ \tikzfig{./axioms/SCCS-14}}\\
        \eqeqref{ax:TM} \scalebox{0.8}{\tikzfig{./axioms/SCCS-15}}\\
        \eqeqref{axiom-d-cnot} \scalebox{0.8}{ \tikzfig{./axioms/SCCS-16}}\\
        \eqeqref{ax:PhaseGadgetT} \scalebox{0.8}{\tikzfig{./axioms/SCCS-17}}\\
        \eqeqref{ax:PhaseGadgetComm} \scalebox{0.8}{\tikzfig{./axioms/SCCS-18}}\\
        \eqdeuxeqref{ax:GadgetCZmin}{ax:CCZ-CS-3-lem-CZ}  \scalebox{0.8}{\tikzfig{./axioms/SCCS-19}}\\
        \eqtroiseqref{ax:TC-control}{ax:SC-control}{ax:ZC-control} \scalebox{0.8}{\tikzfig{./axioms/SCCS-20}}\\
        = \tikzfig{./axioms/SCCS}
    \end{gather*}
\end{proof}

\medskip\noindent\textbf{Commutation package.}
The preceding certificates prove the finite overlap cases for \cref{lem:phase-diagonal-commutation}.

\begin{proof}[Diagonal-commutation certificates for \cref{lem:phase-diagonal-commutation}]
\Cref{tab:diag-comm-lemmas} lists the non-scalar cubic family pairs; lower fragments are obtained by restriction.
Scalars, equal-family cases, disjoint supports, and surrounding symmetries are handled by the monoidal and PROP laws.
Cells with several references record distinct overlap or orientation cases, and lower-triangular cases are read in the reverse direction.
For \(\gCZ\) and \(\gCCZ\), the listed swap lemmas reduce wire-order cases; \(\gCS\) and \(\gSC\) remain oriented.
\end{proof}

The cubic transport certificates below cite this commutation package without further comment.


\subsection{Cubic Affine-Transport Certificates}\label{derivation:cubic-affine-transport}

The remaining certificates establish the cubic transport cases not covered by the quadratic package, citing \cref{lem:phase-diagonal-commutation,lem:quad-diag-affine-transport} directly.

\begin{lemma}\label{ax:TC-target}
    CubicPhase$_d$ $\vdash$ 
    \tikzfig{./axioms/TCNOTtop} = \tikzfig{./axioms/CNOTTtop}
\end{lemma}

\begin{proof}
    \begin{gather*}
        \tikzfig{./axioms/TCNOTtop}\\
        = \scalebox{0.8}{\tikzfig{./axioms/TCNOTtop-00}}\\
        \eqtroiseqref{ax:TC-control}{lem:phase-diagonal-commutation}{ax:T-order-d} \scalebox{0.8}{\tikzfig{./axioms/TCNOTtop-01}}\\
        \eqeqref{ax:GadgetCZmin}\scalebox{0.8}{ \tikzfig{./axioms/TCNOTtop-02}}\\
        \eqeqref{ax:CCZ-CS-3-lem-CZ} \scalebox{0.8}{\tikzfig{./axioms/TCNOTtop-03}}\\
        \eqeqref{lem:phase-diagonal-commutation} \scalebox{0.8}{\tikzfig{./axioms/TCNOTtop-04}}\\
        \eqeqref{ax:CZ-d} \scalebox{0.8}{\tikzfig{./axioms/TCNOTtop-05}}\\
        \eqeqref{ax:PhaseGadgetT} \scalebox{0.8}{\tikzfig{./axioms/TCNOTtop-06}}\\
        \eqeqref{lem:quad-diag-affine-transport} \scalebox{0.8}{\tikzfig{./axioms/TCNOTtop-07}}\\
        \eqeqref{ax:PhaseGadgetComm} \scalebox{0.8}{\tikzfig{./axioms/TCNOTtop-08}}\\
        \eqeqref{axiom-d-cnot} \scalebox{0.8}{\tikzfig{./axioms/TCNOTtop-09}}\\
        \eqeqref{ax:TM} \scalebox{0.8}{\tikzfig{./axioms/TCNOTtop-10}}\\
        \eqeqref{ax:S-order-d} \scalebox{0.8}{\tikzfig{./axioms/TCNOTtop-11}}\\
        \eqdeuxeqref{axiom-mult1}{axiom-multxy} \tikzfig{./axioms/TCNOTtop-12}\\
        \eqdeuxeqref{lem:phase-diagonal-commutation}{ax:T-order-d} \tikzfig{./axioms/TCNOTtop-13}\\
        \eqdeuxeqref{axiom-mult1}{axiom-multxy} \tikzfig{./axioms/TCNOTtop-14}\\
        \eqeqref{axiom-multcnot} \tikzfig{./axioms/TCNOTtop-15}\\
        \eqeqref{lem:quad-diag-affine-transport} \tikzfig{./axioms/TCNOTtop-16}\\
        \eqeqref{lem:quad-diag-affine-transport} \tikzfig{./axioms/TCNOTtop-17}\\
        \eqdeuxeqref{lem:quad-diag-affine-transport}{lem:phase-diagonal-commutation} \tikzfig{./axioms/TCNOTtop-18}\\
        \eqeqref{ax:Z-order-d} \tikzfig{./axioms/TCNOTtop-19}\\
        \eqeqref{ax:PhaseGadgetT} \tikzfig{./axioms/TCNOTtop-20}\\
        \eqeqref{lem:quad-diag-affine-transport} \tikzfig{./axioms/TCNOTtop-21}\\
        \eqdeuxeqref{ax:Z-order-d}{lem:phase-diagonal-commutation}\tikzfig{./axioms/TCNOTtop-22}\\
        \eqeqref{ax:S-order-d} \tikzfig{./axioms/TCNOTtop-23}\\
        \eqeqref{ax:CZ-SWAP} \tikzfig{./axioms/TCNOTtop-24}\\
        \eqeqref{lem:quad-diag-affine-transport}\tikzfig{./axioms/TCNOTtop-25}\\
        \eqeqref{ax:CCZ-CS-3-lem-CZ} \tikzfig{./axioms/TCNOTtop-26}\\
        \eqdeuxeqref{lem:phase-diagonal-commutation}{ax:CZ-d} \tikzfig{./axioms/TCNOTtop-27}\\
        \eqeqref{lem:quad-diag-affine-transport} \tikzfig{./axioms/TCNOTtop-28}\\
        \eqdeuxeqref{axiom-mult1}{axiom-multxy} \tikzfig{./axioms/TCNOTtop-29}\\
        \eqdeuxeqref{lem:quad-diag-affine-transport}{ax:S-order-d} \tikzfig{./axioms/TCNOTtop-30}\\
        \eqeqref{ax:PhaseGadgetT} \tikzfig{./axioms/TCNOTtop-31}\\
        \eqdeuxeqref{lem:quad-diag-affine-transport}{ax:S-order-d}  \tikzfig{./axioms/TCNOTtop-32}\\
        \eqdeuxeqref{lem:quad-diag-affine-transport}{ax:Z-order-d} \tikzfig{./axioms/TCNOTtop-33}\\
        \eqeqref{axiom-d-cnot} \tikzfig{./axioms/TCNOTtop-34}\\
        \eqeqref{ax:T-order-d} \tikzfig{./axioms/TCNOTtop-35}
        \eqeqref{ax:PhaseGadgetT} \tikzfig{./axioms/TCNOTtop-36}\\
        \eqdeuxeqref{ax:TC-control}{ax:T-order-d}  \tikzfig{./axioms/TCNOTtop-37}
        \eqeqref{axiom-d-cnot} \tikzfig{./axioms/CNOTTtop}
    \end{gather*}
\end{proof}

\medskip\noindent\textbf{One-wire cubic transport package.}

\begin{lemma}\label{lem:one-wire-cubic-diag-affine-transport}
In \(\cat{CubicPhase}_d\), every one-wire diagonal family \(G\in\{\gZ,\gS,\gT\}\) transports through \(\gX\), \(\gCX\), and \(\gM\), for every compatible choice of local supports.
\end{lemma}

\begin{proof}
For \(\gZ\) and \(\gS\), use \cref{lem:quad-diag-affine-transport}.
For \(\gT\), use \eqref{ax:TX}, \eqref{ax:TC-control}, \cref{ax:TC-target}, and \eqref{ax:TM}; tensor-context, symmetry, and disjoint-support instances are supplied by the PROP laws.
\end{proof}


\medskip\noindent\textbf{Multiplier transport rules.}
The multiplier transports \eqref{ax:CS-M-top}, \eqref{ax:CS-M-bot}, and
\eqref{ax:CCZ-M} are primitive cubic equations in Figure~\ref*{fig:CubicPhase_rules}.
They supply the multiplier entries of \cref{tab:diag-through-affine}.

\medskip\noindent\textbf{Translation transport for controlled cubic phases.}

\begin{lemma}\label{ax:CS-X-bot}
    CubicPhase$_d$ $\vdash$ 
    \tikzfig{./axioms/CSXbot} = \tikzfig{./axioms/XCSbot}
\end{lemma}

\begin{proof}
    \begin{gather*}
        \tikzfig{./axioms/CSXbot}
        = \tikzfig{./axioms/CSXbot-00}\\
        \eqeqref{lem:quad-diag-affine-transport} \tikzfig{./axioms/CSXbot-01}\\
        \eqeqref{ax:X-CX-target} \tikzfig{./axioms/CSXbot-02}\\
        \eqdeuxeqref{lem:one-wire-cubic-diag-affine-transport}{lem:phase-diagonal-commutation}\tikzfig{./axioms/CSXbot-03}\\
        \eqeqref{ax:X-CX-target} \tikzfig{./axioms/CSXbot-04}\\
        \eqdeuxeqref{lem:one-wire-cubic-diag-affine-transport}{lem:phase-diagonal-commutation} \tikzfig{./axioms/CSXbot-05}\\
        \eqeqref{ax:X-CX-target} \tikzfig{./axioms/CSXbot-06}\\
        \eqeqref{axiom-d-cnot} \tikzfig{./axioms/CSXbot-07}\\
        \eqeqref{ax:PhaseGadgetComm} \tikzfig{./axioms/CSXbot-08}\\
        = \tikzfig{./axioms/CSXbot-09}\\
        \eqeqref{ax:S-order-d} \tikzfig{./axioms/CSXbot-09b}\\
        \eqeqref{ax:CZexp} \tikzfig{./axioms/CSXbot-10}\\
        \eqdeuxeqref{axiom-mult1}{axiom-multxy} \tikzfig{./axioms/CSXbot-11}\\
        \eqeqref{axiom-multcnot} \tikzfig{./axioms/CSXbot-12}\\
        \eqeqref{lem:one-wire-cubic-diag-affine-transport} \tikzfig{./axioms/CSXbot-13}\\
        \eqeqref{axiom-d-cnot} \tikzfig{./axioms/CSXbot-14}
        \eqeqref{lem:quad-diag-affine-transport} \tikzfig{./axioms/CSXbot-15}\\
        \eqeqref{lem:phase-diagonal-commutation} \tikzfig{./axioms/CSXbot-16}
        \eqeqref{ax:S-order-d} \tikzfig{./axioms/CSXbot-17}\\
        \eqeqref{lem:one-wire-cubic-diag-affine-transport} \tikzfig{./axioms/CSXbot-18}
        \eqdeuxeqref{axiom-mult1}{axiom-multxy}  \tikzfig{./axioms/CSXbot-19}\\
        \eqeqref{ax:S-order-d} \tikzfig{./axioms/CSXbot-20}
        \eqdeuxeqref{ax:Z-order-d}{lem:phase-diagonal-commutation}  \tikzfig{./axioms/XCSbot}
    \end{gather*}
\end{proof}

\begin{lemma}\label{ax:CS-X-top}
    CubicPhase$_d$ $\vdash$ 
    \tikzfig{./axioms/CSXtop} = \tikzfig{./axioms/XCStop}
\end{lemma}

\begin{proof}
    \begin{gather*}
        \tikzfig{./axioms/CSXtop}\\
        = \tikzfig{./axioms/CSXtop-00}\\
        \eqeqref{lem:quad-diag-affine-transport} \tikzfig{./axioms/CSXtop-01}\\
        \eqeqref{lem:one-wire-cubic-diag-affine-transport} \tikzfig{./axioms/CSXtop-02}\\
        \eqdeuxeqref{lem:one-wire-cubic-diag-affine-transport}{lem:phase-diagonal-commutation} \tikzfig{./axioms/CSXtop-03}\\
        \eqdeuxeqref{lem:one-wire-cubic-diag-affine-transport}{lem:phase-diagonal-commutation} \tikzfig{./axioms/CSXtop-04}\\
        \eqeqref{axiom-XCNOTtop} \tikzfig{./axioms/CSXtop-05}\\
        \eqdeuxeqref{lem:one-wire-cubic-diag-affine-transport}{lem:phase-diagonal-commutation} \tikzfig{./axioms/CSXtop-06}\\
        \eqtroiseqref{axiom-mult1}{axiom-multxy}{axiom-multcnot2} \tikzfig{./axioms/CSXtop-07}\\
        \eqeqref{axiom-XCNOTtop} \tikzfig{./axioms/CSXtop-08}\\
        \eqeqref{axiom-XM} \tikzfig{./axioms/CSXtop-09}\\
        \eqtroiseqref{axiom-mult1}{axiom-multxy}{axiom-multcnot2}  \tikzfig{./axioms/CSXtop-10}\\
        \eqdeuxeqref{axiom-XCNOT}{axiom-Xd}  \tikzfig{./axioms/CSXtop-11}\\
        \eqeqref{axiom-XCNOTtop}  \tikzfig{./axioms/CSXtop-12}\\
        \eqtroiseqref{axiom-mult1}{axiom-multxy}{axiom-XM} \scalebox{0.8}{\tikzfig{./axioms/CSXtop-13}}\\
        \eqeqref{lem:one-wire-cubic-diag-affine-transport} \scalebox{0.8}{\tikzfig{./axioms/CSXtop-14}}\\
        \eqeqref{ax:S-order-d} \scalebox{0.8}{\tikzfig{./axioms/CSXtop-15}}\\
        \eqdeuxeqref{lem:one-wire-cubic-diag-affine-transport}{lem:phase-diagonal-commutation}\scalebox{0.8}{ \tikzfig{./axioms/CSXtop-16}}\\
        \eqdeuxeqref{lem:one-wire-cubic-diag-affine-transport}{ax:W-order-d} \scalebox{0.8}{\tikzfig{./axioms/CSXtop-17}}\\
        \eqeqref{lem:one-wire-cubic-diag-affine-transport} \scalebox{0.8}{\tikzfig{./axioms/CSXtop-18}}\\
        \eqeqref{lem:phase-diagonal-commutation} \scalebox{0.8}{\tikzfig{./axioms/CSXtop-19}}\\
        \eqeqref{lem:one-wire-cubic-diag-affine-transport} \scalebox{0.8}{\tikzfig{./axioms/CSXtop-20}}\\
        \eqeqref{ax:S-order-d} \scalebox{0.8}{\tikzfig{./axioms/CSXtop-21}}\\
        \eqeqref{lem:one-wire-cubic-diag-affine-transport} \scalebox{0.8}{\tikzfig{./axioms/CSXtop-22}}\\
        \eqeqref{lem:one-wire-cubic-diag-affine-transport} \scalebox{0.8}{\tikzfig{./axioms/CSXtop-23}}\\
        \eqeqref{ax:Z-order-d} \scalebox{0.8}{\tikzfig{./axioms/CSXtop-24}}\\
        \eqeqref{ax:S-order-d} \scalebox{0.8}{\tikzfig{./axioms/CSXtop-25}}\\
        \eqtroiseqref{axiom-mult1}{axiom-multxy}{axiom-XM} \tikzfig{./axioms/CSXtop-26}\\
        \eqeqref{axiom-XCNOTtop} \tikzfig{./axioms/CSXtop-27}\\
        \eqdeuxeqref{axiom-XCNOT}{axiom-Xd} \tikzfig{./axioms/CSXtop-28}\\
        \eqdeuxeqref{lem:phase-diagonal-commutation}{lem:one-wire-cubic-diag-affine-transport} \tikzfig{./axioms/CSXtop-29}\\
        \eqeqref{lem:one-wire-cubic-diag-affine-transport} \tikzfig{./axioms/CSXtop-30}\\
        \eqeqref{lem:phase-diagonal-commutation} \tikzfig{./axioms/CSXtop-31}\\
        \eqeqref{axiom-d-cnot}\tikzfig{./axioms/CSXtop-32}\\
        \eqtroiseqref{lem:phase-diagonal-commutation}{lem:one-wire-cubic-diag-affine-transport}{ax:S-order-d}  \scalebox{0.8}{\tikzfig{./axioms/CSXtop-33}}\\
        \eqeqref{ax:S-order-d} \scalebox{0.8}{\tikzfig{./axioms/CSXtop-34}}\\
        \eqeqref{ax:CZexp} \scalebox{0.9}{\tikzfig{./axioms/CSXtop-35}}\\
        \eqdeuxeqref{axiom-mult1}{axiom-multxy} \scalebox{0.9}{\tikzfig{./axioms/CSXtop-36}}\\
        \eqeqref{axiom-multcnot} \scalebox{0.9}{\tikzfig{./axioms/CSXtop-37}}\\
        \eqeqref{lem:one-wire-cubic-diag-affine-transport} \scalebox{0.9}{\tikzfig{./axioms/CSXtop-38}}\\
        \eqeqref{ax:CZexp} \scalebox{0.9}{\tikzfig{./axioms/CSXtop-39}}\\
        \eqeqref{lem:one-wire-cubic-diag-affine-transport} \scalebox{0.9}{\tikzfig{./axioms/CSXtop-40}}\\
        \eqeqref{lem:quad-diag-affine-transport}\scalebox{0.9}{\tikzfig{./axioms/CSXtop-41}}\\
        \eqdeuxeqref{lem:phase-diagonal-commutation}{ax:CZ-d}\scalebox{0.9}{\tikzfig{./axioms/CSXtop-42}}\\
        \eqeqref{ax:S-order-d} \scalebox{0.9}{\tikzfig{./axioms/CSXtop-43}}\\
        \eqdeuxeqref{lem:one-wire-cubic-diag-affine-transport}{ax:Z-order-d}\scalebox{0.9}{\tikzfig{./axioms/CSXtop-44}}\\
        \eqeqref{ax:PhaseGadgetT} \scalebox{0.9}{\tikzfig{./axioms/CSXtop-45}}\\
        \eqtroiseqref{lem:one-wire-cubic-diag-affine-transport}{ax:Z-order-d}{lem:phase-diagonal-commutation} \scalebox{0.9}{\tikzfig{./axioms/CSXtop-46}}\\
        \eqeqref{lem:one-wire-cubic-diag-affine-transport}  \scalebox{0.9}{\tikzfig{./axioms/CSXtop-47}}\\
        \eqeqref{ax:PhaseGadgetT} \scalebox{0.9}{\tikzfig{./axioms/CSXtop-48}}\\
        = \scalebox{0.9}{\tikzfig{./axioms/XCStop}}\\
    \end{gather*}
    \endgroup
\end{proof}

\begin{lemma}\label{ax:CCZ-X}
    CubicPhase$_d$ $\vdash$ 
    \tikzfig{./axioms/XCCZ} = \tikzfig{./axioms/CCZX}
\end{lemma}

\begin{proof}
    \begin{gather*}
        \tikzfig{./axioms/XCCZ}\\
        = \tikzfig{./axioms/CCZX-00}\\
        \eqdeuxeqref{lem:one-wire-cubic-diag-affine-transport}{axiom-XCNOTtop}  \scalebox{0.8}{\tikzfig{./axioms/CCZX-01}}\\
        \eqeqref{axiom-d-cnot} \scalebox{0.8}{\tikzfig{./axioms/CCZX-02}}\\
        \eqeqref{ax:PhaseGadgetS} \scalebox{0.8}{\tikzfig{./axioms/CCZX-03}}\\
        \eqeqref{lemma-XX2-comm} \scalebox{0.8}{\tikzfig{./axioms/CCZX-04}}\\
        \eqeqref{ax:PhaseGadgetS} \scalebox{0.8}{\tikzfig{./axioms/CCZX-05}}\\
        \eqeqref{lemma-XX2-comm}  \scalebox{0.8}{\tikzfig{./axioms/CCZX-06}}\\
        \eqeqref{axiom-d-cnot} \scalebox{0.8}{\tikzfig{./axioms/CCZX-07}}\\
        \eqeqref{lem:one-wire-cubic-diag-affine-transport} \scalebox{0.8}{\tikzfig{./axioms/CCZX-08}}\\
        \eqeqref{axiom-d-cnot} \scalebox{0.8}{\tikzfig{./axioms/CCZX-09}}\\
        \eqeqref{axiom-I} \scalebox{0.8}{\tikzfig{./axioms/CCZX-10}}\\
        \eqeqref{axiom-d-cnot} \scalebox{0.8}{\tikzfig{./axioms/CCZX-11}}\\
        \eqeqref{ax:PhaseGadgetS} \scalebox{0.8}{\tikzfig{./axioms/CCZX-12}}\\
        \eqeqref{lem:one-wire-cubic-diag-affine-transport} \scalebox{0.8}{\tikzfig{./axioms/CCZX-13}}\\
        \eqeqref{ax:PhaseGadgetS} \scalebox{0.8}{\tikzfig{./axioms/CCZX-15}}\\
        \eqdeuxeqref{axiom-d-cnot}{axiom-d-cnot} \scalebox{0.8}{\tikzfig{./axioms/CCZX-17}}\\
        \eqeqref{axiom-I}\scalebox{0.8}{\tikzfig{./axioms/CCZX-18}}\\
        \eqeqref{lemma-XX2-comm}\scalebox{0.8}{\tikzfig{./axioms/CCZX-19}}\\
        \eqdeuxeqref{ax:PhaseGadgetS}{ax:PhaseGadgetT}  \scalebox{0.8}{\tikzfig{./axioms/CCZX-20}}\\
        \eqeqref{lemma-XX2-comm}  \scalebox{0.8}{\tikzfig{./axioms/CCZX-21}}\\
        \eqdeuxeqref{ax:PhaseGadgetS}{ax:PhaseGadgetT} \scalebox{0.8}{\tikzfig{./axioms/CCZX-22}}\\
        \eqeqref{lemma-XX2-comm}  \scalebox{0.8}{\tikzfig{./axioms/CCZX-23}}\\
        \eqeqref{ax:T-order-d} \scalebox{0.8}{\tikzfig{./axioms/CCZX-24}}\\
        = \tikzfig{./axioms/CCZX-25}\\
        \eqeqref{axiom-d-cnot} \tikzfig{./axioms/CCZX-26}\\
        \eqeqref{ax:S-order-d} \tikzfig{./axioms/CCZX-27}\\
        \eqeqref{lem:quad-diag-affine-transport} \tikzfig{./axioms/CCZX-28}\\
        \eqeqref{axiom-d-cnot} \tikzfig{./axioms/CCZX-29}\\
        \eqdeuxeqref{lem:phase-diagonal-commutation}{lem:one-wire-cubic-diag-affine-transport} \tikzfig{./axioms/CCZX-30}\\
        = \tikzfig{./axioms/CCZX-31}\\
        \eqdeuxeqref{ax:S-order-d}{axiom-d-cnot} \tikzfig{./axioms/CCZX-32}\\
        \eqdeuxeqref{ax:T-order-d}{lem:phase-diagonal-commutation} \tikzfig{./axioms/CCZX-33}
        = \tikzfig{./axioms/CCZX}
    \end{gather*}
\end{proof}


\medskip\noindent\textbf{Controlled affine transport for controlled cubic phases.}

\begin{lemma}\label{ax:CS-CNOT-rev}
    CubicPhase$_d$ $\vdash$
    \tikzfig{./axioms/CSCNOTrev}=\tikzfig{./axioms/CNOTCSrev}
\end{lemma}

\begin{proof}
    \begin{gather*}
        \tikzfig{./axioms/CSCNOTrev}\\
        = \tikzfig{./axioms/CSCNOTrev-00}\\
        \eqdeuxeqref{lem:phase-diagonal-commutation}{lem:quad-diag-affine-transport} \tikzfig{./axioms/CSCNOTrev-01}\\
        \eqeqref{lem:one-wire-cubic-diag-affine-transport} \tikzfig{./axioms/CSCNOTrev-02}\\
        \eqeqref{ax:Z-order-d} \tikzfig{./axioms/CSCNOTrev-03}\\
        \eqeqref{lem:one-wire-cubic-diag-affine-transport} \tikzfig{./axioms/CSCNOTrev-04}\\
        \eqeqref{lem:one-wire-cubic-diag-affine-transport}  \tikzfig{./axioms/CSCNOTrev-05}\\
        \eqeqref{ax:CZ-d} \tikzfig{./axioms/CSCNOTrev-06}\\
        \eqeqref{ax:S-order-d} \tikzfig{./axioms/CSCNOTrev-07}\\
        \eqeqref{ax:PhaseGadgetT} \tikzfig{./axioms/CSCNOTrev-08}\\
        \eqeqref{axiom-multcnot} \tikzfig{./axioms/CSCNOTrev-09}\\
        \eqeqref{lem:one-wire-cubic-diag-affine-transport} \tikzfig{./axioms/CSCNOTrev-10}\\
        \eqeqref{axiom-multcnot} \tikzfig{./axioms/CSCNOTrev-11}\\
        \eqeqref{axiom-d-cnot} \tikzfig{./axioms/CSCNOTrev-12}\\
        \eqeqref{ax:CS-M-bot} \tikzfig{./axioms/CSCNOTrev-13}\\
        \eqeqref{ax:CZ-d} \tikzfig{./axioms/CSCNOTrev-14}\\
        \eqeqref{ax:CS-order-d} \tikzfig{./axioms/CSCNOTrev-15}\\
        \eqeqref{ax:CS-M-top} \tikzfig{./axioms/CSCNOTrev-16}\\
        \eqeqref{ax:CS-order-d} \tikzfig{./axioms/CSCNOTrev-17}\\
        \eqeqref{ax:T-order-d} \tikzfig{./axioms/CSCNOTrev-18}\\
        \eqeqref{lem:one-wire-cubic-diag-affine-transport} \tikzfig{./axioms/CSCNOTrev-19}\\
        \eqeqref{ax:T-order-d} \tikzfig{./axioms/CSCNOTrev-20}\\
        \eqeqref{lem:one-wire-cubic-diag-affine-transport} \tikzfig{./axioms/CSCNOTrev-21}\\
        \eqdeuxeqref{axiom-multxy}{axiom-mult1} \tikzfig{./axioms/CSCNOTrev-22}\\
        \eqeqref{ax:T-order-d} \tikzfig{./axioms/CSCNOTrev-23}
        \eqeqref{ax:S-order-d} \tikzfig{./axioms/CNOTCSrev}
    \end{gather*}
\end{proof}

\begin{lemma}\label{ax:CS-CNOT}
    CubicPhase$_d$ $\vdash$
    \tikzfig{./axioms/CSCNOT} = \tikzfig{./axioms/CNOTCS}
\end{lemma}

\begin{proof}
    \begin{gather*}
        \tikzfig{./axioms/CSCNOT}\\
        = \tikzfig{./axioms/CSCNOT-00}\\
        \eqeqref{lem:quad-diag-affine-transport} \tikzfig{./axioms/CSCNOT-01}\\
        \eqeqref{lem:one-wire-cubic-diag-affine-transport} \tikzfig{./axioms/CSCNOT-02}\\
        \eqeqref{ax:Z-order-d} \tikzfig{./axioms/CSCNOT-03}\\
        \eqeqref{ax:S-order-d} \tikzfig{./axioms/CSCNOT-04}\\
        \eqdeuxeqref{axiom-multcnot}{axiom-multcnot} \tikzfig{./axioms/CSCNOT-06}\\
        \eqeqref{lem:one-wire-cubic-diag-affine-transport} \tikzfig{./axioms/CSCNOT-07}\\
        \eqeqref{axiom-d-cnot} \tikzfig{./axioms/CSCNOT-08}\\
        \eqeqref{ax:T-order-d} \tikzfig{./axioms/CSCNOT-09}\\
        \eqeqref{ax:CS-M-bot} \tikzfig{./axioms/CSCNOT-10}\\
        \eqeqref{ax:CS-order-d} \tikzfig{./axioms/CSCNOT-11}\\
        \eqeqref{ax:CS-M-top} \tikzfig{./axioms/CSCNOT-12}\\
        \eqeqref{ax:CZ-d} \tikzfig{./axioms/CSCNOT-13}\\
        \eqeqref{ax:CS-order-d} \tikzfig{./axioms/CSCNOT-14}\\
        \eqeqref{lem:one-wire-cubic-diag-affine-transport} \tikzfig{./axioms/CSCNOT-15}\\
        \eqdeuxeqref{axiom-mult1}{axiom-multxy} \tikzfig{./axioms/CSCNOT-16}
        \eqdeuxeqref{ax:T-order-d}{ax:S-order-d} \tikzfig{./axioms/CNOTCS}
    \end{gather*}
\end{proof}

\begin{lemma}\label{ax:CCZ-CNOT}
    CubicPhase$_d$ $\vdash$
    \tikzfig{./axioms/CCZCNOT} = \tikzfig{./axioms/CNOTCCZ}
\end{lemma}

\begin{proof}
    \begin{gather*}
        \tikzfig{./axioms/CCZCNOT}\\
        = \tikzfig{./axioms/CCZCNOT-00}\\
        \eqdeuxeqref{axiom-d-cnot}{lemma-XX2-comm} \tikzfig{./axioms/CCZCNOT-01}\\
        \eqeqref{lem:one-wire-cubic-diag-affine-transport} \tikzfig{./axioms/CCZCNOT-02}\\
        \eqeqref{lem:one-wire-cubic-diag-affine-transport} \scalebox{0.9}{\tikzfig{./axioms/CCZCNOT-03}}\\
        \eqeqref{axiom-d-cnot} \tikzfig{./axioms/CCZCNOT-04}\\
        \eqeqref{lem:one-wire-cubic-diag-affine-transport} \tikzfig{./axioms/CCZCNOT-05}\\
        \eqeqref{lem:one-wire-cubic-diag-affine-transport} \scalebox{0.9}{\tikzfig{./axioms/CCZCNOT-06}}\\
        \eqeqref{axiom-d-cnot} \tikzfig{./axioms/CCZCNOT-07}\\
        \eqdeuxeqref{lem:one-wire-cubic-diag-affine-transport}{ax:T-order-d} \tikzfig{./axioms/CCZCNOT-08}\\
        \eqeqref{lem:one-wire-cubic-diag-affine-transport} \scalebox{0.9}{\tikzfig{./axioms/CCZCNOT-09}}\\
        \eqeqref{axiom-d-cnot} \tikzfig{./axioms/CCZCNOT-10}\\
        \eqeqref{ax:CS-CNOT-bot2} \tikzfig{./axioms/CCZCNOT-11}\\
        \eqeqref{ax:CS-CNOT-bot} \scalebox{0.9}{\tikzfig{./axioms/CCZCNOT-12}}\\
        \eqeqref{lem:one-wire-cubic-diag-affine-transport} \scalebox{0.85}{\tikzfig{./axioms/CCZCNOT-13}}\\
        \eqeqref{axiom-d-cnot} \scalebox{0.9}{\tikzfig{./axioms/CCZCNOT-14}}\\
        \eqdeuxeqref{lem:one-wire-cubic-diag-affine-transport}{ax:T-order-d} \tikzfig{./axioms/CCZCNOT-15}\\
        \eqdeuxeqref{ax:CS-order-d}{ax:CS-order-d} \tikzfig{./axioms/CCZCNOT-17}\\
        \eqeqref{lem:one-wire-cubic-diag-affine-transport} \tikzfig{./axioms/CCZCNOT-18}\\
        \eqeqref{axiom-d-cnot} \tikzfig{./axioms/CCZCNOT-19}\\
        \eqeqref{ax:T-order-d} \tikzfig{./axioms/CCZCNOT-20}\\
        \eqtroiseqref{axiom-mult1}{axiom-multxy}{axiom-multcnot} \tikzfig{./axioms/CCZCNOT-21}\\
        \eqeqref{ax:CS-CNOT-bot2} \tikzfig{./axioms/CCZCNOT-22}\\
        \eqeqref{ax:CS-CNOT-bot} \tikzfig{./axioms/CCZCNOT-23}\\
        \eqeqref{axiom-d-cnot} \tikzfig{./axioms/CCZCNOT-24}\\
        \eqeqref{ax:CCZ-M} \tikzfig{./axioms/CCZCNOT-25}\\
        \eqeqref{ax:CS-M-bot} \tikzfig{./axioms/CCZCNOT-26}\\
        \eqeqref{ax:CS-M-top}  \tikzfig{./axioms/CCZCNOT-27}\\
        \eqdeuxeqref{axiom-mult1}{axiom-multxy} \tikzfig{./axioms/CCZCNOT-28}\\
        \eqeqref{ax:CCZ-d} \tikzfig{./axioms/CCZCNOT-29}\\
        \eqtroiseqref{ax:CS-order-d}{ax:CS-order-d}{ax:CS-order-d} \tikzfig{./axioms/CNOTCCZ}
    \end{gather*}
\end{proof}

\subsection{Finite Case Tables for Local Certificates}\label{derivation:transport-tables}

\Cref{tab:diag-comm-lemmas,tab:diag-through-affine} index the finite local certificates used in \cref{lem:phase-diagonal-commutation,lem:phase-one-step-transport}.

\begin{proof}[Transport certificates for \cref{lem:phase-one-step-transport}]
\Cref{tab:diag-through-affine} lists the non-scalar cubic families against \(\gX\), \(\gCX\), and \(\gM\); lower fragments are obtained by restriction, and \(\gW\) is omitted because scalar transport is centrality.
Entries with several references record the required support or orientation cases.
Entries in the \(\gSC\) column that cite \(\gCS\)-labelled lemmas are read after unfolding \(\gSC\) as the conjugate of \(\gCS\) by the two-wire symmetry.
\end{proof}

\newcommand{\tblcell}[1]{\begin{tabular}[c]{@{}c@{}}#1\end{tabular}}

\begin{table}[p]
\centering
\scriptsize
\renewcommand{\arraystretch}{1.35}
\begin{tabular}{c|*{7}{>{\centering\arraybackslash}p{1.1cm}}}
  & $\gZ$ & $\gS$ & $\gT$ & $\gCZ$ & $\gCS$ & $\gSC$ & $\gCCZ$ \\
\hline
$\gZ$
  & --- 
  & $\eqref{ax:SZ}$
  & $\eqref{ax:TZ}$
  & $\eqref{ax:ZCZ}$
  & \tblcell{$\eqref{ax:ZCStop}$\\$\eqref{ax:ZCSbot}$}
  & \tblcell{$\eqref{ax:ZCStop}$\\$\eqref{ax:ZCSbot}$}
  & $\eqref{ax:ZCCZ}$
\\
\hline
$\gS$
  & 
  & --- 
  & $\eqref{ax:TS}$
  & $\eqref{ax:SCZ}$
  & \tblcell{$\eqref{ax:SCStop}$\\$\eqref{ax:SCSbot}$}
  & \tblcell{$\eqref{ax:SCStop}$\\$\eqref{ax:SCSbot}$}
  & $\eqref{ax:SCCZ}$
\\
\hline
$\gT$
  & 
  & 
  & --- 
  & $\eqref{ax:TCZ}$
  & \tblcell{$\eqref{ax:TCStop}$\\$\eqref{ax:TCSbot}$}
  & \tblcell{$\eqref{ax:TCStop}$\\$\eqref{ax:TCSbot}$}
  & $\eqref{ax:TCCZ}$
\\
\hline
$\gCZ$
  & 
  & 
  & 
  & \tblcell{$\eqref{ax:CZCZ}$}
  & \tblcell{$\eqref{ax:CSCZ}$\\$\eqref{ax:CSCZtop}$\\$\eqref{ax:CSCZbot}$}
  & \tblcell{$\eqref{ax:CSCZ}$\\$\eqref{ax:CSCZtop}$\\$\eqref{ax:CSCZbot}$}
  & \tblcell{$\eqref{ax:CCZCZ}$\\$\eqref{ax:CCZCZtop}$}
\\
\hline
$\gCS$
  & 
  & 
  & 
  & 
  & \tblcell{$\eqref{ax:CSCS}$\\$\eqref{ax:CSCS2}$\\$\eqref{ax:CSCS3}$}
  & \tblcell{$\eqref{ax:CSSC}$\\$\eqref{ax:CSCS}$\\$\eqref{ax:CSCS2}$\\$\eqref{ax:CSCS3}$}
  & \tblcell{$\eqref{ax:CCZ-CS-1}$\\$\eqref{ax:CCZ-CS-2}$\\$\eqref{ax:CCZ-CS-3}$}
\\
\hline
$\gSC$
  & 
  & 
  & 
  & 
  & 
  & \tblcell{$\eqref{ax:CSCS}$\\$\eqref{ax:CSCS2}$\\$\eqref{ax:CSCS3}$}
  & \tblcell{$\eqref{ax:CCZ-CS-1}$\\$\eqref{ax:CCZ-CS-2}$\\$\eqref{ax:CCZ-CS-3}$}
\\
\hline
$\gCCZ$
  & 
  & 
  & 
  & 
  & 
  & 
  & \tblcell{$\eqref{ax:CCZCCZ}$\\$\eqref{ax:CCZCCZ2}$}
\end{tabular}
\caption{Finite case table for diagonal commutation. Entries cite the displayed axiom or derived lemma for each non-scalar cubic family pair; multiple entries record distinct overlap or orientation cases.}
\label{tab:diag-comm-lemmas}
\end{table}

\begin{table}[p]
\centering
\scriptsize
\renewcommand{\arraystretch}{1.35}
\begin{tabular}{c|*{7}{>{\centering\arraybackslash}p{1.1cm}}}
  & $\gZ$ & $\gS$ & $\gT$ & $\gCZ$ & $\gCS$ & $\gSC$ & $\gCCZ$ \\
\hline
$\gX$
  & $\eqref{ax:ZX}$
  & $\eqref{ax:SX}$
  & $\eqref{ax:TX}$
  & $\eqref{ax:CZ-X}$
  & \tblcell{$\eqref{ax:CS-X-top}$\\$\eqref{ax:CS-X-bot}$}
  & \tblcell{$\eqref{ax:CS-X-top}$\\$\eqref{ax:CS-X-bot}$}
  & $\eqref{ax:CCZ-X}$
\\
\hline
$\gCX$
  & \tblcell{$\eqref{ax:ZC-target}$\\$\eqref{ax:ZC-control}$}
  & \tblcell{$\eqref{ax:SC-control}$\\$\eqref{ax:SC-target}$}
  & \tblcell{$\eqref{ax:TC-control}$\\$\eqref{ax:TC-target}$}
  & \tblcell{$\eqref{ax:CZ-CX}$\\$\eqref{ax:CC-control}$\\$\eqref{ax:CC-target}$}
  & \tblcell{$\eqref{ax:CS-CNOT}$\\$\eqref{ax:CS-CNOT-rev}$\\$\eqref{ax:CS-CNOT-bot}$\\$\eqref{ax:CS-CNOT-bot2}$\\$\eqref{ax:CS-CNOT-top}$\\$\eqref{ax:CS-CNOT-top2}$}
  & \tblcell{$\eqref{ax:CS-CNOT}$\\$\eqref{ax:CS-CNOT-rev}$\\$\eqref{ax:CS-CNOT-bot}$\\$\eqref{ax:CS-CNOT-bot2}$\\$\eqref{ax:CS-CNOT-top}$\\$\eqref{ax:CS-CNOT-top2}$}
  & \tblcell{$\eqref{ax:CCZ-CNOT}$\\$\eqref{ax:CCZ-CNOTtop}$\\$\eqref{ax:CCZ-CNOTbot}$}
\\
\hline
$\gM$
  & $\eqref{ax:ZM}$
  & $\eqref{ax:SM}$
  & $\eqref{ax:TM}$
  & $\eqref{ax:CZ-M}$
  & \tblcell{$\eqref{ax:CS-M-top}$\\$\eqref{ax:CS-M-bot}$}
  & \tblcell{$\eqref{ax:CS-M-top}$\\$\eqref{ax:CS-M-bot}$}
  & $\eqref{ax:CCZ-M}$
\end{tabular}
\caption{Finite case table for diagonal-affine transport. Entries cite the displayed axiom or derived lemma for commuting a non-scalar cubic family past one placed affine generator; multiple entries record target/control, top/bottom, or orientation cases.}
\label{tab:diag-through-affine}
\end{table}

\clearpage

\section{Group theoretic presentations for reversible linear and affine transformations in prime dimension \texorpdfstring{\(d\)}{d}}
\label{sec:cnot-prime-d}

Fix a prime \(d\).
The affine circuit presentation rests on a group-theoretic calculation over \(\F_d\).
The relation family in \cref{def:lin-aff-relations} is stated group-theoretically; its instances are the circuit rules and diagrammatic axioms used by the affine fragment.
From this list we obtain complete presentations of \(\SL_n(\F_d)\), \(\GL_n(\F_d)\), and \(\AGL_n(\F_d)\) by restricting generators and relations.

On computational basis labels \(x\in\F_d^n\), an affine reversible circuit acts by \(x\mapsto Ax+b\) with \(A\in\GL_n(\F_d)\) and \(b\in\F_d^n\).
The linear part \(A\) lies in \(\SL_n(\F_d)\) when \(\det(A)=1\), and in \(\GL_n(\F_d)\) in general; adjoining translations gives \(\AGL_n(\F_d)\cong \F_d^n\rtimes \GL_n(\F_d)\).
When \(d=2\), one has \(\F_2^\times=\{1\}\), hence \(\GL_n(\F_2)=\SL_n(\F_2)\).
For odd prime \(d\), diagonal scalings must be included explicitly.

\subsection{Conventions and notation}
\label{subsec:lin-affine-conventions}

Identify \(\F_d\) with \(\Z/d\Z\), and write \(\F_d^\times\) for its multiplicative group.
Vectors in \(\F_d^n\) are column vectors and matrices act on the left.
Write \(I_n\) for the \(n\times n\) identity matrix.
For \(1\le i,j\le n\), let \(e_{i,j}\) be the matrix unit with a single \(1\) in row \(i\), column \(j\).

\begin{remark}\label{rem:exp-convention}
If \(g\) satisfies \(g^d=\id\) and \(t\in\F_d\), we write \(g^t:=g^{\tilde t}\) where \(\tilde t\in\{0,\dots,d-1\}\) is the chosen integer representative.
For \(a,b\in\F_d\) one has \(g^{a+b}=g^ag^b\) since exponents are computed modulo \(d\).
We state relations of the form \(g^d=\id\) explicitly because the diagrammatic translation needs named finite-order rules for exponent reduction.
\end{remark}

\subsection{Generators in matrices}

\begin{definition}\label{def:CX}
For distinct indices \(i\neq j\), let \(\CNOT_{ij}\in\GL_n(\F_d)\) be the elementary transvection that sends \(x_j\mapsto x_j+x_i\) and fixes \(x_k\) for \(k\neq j\).
Equivalently, \(\CNOT_{ij}=I_n+e_{j,i}\).
\end{definition}

\begin{definition}\label{def:M}
For \(k\in\F_d^\times\), let \(\M_1(k)=\mathrm{diag}(k,1,\dots,1)\in\GL_n(\F_d)\).
\end{definition}

\begin{definition}\label{def:X-gate}
For \(1\le i\le n\), let \(X_i\in\AGL_n(\F_d)\) be translation in the \(i\)th coordinate, \(x_i\mapsto x_i+1\), fixing the other coordinates.
Equivalently, \(X_i=(I_n,e_i)\) where \(e_i\) is the \(i\)th standard basis vector.
\end{definition}

\subsection{A single list of relations}

Fix \(n\ge 1\).
Let \(\{x_{ij}\mid 1\le i\neq j\le n\}\), \(\{m(k)\mid k\in\F_d^\times\}\), and \(\{t_i\mid 1\le i\le n\}\) be abstract generators.
For the relation list, all parameters \(t,u,b_i\) are in \(\F_d\), and all products \(ab\), \(k\ell\) are computed in \(\F_d\).

\begin{definition}\label{def:lin-aff-relations}
Write \(\mathsf{R}_{\SL}(n,d)\) for the following relations on the generators \(x_{ij}\):
\begin{align}
  x_{ij}^{d} &= \id && (i\neq j),
  \tag{Stein1}\label{ax:add}\\
  x_{ik}^{tu}\,x_{ij}^{t}\,x_{jk}^{u} &= x_{jk}^{u}\,x_{ij}^{t}
  && (i,j,k \text{ distinct}),
  \tag{Stein2}\label{ax:stein1}\\
  x_{ij}\,x_{k\ell} &= x_{k\ell}\,x_{ij}
  && (i\neq \ell,\ j\neq k).
  \tag{Stein3}\label{ax:stein2}
\end{align}
When \(n=2\), the commutator relation \eqref{ax:stein1} is unavailable because it requires three distinct indices, and \eqref{ax:stein2} does not relate the two opposite root subgroups generated by \(x_{12}\) and \(x_{21}\).
We therefore add the usual rank-\(1\) torus and Weyl relations.
For \(a\in\F_d^\times\), define the words \(w(a)\coloneqq x_{12}^{-a^{-1}}\,x_{21}^{a}\,x_{12}^{-a^{-1}}\) and \(h(a)\coloneqq w(a)\,w(1)^{-1}\).
Under \(\Phi^\mathrm{SL}\) these are the matrices
\begin{equation*}
  w(a)=
  \begin{psmallmatrix}
  0 & a\\
  -a^{-1} & 0
  \end{psmallmatrix},
  \qquad
  h(a)=
  \begin{psmallmatrix}
  a & 0\\
  0 & a^{-1}
  \end{psmallmatrix}.
\end{equation*}
We impose
\begin{align}
  h(a)\,h(b) &= h(ab) && (a,b\in\F_d^\times),
  \tag{Torus}\label{ax:torus}\\
  w(1)\,x_{21}\,w(1)^{-1} &= x_{12}^{-1}.
  \tag{Weyl}\label{ax:weyl}
\end{align}

Write \(\mathsf{R}_{\GL}(n,d)\) for \(\mathsf{R}_{\SL}(n,d)\) together with the relations involving \(m(\cdot)\):
\begin{align}
  m(1) &= \id,
  \tag{mult1}\label{ax:mult1}\\
  m(k)\,m(\ell) &= m(k\ell)
  && (k,\ell\in \F_{d}^\times),
  \tag{multM}\label{ax:mult}\\
  m(x)\,x_{i1} &= x_{i1}^{x}\,m(x)
  && (i\ne 1,\ x\in \F_{d}^\times),
  \tag{MCx}\label{ax:diag-on-target}\\
  m(x)\,x_{1j}^{x} &= x_{1j}\,m(x)
  && (j\ne 1,\ x\in \F_{d}^\times),
  \tag{MxC}\label{ax:diag-on-control}\\
  m(x)\,x_{ij} &= x_{ij}\,m(x)
  && (1\notin\{i,j\},\ x\in \F_{d}^\times).
  \tag{MCsep}\label{ax:diag-on-other}
\end{align}

Write \(\mathsf{R}_{\AGL}(n,d)\) for \(\mathsf{R}_{\GL}(n,d)\) together with the relations involving \(t_i\):
\begin{align}
  t_i^d &= \id && (1\le i\le n),
  \tag{Xd}\label{ax:X-order-d}\\
  t_i\,t_j &= t_j\,t_i && (i\neq j),
  \tag{Xcom}\label{ax:X-commute}\\
  m(k)\,t_1 &= t_1^{k}\,m(k) && (k\in\F_d^\times),
  \tag{XM}\label{ax:X-M-same-wire}\\
  m(k)\,t_i &= t_i\,m(k) && (i\neq 1,\ k\in\F_d^\times),
  \tag{XMsep}\label{ax:X-M-other-wire}\\
  t_j\,x_{ij} &= x_{ij}\,t_j && (i\neq j),
  \tag{XCx}\label{ax:X-CX-target}\\
  x_{ij}\,t_i &= t_i\,t_j\,x_{ij} && (i\neq j),
  \tag{XxC}\label{ax:X-CX-control}\\
  t_s\,x_{ij} &= x_{ij}\,t_s && (s\notin\{i,j\}).
  \tag{XCxsep}\label{ax:X-CX-other-wire}
\end{align}
\end{definition}

\subsection{Three abstract groups and three comparison maps}

\begin{definition}\label{def:SLn-presentation-map}
Let \(\St_n(\F_d)\) be the group on generators \(\{x_{ij}\mid i\neq j\}\) modulo \(\mathsf{R}_{\SL}(n,d)\).
Define \(\Phi^\mathrm{SL}:\St_n(\F_d)\to \SL_n(\F_d)\) by \(\Phi^\mathrm{SL}(x_{ij})=\CNOT_{ij}\).
\end{definition}

\begin{definition}\label{def:GLn-presentation-map}
Let \(G_n^{\mathrm{lin}}\) be the group on generators \(\{x_{ij}\mid i\neq j\}\cup\{m(k)\mid k\in\F_d^\times\}\) modulo \(\mathsf{R}_{\GL}(n,d)\).
Define \(\Phi^\mathrm{GL}:G_n^{\mathrm{lin}}\to \GL_n(\F_d)\) by \(\Phi^\mathrm{GL}(x_{ij})=\CNOT_{ij}\) and \(\Phi^\mathrm{GL}(m(k))=\M_1(k)\).
\end{definition}

\begin{definition}\label{def:AGLn-presentation}
Let \(G_n^{\mathrm{aff}}\) be the group on generators
\(
\{x_{ij}\mid i\neq j\}\cup\{m(k)\mid k\in\F_d^\times\}\cup\{t_i\mid 1\le i\le n\}
\)
modulo \(\mathsf{R}_{\AGL}(n,d)\).
Define \(\Phi^\mathrm{AGL}:G_n^{\mathrm{aff}}\to \AGL_n(\F_d)\) by
\(\Phi^\mathrm{AGL}(x_{ij}) = (\CNOT_{ij},0)\), \(\Phi^\mathrm{AGL}(m(k)) = (\M_1(k),0)\), and \(\Phi^\mathrm{AGL}(t_i) = (I_n,e_i)\).
\end{definition}

\subsection{Presentations}

\begin{theorem}\label{thm:SLn-presentation-uniform}
The homomorphism \(\Phi^\mathrm{SL}\) is an isomorphism.
\end{theorem}

\begin{proof}
The defining relations \(\mathsf{R}_{\SL}(n,d)\) hold in \(\SL_n(\F_d)\), hence \(\Phi^\mathrm{SL}\) is well-defined.
We split the proof according to rank.

If \(n=1\), the generator set is empty and \(\SL_1(\F_d)=\{1\}\), so \(\Phi^\mathrm{SL}\) is an isomorphism.
If \(n=2\), this is the type-\(A_1\) Chevalley presentation: \eqref{ax:add} is root-subgroup additivity, \eqref{ax:weyl} is the rank-one Weyl conjugation relation, and \eqref{ax:torus} is the torus multiplicativity relation.
Steinberg's rank-one presentation theorem gives the simply connected Chevalley group of type \(A_1\), which is \(\SL_2(\F_d)\); see \cite[\S6, Lem.~37 and Thm.~8--9]{steinberg_lectures_2016} and \cite{matsumoto1969}.
Under \(\Phi^\mathrm{SL}\), the generators \(x_{12}\) and \(x_{21}\) are sent to the corresponding elementary transvections, so \(\Phi^\mathrm{SL}\) is exactly this standard comparison map and is an isomorphism.

Assume now that \(n\ge3\).
Its image contains all transvections \(\CNOT_{ij}\), hence contains \(E_n(\F_d)\), and \(E_n(\F_d)=\SL_n(\F_d)\) by \cite[\S7, p.~91]{steinberg_lectures_2016}, so \(\Phi^\mathrm{SL}\) is surjective.
The relations \(\mathsf{R}_{\SL}(n,d)\) are Steinberg's type-\(A_{n-1}\) relations for the simply connected Chevalley group.
By the Steinberg--Chevalley analysis, \(\ker(\Phi^\mathrm{SL})\) is central and governed by \(K_2(\F_d)\); see \cite[\S6, Thm.~8--9]{steinberg_lectures_2016} and \cite{matsumoto1969}.
For a finite field, \(K_2(\F_d)=0\) by Quillen \cite{quillen1973higherk}, hence \(\ker(\Phi^\mathrm{SL})=\{1\}\) and \(\Phi^\mathrm{SL}\) is injective.
\end{proof}

\begin{theorem}\label{thm:SLn-presentation}
If \(n\ge 3\), then \(\mathsf{R}_{\SL}(n,d)\) may use only \eqref{ax:add}--\eqref{ax:stein2}, since \eqref{ax:torus}--\eqref{ax:weyl} follow from \eqref{ax:add}--\eqref{ax:stein2} in type \(A_{n-1}\) \cite[\S6, Thm.~9]{steinberg_lectures_2016}.
\end{theorem}

\begin{corollary}\label{thm:GLn-presentation}
The homomorphism \(\Phi^\mathrm{GL}\) is an isomorphism.
In particular, \(\mathsf{R}_{\GL}(n,d)\) presents \(\GL_n(\F_d)\).
\end{corollary}

\begin{proof}
Surjectivity: let \(A\in\GL_n(\F_d)\) and set \(x=\det(A)\).
Then \(S=\mathrm{diag}(x^{-1},1,\dots,1)A\in\SL_n(\F_d)\).
By \cref{thm:SLn-presentation-uniform}, \(S\) is a product of transvections, so \(A=\Phi^\mathrm{GL}(m(x))\Phi^\mathrm{GL}(U)\) for some word \(U\) in the \(x_{ij}\).

Injectivity: using \eqref{ax:diag-on-target}--\eqref{ax:diag-on-other}, commute all \(m(\cdot)\) generators to the far left, and use \eqref{ax:mult} to merge them, so every element of \(G_n^{\mathrm{lin}}\) is represented by a word \(m(x)U\) with \(U\) in the \(x_{ij}\).
If \(\Phi^\mathrm{GL}(m(x)U)=I_n\), taking determinants gives \(x=1\), hence \(\Phi^\mathrm{SL}(U)=I_n\).
By \cref{thm:SLn-presentation-uniform}, \(U=\id\) in \(\St_n(\F_d)\), hence \(m(x)U=\id\) in \(G_n^{\mathrm{lin}}\).
\end{proof}

\begin{corollary}\label{thm:AGLn-presentation}
The homomorphism \(\Phi^\mathrm{AGL}\) is an isomorphism.
In particular, \(\mathsf{R}_{\AGL}(n,d)\) presents \(\AGL_n(\F_d)\).
\end{corollary}

\begin{proof}
Surjectivity: every \((A,b)\in\AGL_n(\F_d)\) factors as \((I_n,b)\circ(A,0)\).
Write \(b=(b_1,\dots,b_n)\).
Since \((I_n,b)=\Phi^\mathrm{AGL}(t_1^{b_1}\cdots t_n^{b_n})\) and \((A,0)\) lies in the image by \cref{thm:GLn-presentation}, the map \(\Phi^\mathrm{AGL}\) is surjective.

Injectivity: using \eqref{ax:X-M-same-wire}--\eqref{ax:X-CX-other-wire}, commute all \(t_i\) generators to the far left, collecting them with \eqref{ax:X-order-d}--\eqref{ax:X-commute} into a word \(t_1^{b_1}\cdots t_n^{b_n}U\), where \(U\) is a word in \(\{x_{ij},m(k)\}\).
If \(\Phi^\mathrm{AGL}(t_1^{b_1}\cdots t_n^{b_n}U)=\id\), then \(\Phi^\mathrm{AGL}(U)=(A,0)\) and \(\Phi^\mathrm{AGL}(t_1^{b_1}\cdots t_n^{b_n})=(I_n,b)\) satisfy \((A,b)=(I_n,0)\), hence \(A=I_n\) and \(b=0\).
By \cref{thm:GLn-presentation}, \(A=I_n\) forces \(U=\id\).
By \eqref{ax:X-order-d}--\eqref{ax:X-commute}, \(b=0\) forces all \(b_i=0\), hence the translation word is \(\id\).
\end{proof}

\subsection{Beyond finite fields: presentations over \texorpdfstring{\(\Z\)}{Z} and other rings}
\label{subsec:beyond-finite-fields}

The finite-field presentation used in \cref{thm:SLn-presentation-uniform} is the \(R=\F_d\) instance of a more general Steinberg construction over coefficient rings.
For a unital ring \(R\) and \(n\ge 3\), one defines the Steinberg group \(\St_n(R)\) on generators \(x_{ij}(r)\) for \(i\neq j\) and \(r\in R\), subject to the usual Steinberg relations.
The assignment \(x_{ij}(r)\mapsto I_n+re_{j,i}\) induces a surjection \(\St_n(R)\twoheadrightarrow E_n(R)\) onto the elementary subgroup of \(\GL_n(R)\), and its kernel is central and canonically identified with \(K_2(R)\) \cite{weibel_kbook_2013}.
In the finite-field case, \(K_2(\F_d)=0\), which is exactly why the relations in \(\mathsf{R}_{\SL}(n,d)\) present \(\SL_n(\F_d)\).

Over \(R=\Z\), this kernel does not vanish.
One has \(K_2(\Z)\cong \Z/2\).
A generator of this kernel may be represented by the Steinberg symbol \(\{-1,-1\}=(x_{12}(1)x_{21}(-1)x_{12}(1))^4\), which generates \(\ker(\St_n(\Z)\to \SL_n(\Z))\) for all \(n\ge 3\) \cite{weibel_kbook_2013}.
Consequently, a Steinberg-style presentation yields \(\St_n(\Z)\); to obtain a presentation of \(\SL_n(\Z)\) one must add an additional relation killing this central element.

Finite presentations for \(\SL_3(\Z)\) and, more generally, for \(\SL_3\) over several rings of integers are worked out by Conder, Robertson and Williams \cite{conder_robertson_williams_1992}.
These results suggest a route to extending our diagrammatic axiomatisations from prime fields to integer-linear and integer-affine fragments, by retaining the local Steinberg relations and adjoining a small number of additional axioms that account for the relevant \(K_2\)-kernel.

A related extension covers determinant-\(\pm 1\) matrices.
Starting from a presentation of \(\SL_n(\Z)\), one may adjoin a generator \(r\) representing a determinant \(-1\) matrix such as \(\mathrm{diag}(-1,1,\dots,1)\), together with relations describing its conjugation action on elementary transvections, to obtain a presentation of \(\SL_n^{\pm}(\Z)=\GL_n(\Z)\).
Adding translations yields \(\Z^n\rtimes\GL_n(\Z)\); diagrammatically, this replaces the torsion relations \(t_i^d=\id\) by infinite-order translation laws, while keeping the same semidirect-product interaction between translations and linear generators as in \(\mathsf{R}_{\AGL}(n,d)\).

\end{document}